\documentclass[12pt]{amsart}
\usepackage{amsfonts}
\usepackage{changepage}
\usepackage{natbib}
\usepackage{eurosym}
\usepackage{amssymb, amsmath, url,color,  amsthm, amssymb,graphicx,multirow, appendix}
\usepackage[margin=0.8in]{geometry}
\usepackage{bigints}
\usepackage[scr=boondoxo,scrscaled=1.05]{mathalfa}
\usepackage{lscape,amssymb, amsmath,verbatim,graphicx}
\usepackage{epsfig,subfigure,float,subfigure}
\usepackage{graphicx}
\usepackage{epsfig}
\usepackage{threeparttable}
\usepackage{ragged2e}
\usepackage{epstopdf}
\usepackage{booktabs}
\usepackage{float}
\usepackage{multirow}
\usepackage{adjustbox}
\usepackage{lscape,lipsum}
\usepackage{url}
\usepackage{natbib}
\usepackage[T1]{fontenc}
\usepackage{a4wide, url}
\usepackage[english]{babel}
\usepackage{graphicx}
\usepackage{array}
\usepackage{multirow}
\usepackage{amsthm}
\usepackage{amsmath}
\usepackage[foot]{amsaddr}
\usepackage{amsfonts}
\usepackage{amssymb}
\usepackage{pgf}
\usepackage{paralist}
\usepackage{pdflscape}
\usepackage{lscape,amssymb, amsmath,verbatim,graphicx}
\usepackage{epsfig,subfigure,float,subfigure}
\usepackage{graphicx}
\usepackage{epsfig}
\usepackage{amssymb}
\usepackage{amsmath}
\usepackage{amsfonts}
\usepackage{geometry}
\usepackage{scalefnt}
\usepackage{graphicx}
\usepackage{caption}
\usepackage{setspace}
\usepackage{array}
\usepackage{multirow}
\usepackage{amsthm}
\usepackage{amsmath}
\usepackage{amsfonts}
\usepackage{amssymb}
\usepackage{pgf}
\usepackage{paralist}
\usepackage{pdflscape}
\usepackage{rotating}
\usepackage{scalefnt}
\usepackage{xr}
\usepackage{hyperref}
\usepackage{algpseudocode}
\usepackage{algorithm}

\allowdisplaybreaks
\DeclareFontFamily{OT1}{pzc}{}
\DeclareFontShape{OT1}{pzc}{m}{it}{<-> s * [1.10] pzcmi7t}{}
\DeclareMathAlphabet{\mathpzc}{OT1}{pzc}{m}{it}

\newtheoremstyle{exampstyle}
{4pt} 
{4pt} 
{\itshape} 
{} 
{\bfseries} 
{.} 
{.75em} 
{} 
\theoremstyle{exampstyle}

\numberwithin{table}{section}
\numberwithin{figure}{section}
\newtheorem{theorem}{Theorem}[section]
\newtheorem{lemma}{Lemma}[section]
\newtheorem{corollary}{Corollary}[section]

\newtheorem{assumption}{Assumption}[section]

\def\beq{\begin{equation}}
\def\eeq{\end{equation}}
\def\bals{\begin{align*}}
\def\eals{\end{align*}}
\def\bal{\begin{align}}
\def\eal{\end{align}}

\numberwithin{equation}{section}
\numberwithin{theorem}{section}
\numberwithin{corollary}{section}
\allowdisplaybreaks
\input{tcilatex}

\begin{document}
\title[Monitoring RCA models]{Real-time monitoring with RCA models}
\author{Lajos Horv\'ath$^1$}
\email{horvath@math.utah.edu}
\author{Lorenzo Trapani$^2$}
\email{lt285@leicester.ac.uk}
\address{$^1$Department of Mathematics, University of Utah, Salt Lake City,
UT 84112--0090 USA }
\address{$^2$School of Business and Economics, University of Leicester,
Leicester, U.K.; Department of Economics and Management, University of
Pavia, Pavia, Italy}

\begin{abstract}
We propose a family of \textit{weighted} statistics based on the CUSUM
process of the WLS residuals for the online detection of changepoints in a
Random Coefficient Autoregressive model, using both the standard CUSUM and
the Page-CUSUM process. We derive the asymptotics under the null of no
changepoint for all possible weighing schemes, including the case of the
standardised CUSUM, for which we derive a Darling-Erd\H{o}s-type limit
theorem; our results guarantee the procedure-wise size control under both an
open-ended and a closed-ended monitoring. In addition to considering the
standard RCA model with no covariates, we also extend our results to the
case of exogenous regressors. Our results can be applied irrespective of
(and with no prior knowledge required as to) whether the observations are
stationary or not, and irrespective of whether they change into a stationary
or nonstationary regime. Hence, our methodology is particularly suited to
detect the onset, or the collapse, of a bubble or an epidemic. Our
simulations show that our procedures, especially when standardising the
CUSUM process, can ensure very good size control and short detection delays.
We complement our theory by studying the online detection of breaks in
epidemiological and housing prices series.
\end{abstract}

\maketitle

\doublespacing

\section{Introduction\label{intro}}

Economic history is full of events where stationary observations suddenly
become explosively increasing, or - vice versa - where explosive trends
crash. A classical example are financial bubbles, where periods of
\textquotedblleft tame\textquotedblright\ fluctuations are followed by
exuberant growth, in turn then followed by\ a collapse which is typically
modelled as a stationary regime (see e.g. \citealp{harvey2017improving}; and %
\citealp{phillips2018financial}). Historical examples include the
\textquotedblleft Tulipmania\textquotedblright\ in the Netherlands in the
1630's and the \textquotedblleft South Sea\textquotedblright\ bubble of
1720. More recent instances include the Japan's real estate and stock market
bubble of the 1980's, which after the collapse turned into a negative
period, the so called \textquotedblleft Lost Decade\textquotedblright ; and,
in the US, the \textquotedblleft Dotcom\textquotedblright\ bubble of the
1990's and the housing bubble between 1996 and 2006. Similar phenomena, with
observations undergoing a change in persistence or even more radical changes
in their nature (from stationary to explosive and vice versa), are also
encountered in applied sciences. For example, in epidemiology, the onset of
a pandemic is characterised by a sudden, explosive growth in the number of
cases, followed by a return to a stationary regime when the pandemic
subsides. Hence the importance of developing tools allowing to test for
changes in persistence, or a switch between regimes.

In particular, a research area where this issue has been paid particular
attention is financial econometrics, where several methodologies have been
developed to tests (ex-ante or ex-post) for the onset or collapse of a
financial bubble. In the interest of a concise discussion, we refer the
reader to the seminal articles on ex-post detection by \citet{phillips2011},
and \citet{phillips2015testing2}, and also to \citet{skrobotov2023testing}
for a review. In addition, several contributions deal with the real-time
detection of bubble episodes - see, \textit{inter alia}, %
\citet{homm2012testing}, \citet{phillips2015testing}, and a recent
contribution by \citet{whitehouse2023real}. A common trait to this
literature is that it relies on an AutoRegressive (AR) framework as the
workhorse model. Using an AR specification has several advantages: it is a
parsimonious and well-studied set-up, and it naturally lends itself to
modelling both stationary and nonstationary regimes. %
\citet{phillips2011dating} show that an AR model with explosiveness is an
adequate representation for bubble behaviour under very general assumptions.
Moreover, an AR specification lends itself to constructing tests for bubble
behaviour using the Dickey-Fuller test and its variants- e.g., the Augmented
DF test (ADF; see \citealp{diba1988theory}), the sup-ADF test (%
\citealp{phillips2011}; \citealp{phillips2011dating}), and the generalized
sup-ADF test (\citealp{phillips2015testing2}; \citealp{phillips2015testing}%
). From a technical viewpoint, however, using an AR model is fraught with
difficulties when monitoring for changes from an explosive to a stationary
regime, although several promising solutions have been proposed such as the
reverse regression approach by \citet{phillips2018financial}. In addition to
the ability of detecting changes, ensuring a timely detection is important.
To the best of our knowledge, no optimality results exist; indeed, typical
procedures are based on \textit{unweighted} CUSUM (or related) statistics,
which are not designed to ensure optimal detection timings (see e.g. %
\citealp{aue2004delay}).

\medskip

\textit{Main contributions of this paper}

\medskip

We address the general issue of online, real-time detection of changes in
persistence and also between regimes (i.e., from stationarity to
nonstationarity and vice versa) by using a Random Coefficient Autoregressive
(RCA) model, viz. 
\begin{equation}
y_{i}=\left( \beta _{i}+\epsilon _{i,1}\right) y_{i-1}+\epsilon _{i,2},
\label{rca-1}
\end{equation}%
where $y_{0}$ denotes an initial value. The RCA model was firstly studied by %
\citet{andel} and \citet{nichollsquinn}. It belongs in the wider class of
nonlinear models for time series (see \citealp{fanyao}), which have been
proposed \textquotedblleft as a reaction against the supremacy of linear
ones - a situation inherited from strong, though often implicit, Gaussian
assumptions\textquotedblright\ (\citealp{akharif2003}). The RCA\ model is
particularly well-suited to modelling time series with potentially explosive
behaviour; for example, in the seminal paper by \citet{diba1988theory}, a
rational bubble is modelled as having an RCA dynamics (see their equation
(17)). Model (\ref{rca-1}) also nests, as a special case corresponding to
having $\beta _{i}=1$, the so-called Stochastic Unit Root model (STUR; see %
\citealp{granger1997introduction}), and it also allows for the possibility
of (conditional) heteroskedasticity in $y_{i}$; \citet{tsay1987} shows that
the widely popular ARCH\ model by \citet{engle1982} can be cast into the RCA
setup, which therefore can be viewed as a second-order equivalent. Finally,
a major advantage of (\ref{rca-1}) compared to standard AR models is that it
is possible to construct estimators of $\beta _{0}$ that are always
asymptotically normal, irrespective of whether $y_{i}$ is stationary or
nonstationary (\citealp{aue2011}). Although the literature has developed
several contributions on ex-post changepoint detection using an RCA model
(e.g. \citealp{horvath2022changepoint}), to the best of our knowledge there
are still significant limitations in the context of online detection - e.g., %
\citet{na2010}, \citet{li2015} and \citet{li2015b} all consider sequential
changepoint detection using unweighted CUSUM-based statistics, but only
under the maintained assumption of stationarity, which precludes the ability
to detect changes from stationary to exponential behaviour, and vice versa.

In this paper, we fill the gaps discussed above, by proposing a family of
statistics based on the Weighted Least Squares (WLS) estimator, and in
particular on the \textit{weighted} CUSUM\ process of the WLS residuals,
designed to ensure a faster detection than unweighted statistics. We study
the standard\ CUSUM process, and the so-called \textquotedblleft
Page-CUSUM\textquotedblright\ (\citealp{fremdt2015page}); and both
\textquotedblleft open-ended\textquotedblright\ and \textquotedblleft
closed-ended\textquotedblright\ procedures, where sequential monitoring goes
on indefinitely, or stops at a pre-specified point in time, respectively. We
make at least five contributions to the current literature on online
changepoint detection. First, we derive the limiting distribution of our
statistics under the closed-ended case; typically, in the literature, the
critical values obtained under the open-ended case are employed even for the
closed-ended case, as a (conservative) choice, ultimately leading to loss of
power. Second, we study the case of a closed-ended procedure with a
\textquotedblleft very short\textquotedblright\ monitoring horizon, adapting
the boundary function for this case, and deriving the corresponding limiting
distribution; this case has not been considered in the literature before,
but it is of pratical relevance because the researcher may prefer to carry
out monitoring for a short time, and then - in the absence of changepoints -
restart the procedure afresh. Third, in the case of the \textit{standardised}
CUSUM,\footnote{%
We define \textquotedblleft standardised CUSUM\textquotedblright\ the CUSUM\
process weighted using weights proportional to its standard deviation, as
opposed to the \textquotedblleft weighted CUSUM\textquotedblright\ which
uses lighter weights - see also below for a formal definition.} we propose
an approximation to compute critical values which offers a superior
alternative to asymptotic critical values, whose accuracy is marred by the
notoriously weak convergence to the Extreme Value Distribution; simulations
show that our approximation works extremely well even in small samples, and
for all cases considered, offering size control and short detection delays.
Fourth, we derive the limiting distribution of the \textit{weighted}
Page-CUSUM statistics, showing also that such a limit is a well-defined
random variable. Fifth, as well as studying the basic RCA\ model, we also
develop the full-blown theory (irrespective of whether $y_{i}$ is stationary
or not) for the case where there are covariates in (\ref{rca-1}); to the
best of our knowledge, this is the first contribution to deal with this
case. As a final remark, we would like to add that a major advantage of the
RCA\ set-up is that our statistics can be employed under both stationarity
and explosiveness with no modifications required; indeed, no prior knowledge
on the stationarity or not of the observations is required. Hence, our
methodology can be applied to detect changes in the persistence of a
stationary series, or changes from stationarity to a non-stationary regime,
or vice versa from a non-stationary regime to a stationary one. This has
important practical consequences, allowing for e.g. faster responses from
public health authorities in the presence of a pandemic, or from policy
makers in the presence of inflationary shocks, or a more accurate date
stamping of bubble onsets and collapses.

The remainder of the paper is organised as follows. We introduce our main
model, and the test statistics, in Section \ref{model}. The asymptotics
under the null and the alternative is reported in Section \ref{asymptotics}.
We consider the extension of the RCA model to the case of exogenous
regressors in Section \ref{covariates}. Monte Carlo evidence is in Section %
\ref{simulations}; Section \ref{empirics} contains two empirical
applications, to Covid-19 hospitalisation data and house prices. Section \ref%
{conclusions} concludes.

NOTATION. Henceforth, we use: \textquotedblleft $\rightarrow $%
\textquotedblright\ for the ordinary limit; \textquotedblleft $\overset{%
\mathcal{P}}{\rightarrow }$\textquotedblright\ for convergence in
probability; \textquotedblleft a.s.\textquotedblright\ for \textquotedblleft
almost surely\textquotedblright ; \textquotedblleft $\overset{a.s.}{%
\rightarrow }$\textquotedblright\ for almost sure convergence; and
\textquotedblleft $\left\Vert \cdot \right\Vert $\textquotedblright\ for the
Euclidean norm of vectors and matrices. Finally, $\left\{ W\left( t\right)
,0\leq t\leq 1\right\} $ denotes a standard Wiener process. Other notation
is introduced later on in the paper.

\section{Sequential monitoring of RCA\ models\label{model}}

Recall the RCA model%
\begin{equation}
y_{i}=\left( \beta _{i}+\epsilon _{i,1}\right) y_{i-1}+\epsilon _{i,2}.
\label{rca}
\end{equation}%
It is well-known (\citealp{aue2006}) that, in (\ref{rca}), the stationarity
or lack thereof of $y_{i}$ is determined by the value of $E\log \left\vert
\beta _{0}+\epsilon _{0,1}\right\vert $:

\begin{itemize}
\item[-] if $-\infty \leq E\log \left\vert \beta _{0}+\epsilon
_{0,1}\right\vert <0$, then $y_{i}$ converges exponentially fast to a
strictly stationary solution $\left\{ \overline{y}_{i},-\infty <i<\infty
\right\} $ for all initial values $y_{0}$;

\item[-] if $E\log \left\vert \beta _{0}+\epsilon _{0,1}\right\vert >0$,
then $y_{i}$ is nonstationary with $\left\vert y_{i}\right\vert \overset{a.s.%
}{\rightarrow }\infty $ exponentially fast (\citealp{berkes2009});

\item[-] if $E\log \left\vert \beta _{0}+\epsilon _{0,1}\right\vert =0$,
then $\left\vert y_{i}\right\vert \overset{\mathcal{P}}{\rightarrow }\infty $%
, but at a rate slower than exponential (see Lemma A.4 in \citealp{HT2016}).
\end{itemize}

The following assumption must hold under all three cases above.

\begin{assumption}
\label{as-1}$\{\left( \epsilon _{i,1},\epsilon _{i,2}\right) ,-\infty
<i<\infty \}$ are independent and identically distributed random variables
with \textit{(i)} $E\epsilon _{i,1}=E\epsilon _{i,2}=0$; \textit{(ii)} $%
0<E\epsilon _{i,1}^{2}=\sigma _{1}^{2}<\infty $ and $0<E\epsilon
_{i,2}^{2}=\sigma _{2}^{2}<\infty $; \textit{(iii)} $E\epsilon
_{i,1}\epsilon _{i,2}=0$; \textit{(iv)} $E|\epsilon _{i,1}|^{\nu }<\infty $
and $E|\epsilon _{i,2}|^{\nu }<\infty $ for some $\nu >2$.
\end{assumption}

We further assume that the autoregressive parameter $\beta _{i}$ is constant
over the training segment $\left\{ y_{i},1\leq i\leq m\right\} $

\begin{assumption}
\label{non-contamination}$\beta _{i}=\beta _{0}$ for $1\leq i\leq m$.
\end{assumption}

The requirement in Assumption \ref{non-contamination} is known as the 
\textit{non-contamination assumption} (see \citealp{CSW96}). We subsequently
test for the null hypothesis that, as new data come in after $m$, $\beta
_{i} $ remains constant, viz.%
\begin{equation}
H_{0}:\beta _{0}=\beta _{m+1}=\beta _{m+2}=...  \label{null}
\end{equation}%
The WLS\ estimator using the observations on the training segment $\left\{
y_{i},1\leq i\leq m\right\} $ is%
\begin{equation}
\widehat{\beta }_{m}=\left( \sum_{i=2}^{m}\frac{y_{i-1}^{2}}{1+y_{i-1}^{2}}%
\right) ^{-1}\left( \sum_{i=2}^{m}\frac{y_{i}y_{i-1}}{1+y_{i-1}^{2}}\right) .
\label{ols}
\end{equation}

\subsection{Standard weighted CUSUM-based detectors\label{standard}}

We define the CUSUM\ process of the WLS\ residuals as 
\begin{equation}
Z_{m}\left( k\right) =\left\vert \sum_{i=m+1}^{m+k}\frac{\left( y_{i}-%
\widehat{\beta }_{m}y_{i-1}\right) y_{i-1}}{1+y_{i-1}^{2}}\right\vert ,\text{
\ \ }k\geq 1.  \label{detector}
\end{equation}%
Heuristically, under the null of no change, the residuals have zero mean;
hence, the partial sum process $Z_{m}\left( k\right) $ - customarily known
as the \textit{detector} - should also fluctuate around zero with increasing
variance. Conversely, in the presence of a break (at, say, $k^{\ast }$), $%
\widehat{\beta }_{m}$ is a biased estimator for the \textquotedblleft
new\textquotedblright\ autoregressive parameter $\beta _{m+k^{\ast }+1}$;
thus, $Z_{m}\left( k\right) $ should have a drift term. Hence, a break is
flagged if $Z_{m}\left( k\right) $ exceeds a threshold. We call such a
threshold the \textit{boundary function}, and propose the following family%
\begin{equation}
g_{m,\psi }\left( k\right) =c_{\alpha ,\psi }\mathscr{s}m^{1/2}\left( 1+%
\frac{k}{m}\right) \left( \frac{k}{m+k}\right) ^{\psi }.  \label{boundary}
\end{equation}%
In (\ref{boundary}), $0\leq \psi \leq 1/2$, and $\mathscr{s}$ is defined as%
\begin{equation}
\mathscr{s}^{2}=\left\{ 
\begin{array}{l}
a_{1}\sigma _{1}^{2}+a_{2}\sigma _{2}^{2},\text{ if }-\infty \leq E\log
\left\vert \beta _{0}+\epsilon _{0,1}\right\vert <0, \\ 
\sigma _{1}^{2},\text{ if }E\log \left\vert \beta _{0}+\epsilon
_{0,1}\right\vert \geq 0,%
\end{array}%
\right.  \label{sig}
\end{equation}%
with $\sigma _{1}^{2}$ and $\sigma _{2}^{2}$ defined in Assumption \ref{as-1}%
, and%
\begin{equation*}
a_{1}=E\left( \frac{\overline{y}_{0}^{2}}{1+\overline{y}_{0}^{2}}\right) ^{2}%
\text{, \ \ and \ \ }a_{2}=E\left( \frac{\overline{y}_{0}}{1+\overline{y}%
_{0}^{2}}\right) ^{2},
\end{equation*}%
where recall that $\left\{ \overline{y}_{i},-\infty <i<\infty \right\} $ is
the stationary solution of (\ref{rca}).

On account of (\ref{detector}) and (\ref{boundary}), a changepoint is found
at a stopping time $\tau _{m,\psi }$ defined as%
\begin{equation}
\tau _{m,\psi }=%
\begin{cases}
\inf \{k\geq 1:Z_{m}\left( k\right) \geq g_{m,\psi }\left( k\right) \}, \\ 
\infty ,\ \text{if}\ Z_{m}\left( k\right) <g_{m,\psi }\left( k\right) \text{
for all }1\leq k<\infty .%
\end{cases}
\label{decision}
\end{equation}%
The constant $c_{\alpha ,\psi }$ in (\ref{boundary}) is chosen so as to
ensure that: (a) under the null, the procedure-wise probability of Type I\
Error does not exceed a user-chosen value $\alpha $, viz. $%
\lim_{m\rightarrow \infty }P\left\{ \tau _{m,\psi }=\right. $ $\left. \infty
|H_{0}\right\} =\alpha $; and (b) under the alternative, $\lim_{m\rightarrow
\infty }P\left\{ \tau _{m,\psi }<\infty |H_{A}\right\} =1$.

In (\ref{decision}), the monitoring goes on indefinitely (\textquotedblleft
open-ended\textquotedblright ). On the other hand, in some applications it
may be desirable to stop the sequential monitoring procedure after $m^{\ast
} $ period (\textquotedblleft closed-ended\textquotedblright ), e.g. in order to extend or
update the training period. In this case, we use the same detector $Z_{m}\left( k\right) $
defined in (\ref{detector}), but we modify (\ref{decision}) as%
\begin{equation}
\tau _{m,\psi }^{\ast }=%
\begin{cases}
\inf \{1\leq k\leq m^{\ast }:Z_{m}\left( k\right) \geq g_{m,\psi }^{\ast
}\left( k\right) \}, \\ 
m^{\ast },\ \text{if}\ Z_{m}\left( k\right) <g_{m,\psi }^{\ast }\left(
k\right) \text{ for all }1\leq k\leq m^{\ast },%
\end{cases}
\label{decision2}
\end{equation}%
with boundary function 
\begin{equation}
g_{m,\psi }^{\ast }\left( k\right) =c_{\alpha ,\psi }^{\ast }\mathscr{s}%
m^{1/2}\left( 1+k/m\right) \left( k/\left( m+k\right) \right) ^{\psi }.
\label{boundary2}
\end{equation}%
The boundary function $g_{m,\psi }^{\ast }\left( k\right) $ is suitable for
the case where the monitoring horizon $m^{\ast }$ is \textquotedblleft
long\textquotedblright , i.e. when it goes on for a period which is at least
proportional (or more than proportional) to the size of the training sample $%
m$. If a shorter monitoring horizon is considered, where $m^{\ast }=o\left(
m\right) $, then we use the same detector as above, $Z_{m}\left( k\right) $,
but the boundary function needs to be modified as%
\begin{equation}
\overline{g}_{m,\psi }\left( k\right) =\overline{c}_{\alpha ,\psi }%
\mathscr{s}\left( m^{\ast }\right) ^{1/2-\psi }k^{\psi }.  \label{boundary-3}
\end{equation}%
In this case, the stopping rule is defined as%
\begin{equation}
\overline{\tau }_{m,\psi }=%
\begin{cases}
\inf \{1\leq k\leq m^{\ast }:Z_{m}\left( k\right) \geq \overline{g}_{m,\psi
}\left( k\right) \}, \\ 
m^{\ast },\ \text{if}\ Z_{m}\left( k\right) <\overline{g}_{m,\psi }\left(
k\right) \text{ for all }1\leq k\leq m^{\ast }.%
\end{cases}
\label{stop-bar}
\end{equation}

\subsection{Page-CUSUM detectors\label{page-cusum}}

In a series of recent contributions, \citet{fremdt2015page}, %
\citet{kirch2022asymptotic} and \citet{kirch2022sequential} study a
different family of detectors, known as the \textquotedblleft
Page-CUSUM\textquotedblright\ processes, designed to offer a shorter
detection delay: 
\begin{equation}
Z_{m}^{\dagger }\left( k\right) =\max_{1\leq \ell \leq k}\left\vert
\sum_{i=m+\ell }^{m+k}\frac{\left( y_{i}-\widehat{\beta }_{m}y_{i-1}\right)
y_{i-1}}{1+y_{i-1}^{2}}\right\vert ,\text{ \ \ }k\geq 1.  \label{page}
\end{equation}%
Intuitively, this family of detectors searches for the \textquotedblleft
worst-case scenario\textquotedblright\ at each point in time $k$, and
therefore should guarantee faster detection in the presence of a
changepoint. Consistently with the approach studied in this paper, we
consider \textit{weighted} versions of $Z_{m}^{\dagger }\left( k\right) $;
the corresponding stopping rules are defined as:%
\begin{equation}
\tau _{m,\psi }^{\dagger }=%
\begin{cases}
\inf \{k\geq 1:Z_{m}^{\dagger }\left( k\right) \geq g_{m,\psi }\left(
k\right) \}, \\ 
\infty ,\ \text{if}\ Z_{m}^{\dagger }\left( k\right) <g_{m,\psi }\left(
k\right) \text{ for all }1\leq k<\infty ,%
\end{cases}
\label{kirch-1}
\end{equation}%
with $g_{m,\psi }\left( k\right) $ defined in (\ref{boundary}), for an
open-ended procedure (replacing the critical value $c_{\alpha ,\psi }$ with $%
c_{\alpha ,\psi }^{\dagger }$);%
\begin{equation}
\tau _{m,\psi }^{\ast \dagger }=%
\begin{cases}
\inf \{1\leq k\leq m^{\ast }:Z_{m}^{\dagger }\left( k\right) \geq g_{m,\psi
}^{\ast }\left( k\right) \}, \\ 
m^{\ast },\ \text{if}\ Z_{m}^{\dagger }\left( k\right) <g_{m,\psi }^{\ast
}\left( k\right) \text{ for all }1\leq k\leq m^{\ast },%
\end{cases}
\label{kirch-2}
\end{equation}%
with $g_{m,\psi }^{\ast }\left( k\right) $ defined in (\ref{boundary2}), for
the case of a closed-ended procedure with a long monitoring horizon
(replacing the critical value $c_{\alpha ,\psi }^{\ast }$ with $c_{\alpha
,\psi }^{\ast \dagger }$); and 
\begin{equation}
\overline{\tau }_{m,\psi }^{\dagger }=%
\begin{cases}
\inf \{1\leq k\leq m^{\ast }:Z_{m}^{\dagger }\left( k\right) \geq \overline{g%
}_{m,\psi }\left( k\right) \}, \\ 
m^{\ast },\ \text{if}\ Z_{m}^{\dagger }\left( k\right) <\overline{g}_{m,\psi
}\left( k\right) \text{ for all }1\leq k\leq m^{\ast },%
\end{cases}
\label{kirch-3}
\end{equation}%
with $\overline{g}_{m,\psi }\left( k\right) $ defined in (\ref{boundary-3}),
for the case of a closed-ended procedure with a short monitoring horizon
(replacing the critical value $\overline{c}_{\alpha ,\psi }$ with $\overline{%
c}_{\alpha ,\psi }^{\dagger }$).

\section{Asymptotics\label{asymptotics}}

We begin by listing a set of technical assumptions, which complement
Assumption \ref{as-1} and are required in the case of nonstationarity, $%
E\log \left\vert \beta _{0}+\epsilon _{0,1}\right\vert \geq 0$.

\begin{assumption}
\label{as-3}If $E\log \left\vert \beta _{0}+\epsilon _{0,1}\right\vert \geq
0 $, it holds that: \textit{(i)} $\epsilon _{0,2}$ has a bounded density; 
\textit{(ii)} $\{\epsilon _{i,1},-\infty <i<\infty \}$ and $\{\epsilon
_{i,2},-\infty <i<\infty \}$ are independent.
\end{assumption}

\begin{assumption}
\label{as-4}If $E\log \left\vert \beta _{0}+\epsilon _{0,1}\right\vert >0$,
it holds that $P\left\{ \left( \beta _{0}+\epsilon _{0,1}\right)
y_{0}+\epsilon _{0,2}=x\right\} =0$ for all $-\infty <x<\infty $.
\end{assumption}

Assumption \ref{as-3} is required specifically in order to derive the
anti-concentration bound in (\ref{ht2016-x}). Assumption \ref{as-4} ensures
that $\left\vert y_{i}\right\vert \overset{a.s.}{\rightarrow }\infty $,
ruling out that it could be identically equal to zero (see %
\citealp{berkes2009}).

\subsection{Asymptotics under the null\label{null-asy}}

We derive the weak limits of our test statistics. Results differ depending
on the choice of the weight $\psi $, on whether an open-ended or a
closed-ended monitoring scheme is used (and, in the latter case, on the
length of the monitoring horizon), and on whether the detector is
constructed using the standard CUSUM or the Page-CUSUM. In all cases, we
derive nuisance free limiting distributions, from which critical values can
be obtained by simulation for a given desired nominal significance level, $%
\alpha $.

\subsubsection{Weighted standard CUSUM-based detectors\label{std-cusum}}

We begin by studying the standard CUSUM detectors defined in Section \ref%
{standard}, starting with the open-ended case, also studied in %
\citet{lajos04} in the context of a linear regression.

\begin{theorem}
\label{th-1}We assume that Assumptions \ref{as-1} and \ref{non-contamination}
hold, and either: (i) $E\log \left\vert \beta _{0}+\epsilon
_{0,1}\right\vert <0$ holds; or (ii) $E\log \left\vert \beta _{0}+\epsilon
_{0,1}\right\vert =0$ and Assumption \ref{as-3} hold; or (iii) $E\log
\left\vert \beta _{0}+\epsilon _{0,1}\right\vert >0$ and Assumptions \ref%
{as-3}-\ref{as-4} hold. Then, under $H_{0}$ it holds that, for all $\psi
<1/2 $%
\begin{equation*}
\lim_{m\rightarrow \infty }P\left\{ \tau _{m,\psi }=\infty \right\}
=P\left\{ \sup_{0<u\leq 1}\frac{\left\vert W\left( u\right) \right\vert }{%
u^{\psi }}<c_{\alpha ,\psi }\right\} .
\end{equation*}
\end{theorem}

\medskip

We now turn to the closed-ended case. We begin by considering the case of (%
\ref{decision2}), where the monitoring goes on for a sufficient amount of
time:%
\begin{equation}
m^{\ast }=O\left( m^{\lambda }\right) \text{ for some }\lambda \geq 1,\text{
\ \ and \ \ }\lim_{m\rightarrow \infty }\frac{m^{\ast }}{m}=m_{0}\in \left(
0,\infty \right] >0.  \label{horizon}
\end{equation}%
In this case, we use the boundary function $g_{m,\psi }^{\ast }\left(
k\right) $\ defined in (\ref{boundary2}), and the corresponding stopping
rule (\ref{decision2}).\ Define $m_{\ast }=m_{0}/\left( 1+m_{0}\right) $ if $%
m_{0}<\infty $, and $m_{\ast }=1$ if $m_{0}=\infty $.

\begin{theorem}
\label{th-2}We assume that the conditions of Theorem \ref{th-1}, and (\ref%
{horizon}), are satisfied. Then, under $H_{0}$ it holds that, for all $\psi
<1/2$%
\begin{equation}
\lim_{m\rightarrow \infty }P\left\{ \tau _{m,\psi }^{\ast }=\infty \right\}
=P\left\{ \sup_{0<u\leq m_{\ast }}\frac{\left\vert W\left( u\right)
\right\vert }{u^{\psi }}<c_{\alpha ,\psi }^{\ast }\right\} .  \label{cor2}
\end{equation}
\end{theorem}

We now study the case of (\ref{stop-bar}), in a closed-ended set-up where
monitoring stops after very few steps, viz.%
\begin{equation}
m^{\ast }\rightarrow \infty ,\text{ \ \ and \ \ }\lim_{m\rightarrow \infty }%
\frac{m^{\ast }}{m}=0.  \label{horizon-short}
\end{equation}%
In this case, we use the boundary function $\overline{g}_{m,\psi }\left(
k\right) $ defined in (\ref{boundary-3}), and the corresponding stopping
rule defined in (\ref{stop-bar}).

\begin{theorem}
\label{monitor-short}We assume that the conditions of Theorem \ref{th-1} and
(\ref{horizon-short}) are satisfied. Then, under $H_{0}$ it holds that, for
all $\psi <1/2$%
\begin{equation}
\lim_{m\rightarrow \infty }P\left\{ \overline{\tau }_{m,\psi }=\infty
\right\} =P\left\{ \sup_{0<u\leq 1}\frac{\left\vert W\left( u\right)
\right\vert }{u^{\psi }}<\overline{c}_{\alpha ,\psi }\right\} .
\label{weighted-bar}
\end{equation}
\end{theorem}

We now turn to the case $\psi =1/2$, studying the closed-ended monitoring
procedure under both (\ref{horizon}) and (\ref{horizon-short}). Define%
\begin{equation}
\gamma \left( x\right) =\sqrt{2\log x}\text{ \ \ and \ \ }\delta \left(
x\right) =2\log x+\frac{1}{2}\log \log x-\frac{1}{2}\log \pi ,
\label{de-norming}
\end{equation}%
and 
\begin{equation}
c_{\alpha ,0.5}^{\ast }=\overline{c}_{\alpha ,0.5}=\frac{x+\delta \left(
\log m^{\ast }\right) }{\gamma \left( \log m^{\ast }\right) }.
\label{asy-crv-de}
\end{equation}

\begin{theorem}
\label{de}We assume that the conditions of Theorem \ref{th-1} are satisfied.
Then, under $H_{0}$ for all $-\infty <x<\infty $:

\begin{itemize}
\item[-] if (\ref{horizon}) holds, then it holds that $\lim_{m\rightarrow
\infty }P\left\{ \tau _{m,0.5}^{\ast }=m^{\ast }\right\} $ $=$ $\exp \left(
-\exp \left( -x\right) \right) $;

\item[-] if (\ref{horizon-short}) holds, then it holds that $%
\lim_{m\rightarrow \infty }P\left\{ \overline{\tau }_{m,0.5}=m^{\ast
}\right\} $ $=$ $\exp \left( -\exp \left( -x\right) \right) $.
\end{itemize}
\end{theorem}

Upon inspecting the proofs of Theorems \ref{th-1}-\ref{monitor-short} and %
\ref{de}, the limiting distributions of the weighted CUSUM with $\psi <1/2$,
and of the standardised CUSUM with $\psi =1/2$ are asymptotically
independent,\footnote{%
Intuitively, this is because $W\left( \cdot \right) $ has independent
increments, and the limiting distribution, in the case of $0\leq \psi <1/2$,
is determined by the \textquotedblleft central\textquotedblright\ values of $%
W\left( \cdot \right) $; conversely, when $\psi =1/2$, the limiting law is
determined by values at the very beginning of $W\left( \cdot \right) $.} and
therefore the two procedures, in principle, can be combined.

Theorem \ref{de} offers an explicit formula to compute asymptotic critical
values; however, these are bound to be inaccurate due to the slow
convergence to the Extreme Value distribution. In particular, simulations
show that, in finite samples, asymptotic critical values overstate the true
values thus leading to low power. A possible correction can be proposed, similarly to \citet{gombay}, as follows. Define $%
h_{m^{\ast }}=h_{m^{\ast }}\left( m^{\ast }\right) $ such that, as $m^{\ast
}\rightarrow \infty $%
\begin{equation}
h_{m^{\ast }}\rightarrow \infty \text{ \ \ and \ \ }h_{m^{\ast }}/m^{\ast
}\rightarrow 0,  \label{hm}
\end{equation}%
and let $\phi _{m}=\left( m^{\ast }+h_{m^{\ast }}\right) /\left( 2h_{m^{\ast
}}\right) $. Let also $\widehat{c}_{\alpha ,0.5}\rightarrow \infty $ denote
the solution of 
\begin{equation}
\frac{\widehat{c}_{\alpha ,0.5}\exp \left( -\frac{1}{2}\widehat{c}_{\alpha
,0.5}^{2}\right) }{\left( 2\pi \right) ^{1/2}}\left( \log \phi _{m}+\frac{%
4-\log \phi _{m}}{\widehat{c}_{\alpha ,0.5}^{2}}\right) =\alpha .
\label{vostr-cv}
\end{equation}

\begin{theorem}
\label{thgombay}We assume that the conditions of Theorem \ref{de} are
satisfied. Then, replacing $c_{\alpha ,0.5}^{\ast }$\ and $\overline{c}%
_{\alpha ,0.5}$\ with $\widehat{c}_{\alpha ,0.5}$ defined in (\ref{vostr-cv}%
), it holds that $\lim_{m\rightarrow \infty }P\left\{ \tau _{m,0.5}^{\ast
}=m^{\ast }\right\} $ $=$ $\alpha $ and $\lim_{m\rightarrow \infty }P\left\{ 
\overline{\tau }_{m,0.5}=m^{\ast }\right\} =\alpha $ respectively.
\end{theorem}

The choice of $h_{m^{\ast }}$\ is a matter of tuning; qualitatively, as $%
h_{m^{\ast }}$ increases, critical values become smaller (thus making the
procedure more conservative) and vice versa. Our simulations indicate that $%
h_{m^{\ast }}=\left( \log m^{\ast }\right) ^{1/2}$ yields the best results
in terms of size and power.

\subsubsection{Weighted Page-CUSUM detectors\label{kirch-cusum}}

We now consider the use of the Page-CUSUM detector $Z_{m}^{\dagger }\left(
k\right) $ defined in (\ref{page}), using the stopping rules defined in (\ref%
{kirch-1})-(\ref{kirch-3}). Let $\left\{ W_{1}\left( x\right) ,x\geq
0\right\} $\ and $\left\{ W_{2}\left( x\right) ,x\geq 0\right\} $\ denote
two independent standard Wiener processes.

\begin{theorem}
\label{th-kirch}We assume that the conditions of Theorem \ref{th-1} are
satisfied. Then, under $H_{0}$, for all $\psi <1/2$

\begin{itemize}
\item[-] it holds that%
\begin{equation*}
\lim_{m\rightarrow \infty }P\left\{ \tau _{m,\psi }^{\dagger }=\infty
\right\} =P\left\{ \sup_{0<x<\infty }\frac{\sup_{0\leq t\leq x}\left\vert
\left( W_{2}\left( x\right) -W_{2}\left( t\right) \right) -\left( x-t\right)
W_{1}\left( 1\right) \right\vert }{\left( 1+x\right) \left( x/\left(
1+x\right) \right) ^{\psi }}<c_{\alpha ,\psi }^{\dagger }\right\} ;
\end{equation*}

\item[-] if, in addition, (\ref{horizon}) holds, then it holds that%
\begin{equation*}
\lim_{m\rightarrow \infty }P\left\{ \tau _{m,\psi }^{\dagger \ast }=\infty
\right\} =P\left\{ \sup_{0<x\leq m_{0}}\frac{\sup_{0\leq t\leq x}\left\vert
\left( W_{2}\left( x\right) -W_{2}\left( t\right) \right) -\left( x-t\right)
W_{1}\left( 1\right) \right\vert }{\left( 1+x\right) \left( x/\left(
1+x\right) \right) ^{\psi }}<c_{\alpha ,\psi }^{\dagger \ast }\right\} ;
\end{equation*}

\item[-] if, in addition, (\ref{horizon-short}) holds, then it holds that%
\begin{equation*}
\lim_{m\rightarrow \infty }P\left\{ \overline{\tau }_{m,\psi }^{\dagger
}=\infty \right\} =P\left\{ \sup_{0<x\leq 1}\frac{\sup_{0\leq t\leq
x}\left\vert \left( W\left( x\right) -W\left( t\right) \right) \right\vert }{%
x^{\psi }}<\overline{c}_{\alpha ,\psi }^{\dagger }\right\} .
\end{equation*}
\end{itemize}
\end{theorem}

\subsubsection{LRV\ estimation\label{lrv}}

In order to apply all the results above, we require an estimate of $%
\mathscr{s}^{2}$ defined in (\ref{sig}). This can be constructed using the
data in the training sample as%
\begin{equation}
\widehat{\mathscr{s}}_{m}^{2}=\frac{1}{m}\sum_{i=2}^{m}\left( \frac{\left(
y_{i}-\widehat{\beta }_{m}y_{i-1}\right) y_{i-1}}{1+y_{i-1}^{2}}\right) ^{2}.
\label{sig-hat}
\end{equation}

\begin{corollary}
\label{sig-est}Under the conditions of Theorem \ref{th-1}, it holds that $%
\widehat{\mathscr{s}}_{m}^{2}\overset{\mathcal{P}}{\rightarrow }\mathscr{s}%
^{2}$.
\end{corollary}

We note that $\widehat{\mathscr{s}}_{m}^{2}$ does not depend on $\psi $, so
the result in Corollary \ref{sig-est} can be applied for all $0\leq \psi
\leq 1/2$, and for both open-ended and closed-ended procedures.

\subsection{Asymptotics under the alternative\label{alternative-asy}}

We consider the following alternative, where the deterministic part of the
autoregressive coefficient of (\ref{rca}) undergoes a change 
\begin{equation}
y_{i}=\left\{ 
\begin{array}{ll}
\left( \beta _{0}+\epsilon _{i,1}\right) y_{i-1}+\epsilon _{i,2} & 1\leq
i\leq m+k^{\ast }, \\ 
\left( \beta _{A}+\epsilon _{i,1}\right) y_{i-1}+\epsilon _{i,2} & 
i>m+k^{\ast },%
\end{array}%
\right.  \label{alternative}
\end{equation}%
where $\beta _{0}\neq \beta _{A}$ and $k^{\ast }$ is the time of change. In (%
\ref{alternative}), we do not put any constraints on the values of $\beta
_{0}$ and $\beta _{A}$. Hence, under (\ref{alternative}), the observations
could transition from a stationary regime to another stationary regime (if
both $E\log \left\vert \beta _{0}+\epsilon _{0,1}\right\vert <0$ and $E\log
\left\vert \beta _{A}+\epsilon _{0,1}\right\vert <0$); from a nonstationary
regime to another nonstationary regime (if both $E\log \left\vert \beta
_{0}+\epsilon _{0,1}\right\vert \geq 0$ and $E\log \left\vert \beta
_{A}+\epsilon _{0,1}\right\vert \geq 0$); or either regime can be stationary
and the other one nonstationary (which arises if $E\log \left\vert \beta
_{0}+\epsilon _{0,1}\right\vert <0$ and $E\log \left\vert \beta
_{A}+\epsilon _{0,1}\right\vert \geq 0$, implying a switch from stationarity
to nonstationarity, or if $E\log \left\vert \beta _{0}+\epsilon
_{0,1}\right\vert \geq 0$ and $E\log \left\vert \beta _{A}+\epsilon
_{0,1}\right\vert <0$, implying a switch from nonstationarity to
stationarity). We entertain the possibility that the amplitude of change may
depend on the (training) sample size $m$, thus defining%
\begin{equation}
\Delta _{m}=\beta _{A}-\beta _{0}.  \label{delta}
\end{equation}

\begin{theorem}
\label{power}We assume that (\ref{alternative}) holds with $k^{\ast
}=O\left( m\right) $. Under the conditions of Theorem \ref{th-1}, if it
holds that%
\begin{equation}
\lim_{m\rightarrow \infty }m^{1/2}\left\vert \Delta _{m}\right\vert =\infty ,
\label{power-weighted}
\end{equation}%
then it holds that $\lim_{m\rightarrow \infty }P\left\{ \tau _{m,\psi
}<\infty |H_{A}\right\} =1$, for all $\psi <1/2$. Under the conditions of
Theorem \ref{de}, if it holds that%
\begin{equation}
\lim_{m\rightarrow \infty }\frac{m^{1/2}\left\vert \Delta _{m}\right\vert }{%
\sqrt{\log \log m}}=\infty ,  \label{power-de}
\end{equation}%
then it holds that $\lim_{m\rightarrow \infty }P\left\{ \tau _{m,0.5}^{\ast
}<\infty |H_{A}\right\} =1$. The same results hold under the conditions of
Theorem \ref{th-2}; and under the conditions of Theorem \ref{monitor-short}
upon replacing $m$ with $m^{\ast }$ in (\ref{power-weighted}) and (\ref%
{power-de}). The same results also hold, for $\psi <1/2$, using the
procedures based on the Page-CUSUM detector $Z_{m}^{\dagger }\left( k\right) 
$ defined in (\ref{kirch-1})-(\ref{kirch-3}).
\end{theorem}

Condition (\ref{power-weighted}) entails that $\left\vert \Delta
_{m}\right\vert $ can drift to zero, but not too fast (i.e., power is
ensured as long as breaks are \textquotedblleft not too
small\textquotedblright ), for all $\psi <1/2$. When using $\psi =1/2$,
condition (\ref{power-de}) suggests that there is a (minor) loss of power. A
straightforward extension of the proofs in \citet{aue2004delay} and %
\citet{aue2008monitoring} would yield that the detection delay is
proportional to $m^{\left( 1-2\psi \right) /\left( 2\left( 1-\psi \right)
\right) }$ when $\psi <1/2$, and to $\log \log m$ when $\psi =1/2$, thus
indicating that heavier weights ensure faster detection.

\section{Monitoring the RCA\ model with covariates{\label{covariates}}}

In a recent contribution, \citet{astill2023using} argue in favour of adding
covariates to the basic AR specification, showing theoretically and
empirically that this results in better (and quicker) detection of bubble
episodes. Hence, we modify (\ref{rca}) as%
\begin{equation}
y_{i}=\left( \beta _{i}+\epsilon _{i,1}\right) y_{i-1}+\mathbf{\lambda }%
_{0}^{\intercal }\mathbf{x}_{i}+\epsilon _{i,2},  \label{rca-x}
\end{equation}%
where $y_{0}$ is an initial value and $\mathbf{x}_{i}\in \mathbb{R}^{p}$.
Equation (\ref{rca-x}) is, essentially, a dynamic model with exogenous
covariates, with $\mathbf{\lambda }_{0}$ constant over time.

In order to monitor for the stability of the autoregressive coefficient, we
propose again a detector based on the WLS loss function%
\begin{equation}
\mathcal{G}_{m}\left( \beta ,\mathbf{\lambda }\right) =\sum_{i=2}^{m}\frac{%
\left( y_{i}-\beta y_{i-1}-\mathbf{\lambda }^{\intercal }\mathbf{x}%
_{i}\right) ^{2}}{1+y_{i-1}^{2}}.  \label{loss-x}
\end{equation}%
The estimators of $\beta _{0}$ and $\mathbf{\lambda }_{0}$ are defined as $%
\left( \widehat{\beta }_{m},\widehat{\mathbf{\lambda }}_{m}\right) =\arg
\min_{\beta ,\mathbf{\lambda }}\mathcal{G}_{m}\left( \beta ,\mathbf{\lambda }%
\right) $, and satisfy%
\begin{equation}
\frac{\partial }{\partial \beta }\mathcal{G}_{m}\left( \widehat{\beta }_{m},%
\widehat{\mathbf{\lambda }}_{m}\right) =-2\sum_{i=2}^{m}\frac{\left( y_{i}-%
\widehat{\beta }_{m}y_{i-1}-\widehat{\mathbf{\lambda }}_{m}^{\intercal }%
\mathbf{x}_{i}\right) y_{i-1}}{1+y_{i-1}^{2}}=0,  \label{foc-x}
\end{equation}%
which suggests the following detector%
\begin{equation}
Z_{m}^{X}\left( k\right) =\left\vert \sum_{i=m+1}^{m+k}\frac{\left( y_{i}-%
\widehat{\beta }_{m}y_{i-1}-\widehat{\mathbf{\lambda }}_{m}^{\intercal }%
\mathbf{x}_{i}\right) y_{i-1}}{1+y_{i-1}^{2}}\right\vert .
\label{detector-x}
\end{equation}

The following assumptions complement Assumptions \ref{as-1}-\ref{as-4}.

\begin{assumption}
\label{as-x-1}\textit{(i)} $E\left\Vert \mathbf{x}_{i}\right\Vert ^{\kappa
_{1}}<\infty $ for some $\kappa _{1}>4$; \textit{(ii)} $\mathbf{x}_{i}=%
\mathbf{g}\left( \eta _{i},\eta _{i-1},...\right) $, where $\mathbf{g}:%
\mathcal{S}^{\infty }\rightarrow \mathbb{R}^{p}$ is a non-random, measurable
function and $\left\{ \eta _{i},-\infty <i<\infty \right\} $ are \textit{%
i.i.d.} random variables with values in the measurable space $\mathcal{S}$
and $\left( E\left\Vert \mathbf{x}_{i}-\mathbf{x}_{i,j}^{\ast }\right\Vert
^{\kappa _{1}}\right) ^{1/\kappa _{1}}\leq c_{0}j^{-\kappa _{2}}$, with some 
$c_{0}>0$ and $\kappa _{2}>2$, where $\mathbf{x}_{i,j}^{\ast }=\mathbf{g}%
\left( \eta _{i},...,\eta _{i-j+1},\eta _{i-j,i,j}^{\ast },\eta
_{i-j-1,i,j}^{\ast }...\right) $ where $\left\{ \eta _{\ell ,i,j}^{\ast
},-\infty <\ell ,i,j<\infty \right\} $ are \textit{i.i.d.} random copies of $%
\eta _{0}$, independent of $\left\{ \eta _{i},-\infty <i<\infty \right\} $.
\end{assumption}

\begin{assumption}
\label{as-x-2}$\left\{ \left( \epsilon _{i,1},\epsilon _{i,2}\right)
,-\infty <i<\infty \right\} $ and $\left\{ \eta _{i},-\infty <i<\infty
\right\} $ are independent.
\end{assumption}

\begin{assumption}
\label{as-x-3}If $E\log \left\vert \beta _{0}+\epsilon _{0,1}\right\vert >0$%
, it holds that $P\left\{ \left( \beta _{0}+\epsilon _{0,1}\right) y_{0}+%
\mathbf{\lambda }^{\intercal }\mathbf{x}_{0}+\epsilon _{0,2}=x\right\} =0$
for all $-\infty<x<\infty$.
\end{assumption}

Assumption \ref{as-x-1} states that the regressors $\mathbf{x}_{i}$ form a
decomposable Bernoulli shift - i.e., a weakly dependent, stationary process
which can be well-approximated by an $m$-dependent sequence. The concepts of
Bernoulli shift and decomposability appeared first in %
\citet{ibragimov1962some}; see also \citet{wu2005} and \citet{berkeshormann}%
. Bernoulli shifts have proven a convenient way to model dependent time
series, mainly due to their generality and to the fact that they are much
easier to verify than e.g. mixing conditions: \citet{aue09} and %
\citet{linliu}, \textit{inter alia}, provide numerous examples of such DGPs,
which include ARMA models, ARCH/GARCH sequences, and other nonlinear time
series models (such as e.g. random coefficient autoregressive models and
threshold models). Indeed, under stationarity, $y_{i}$ itself can be
approximated by a decomposable Bernoulli shift (%
\citealp{horvath2022changepoint}). Assumption \ref{as-x-2} states that the
exogenous variables $\mathbf{x}_{i}$ are independent of the innovations $%
\epsilon _{i,1}$ and $\epsilon _{i,2}$, and Assumption \ref{as-x-3} is
similar to Assumption \ref{as-4}, ensuring $\left\vert y_{i}\right\vert 
\overset{a.s.}{\rightarrow }\infty $.

Finally, we note that, in the presence of covariates, we need to exclude the
boundary case $E\log \left\vert \beta _{0}+\epsilon _{0,1}\right\vert =0$;
this is because we would need an exact (and large enough) rate of divergence
for $\left\vert y_{i}\right\vert $ as $i\rightarrow \infty $, but this
result is not available in the case $E\log |\beta _{0}+\epsilon _{0,1}|=0$
(see also Theorem 4 in \citealp{HT2019}, and the discussion thereafter). 
\newline
The boundary function is defined as%
\begin{equation}
g_{m,\psi }^{\left( x\right) }\left( k\right) =c_{\alpha ,\psi }^{\left(
x\right) }\mathscr{s}_{x}^{2}m^{1/2}\left( 1+\frac{k}{\mathscr{s}_{x,d}^{2}m}%
\right) \left( \frac{k}{\mathscr{s}_{x,d}^{2}m+k}\right) ^{\psi },
\label{boundary-x}
\end{equation}%
where $c_{\alpha ,\psi }^{\left( x\right) }$ is a critical value, and 
\begin{equation}
\mathscr{s}_{x}^{2}=\left\{ 
\begin{array}{l}
\mathscr{s}_{x,2}^{2}/\mathscr{s}_{x,1},\text{ if }-\infty \leq E\log
\left\vert \beta _{0}+\epsilon _{0,1}\right\vert <0, \\ 
\sigma _{1},\text{ if }E\log \left\vert \beta _{0}+\epsilon
_{0,1}\right\vert >0,%
\end{array}%
\right.  \label{var-x}
\end{equation}%
\begin{equation*}
\mathscr{s}_{x,d}^{2}=\left\{ 
\begin{array}{l}
\mathscr{s}_{x,2}^{2}/\mathscr{s}_{x,1}^{2},\text{ if }-\infty \leq E\log
\left\vert \beta _{0}+\epsilon _{0,1}\right\vert <0, \\ 
1,\text{ if }E\log \left\vert \beta _{0}+\epsilon _{0,1}\right\vert >0,%
\end{array}%
\right.
\end{equation*}%
with 
\begin{equation}
\mathscr{s}_{x,1}^{2}=\mathbf{a}^{\intercal }\mathbf{QCQa,}\text{ \ \ and \
\ }\mathscr{s}_{x,2}^{2}=\sigma _{1}^{2}E\left( \frac{\overline{y}_{0}^{2}}{%
1+\overline{y}_{0}^{2}}\right) ^{2}+\sigma _{2}^{2}E\left( \frac{\overline{y}%
_{0}}{1+\overline{y}_{0}^{2}}\right) ^{2},  \label{sig-x-2}
\end{equation}%
where $\mathbf{a}$, $\mathbf{Q}$ and $\mathbf{C}$ are defined in (\ref{q})-(%
\ref{a}) in the Supplement. The stopping rule is%
\begin{equation}
\tau _{m,\psi }^{\left( x\right) }=%
\begin{cases}
\inf \{k\geq 1:Z_{m}^{X}\left( k\right) \geq g_{m,\psi }^{\left( x\right)
}\left( k\right) \}, \\ 
\infty ,\ \text{if}\ Z_{m}^{X}\left( k\right) <g_{m,\psi }^{\left( x\right)
}\left( k\right) \text{ for all }1\leq k<\infty ,%
\end{cases}
\label{decision-x-1}
\end{equation}%
and%
\begin{equation}
\tau _{m,\psi }^{\ast \left( x\right) }=%
\begin{cases}
\inf \{k\geq 1:Z_{m}^{X}\left( k\right) \geq g_{m,\psi }^{\left( x\right)
}\left( k\right) \}, \\ 
m^{\ast },\ \text{if}\ Z_{m}^{X}\left( k\right) <g_{m,\psi }^{\left(
x\right) }\left( k\right) \text{ for all }1\leq k\leq m^{\ast },%
\end{cases}
\label{decision-x-2}
\end{equation}%
for an open-ended and a closed-ended monitoring procedure respectively.

\begin{theorem}
\label{th-x-1}We assume that Assumptions \ref{as-1}, \ref{as-x-1}, and \ref%
{as-x-2} are satisfied, and either (i) $E\log \left\vert \beta _{0}+\epsilon
_{0,1}\right\vert <0$, or (ii) $E\log \left\vert \beta _{0}+\epsilon
_{0,1}\right\vert >0$ and Assumption \ref{as-x-3} hold. Then, under $H_{0}$,
the results of Theorems \ref{th-1}, \ref{th-2}\ and \ref{monitor-short} hold
for $\tau _{m,\psi }^{\left( x\right) }$, $\tau _{m,\psi }^{\ast \left(
x\right) }$ and $\overline{\tau }_{m,\psi }^{\left( x\right) }$\
respectively.
\end{theorem}

\begin{theorem}
\label{th-x-2}We assume that the conditions of Theorem \ref{th-x-1} are
satisfied. Then, for $\psi =1/2$, under $H_{0}$, the same results as in
Theorem \ref{de} hold.
\end{theorem}

Along the same lines as in Section \ref{kirch-cusum}, it is possible to
construct weighted monitoring schemes based on the detector%
\begin{equation}
Z_{m}^{\dagger \left( X\right) }\left( k\right) =\max_{1\leq \ell \leq
k}\left\vert \sum_{i=m+\ell }^{m+k}\frac{\left( y_{i}-\widehat{\beta }%
_{m}y_{i-1}-\widehat{\mathbf{\lambda }}_{m}^{\intercal }\mathbf{x}%
_{i}\right) y_{i-1}}{1+y_{i-1}^{2}}\right\vert ,\text{ \ \ }k\geq 1.
\label{kirch-x}
\end{equation}

\begin{theorem}
\label{kirch-covar}We assume that Assumptions \ref{as-1}, \ref{as-x-1}, and %
\ref{as-x-2}, and (\ref{horizon}), are satisfied, and either (i) $E\log
\left\vert \beta _{0}+\epsilon _{0,1}\right\vert <0$, or (ii) $E\log
\left\vert \beta _{0}+\epsilon _{0,1}\right\vert >0$ and Assumption \ref%
{as-x-3} hold. Then, under $H_{0}$, the results of Theorem \ref{th-kirch}
hold.
\end{theorem}

In all the results above, the limiting behaviour of the stopping time is the
same as in the absence of covariates; however, this does not mean that these
do not play a role, since they build into the recursion that defines $y_{i}$%
. As far as Theorem \ref{th-x-2} is concerned, the same approximation for
critical values as in (\ref{vostr-cv}) can be used.

Under the alternative that the deterministic part of the autoregressive root
changes, 
\begin{equation}
y_{i}=\left\{ 
\begin{array}{ll}
\left( \beta _{0}+\epsilon _{i,1}\right) y_{i-1}+\mathbf{\lambda }%
_{0}^{\intercal }\mathbf{x}_{i}+\epsilon _{i,2} & 1\leq i\leq m+k^{\ast },
\\ 
\left( \beta _{A}+\epsilon _{i,1}\right) y_{i-1}+\mathbf{\lambda }%
_{0}^{\intercal }\mathbf{x}_{i}+\epsilon _{i,2} & i>m+k^{\ast },%
\end{array}%
\right.  \label{alt-cov}
\end{equation}%
the same results as in Section \ref{alternative-asy} hold, as summarised in
the following theorem.

\begin{theorem}
\label{alt-covariates}We assume that the conditions of Theorem \ref{th-x-1}
are satisfied. Then, under (\ref{alt-cov}), the same results as in Theorem %
\ref{power} hold.
\end{theorem}

\section{Simulations\label{simulations}}

We provide some Monte Carlo evidence and guidelines on implementation;
further details and results are reported in Section \ref{furtherMC} in the
Supplement. We consider the following Data Generating Process, based on (\ref%
{rca-x})%
\begin{equation}
y_{i}=\left( \beta _{0}+\epsilon _{i,1}\right) y_{i-1}+\lambda
_{0}x_{i}+\epsilon _{i,2},  \label{dgp}
\end{equation}%
for $1\leq i\leq m+1,000$, where we simulate $\epsilon _{i,1}$ and $\epsilon
_{i,2}$ as independent of one another and \textit{i.i.d.} with distributions 
$N\left( 0,\sigma _{1}^{2}\right) $ and $N\left( 0,\sigma _{2}^{2}\right) $
respectively, discarding the first $1,000$ observations to avoid dependence
on initial conditions. We used $\sigma _{1}^{2}=0.01$ in all experiments,
and we have considered three cases: in \textbf{Case I}, we set $\beta
_{0}=0.5$, with $E\log \left\vert \beta _{0}+\epsilon _{i,1}\right\vert
=-0.717$, corresponding to a stationary regime; in \textbf{Case II}, we set $%
\beta _{0}=1.05$, with $E\log \left\vert \beta _{0}+\epsilon
_{i,1}\right\vert =0.044$, corresponding to a mildly explosive regime; and
in \textbf{Case III}, we set $\beta _{0}=1$, with $E\log \left\vert \beta
_{0}+\epsilon _{i,1}\right\vert =-0.007$, indicating a stationary, but
\textquotedblleft nearly non-stationary\textquotedblright\ process (this
case corresponds to the STUR model). The variance of the idiosyncratic shock
is $\sigma _{2}^{2}=0.5$ in Case I, and $\sigma _{2}^{2}=0.1$ in Cases II
and III; in unreported simulations, using different
values does not result in any significant changes, save for the (expected)
fact that tests have better properties (in terms of size and power) for
smaller values of $\sigma _{2}^{2}$. When covariates are used, we set $%
\lambda _{0}=1$ and generate $x_{i}$ as \textit{i.i.d.}$N\left( 0,1\right) $%
. Critical values (for a nominal level equal to $5\%$) are computed using
the results in Theorem \ref{th-2} when using the weighted CUSUM with $\psi
<1/2$, and Theorem \ref{th-kirch}\textit{(ii)} when using Page-CUSUM
statistics. When using the weighted CUSUM with $\psi =1/2$, we use both the
asymptotic critical values $c_{\alpha ,0.5}$\ defined in (\ref{asy-crv-de}),
and the approximation $\widehat{c}_{\alpha ,0.5}$\ defined as the solution
of (\ref{vostr-cv}), with $h_{m^{\ast }}=\left( \log m^{\ast }\right) ^{1/2}$%
. Results are based on $1,000$ replications.

\medskip

Empirical rejection frequencies are in Tables \ref{tab:ERF1}-\ref{tab:ERF3},
where we consider the case where, in (\ref{dgp}), $\lambda _{0}=0$ - i.e., a
\textquotedblleft pure\textquotedblright\ RCA without covariates. In the
case of stationarity (Table \ref{tab:ERF1}), the procedure-wise probability
of Type I Errors is always controlled when using $\psi =1/2$; the asymptotic
critical values $c_{\alpha ,0.5}$\ defined in (\ref{asy-crv-de}) lead to
under-rejection, as can be expected in light of the slow convergence to the
asymptotic distribution; conversely, their approximation using $\widehat{c}%
_{\alpha ,0.5}$\ defined as the solution of (\ref{vostr-cv}) achieves size
control in all cases considered. When using $\psi <1/2$, both in the case of
the CUSUM and the Page-CUSUM, the procedure-wise probability of Type I
Errors is also controlled, but larger sample sizes $m$ are required. In the
explosive case (Table \ref{tab:ERF2}), our procedures tend to under-reject
whenever using $\psi =1/2$ (both with $c_{\alpha ,0.5}$\ and $\widehat{c}%
_{\alpha ,0.5}$), although this seems to slowly improve as $m$ increases,
for all lengths of the monitoring horizon $m^{\ast }$; results are less
conservative when using $\psi <1/2$, again both in the case of the CUSUM and
the Page-CUSUM. On the other hand, in the STUR case (Table \ref{tab:ERF3}),
empirical rejection frequencies tend to be higher than in the other cases.
The procedure-wise probability of Type I Errors is controlled in all cases
when using $\psi =1/2$ and $c_{\alpha ,0.5}$, whereas using $\widehat{c}%
_{\alpha ,0.5}$\ requires either large $m$ ($\geq 200$), or not overly long
monitoring horizons when $m<200$ (in those cases, using $m^{\ast }\leq m$
always results in size control). When $\psi <1/2$, size control is more
problematic and the procedure tends to be oversized, unless $m^{\ast }$ is
\textquotedblleft small\textquotedblright\ compared to $m$. Whilst results
are broadly similar for the CUSUM and the Page-CUSUM, we note that the
latter yields marginally improved size control over the former.

The main message of Tables \ref{tab:ERF1}-\ref{tab:ERF3} is that, broadly
speaking, using the CUSUM with $\psi =1/2$ seems the preferred solution
across all cases considered. The choice between $c_{\alpha ,0.5}$\ and $%
\widehat{c}_{\alpha ,0.5}$ essentially depends on $m$ (and $m^{\ast }$), but
in the vast majority of the cases considered, using $\widehat{c}_{\alpha
,0.5}$ offers excellent size control without having an overly conservative
procedure. Our results suggest avoiding overly long monitoring horizons when 
$m\leq 100$. In these cases, monitoring can be carried out firstly over a
short horizon and, if no changepoints are detected, the procedure can be
started afresh including the previous monitoring horizon within the training
sample. In Tables \ref{tab:ERF1c}-\ref{tab:ERF3c} in the Supplement, we
report empirical rejection frequencies in the presence of covariates in (\ref%
{dgp}). Results are broadly similar, but empirical rejection frequencies
increase in all cases considered. In the stationary case (Table \ref%
{tab:ERF1c}), results are essentially the same, and, in the explosive case
(Table \ref{tab:ERF2c}), the increase in empirical rejection frequencies
leads to improvements, making the procedure less conservative. The latter
result is interesting, and it suggests that sequential monitoring during a
bubble phase may benefit from including control variables in the price
dynamics. On the other hand, results worsen in the STUR case (Table \ref%
{tab:ERF3c}), and the procedures tend to (sometimes massively) over-reject.
The best results are obtained with $\psi =1/2$ and $c_{\alpha ,0.5}$, which
achieves size control at least when $m>50$ (or when $m\leq 50$, but the
monitoring horizon is not too long, i.e. $m^{\ast }\leq m$). This suggests
that, when covariates are included in the basic RCA specification, $\psi
=1/2 $ and $c_{\alpha ,0.5}$ should be employed - especially when trying to detect the inception of a bubble from a near
stationary regime.

We would like to emphasize that all the cases considered above are fully
under the researcher's control: whether to include covariates or not, and
the length of the monitoring horizon $m^{\ast }$, can be decided \textit{a
priori}; as mentioned above, preliminary analysis on $y_{i}$ during the
training period offers further help.

\medskip

\begin{table*}[h!]
\caption{{\protect\footnotesize {Empirical rejection frequencies under the
null of no changepoint and no covariates - Case I, $\protect\beta_0=0.5$}}}
\label{tab:ERF1}\centering
\par
{\scriptsize 
\begin{tabular}{lllllllllllllllllllll}
\hline\hline
&  &  &  &  &  &  &  &  &  &  &  &  &  &  &  &  &  &  &  &  \\ 
&  &  &  &  &  & \multicolumn{5}{c}{Weighted CUSUM} & \multicolumn{1}{c}{} & 
\multicolumn{3}{c}{Standardised CUSUM} & \multicolumn{1}{c}{} & 
\multicolumn{5}{c}{Weighted Page-CUSUM} \\ 
&  &  &  & $\psi $ &  & \multicolumn{1}{c}{$0$} & \multicolumn{1}{c}{} & 
\multicolumn{1}{c}{$0.25$} & \multicolumn{1}{c}{} & \multicolumn{1}{c}{$0.45$%
} & \multicolumn{1}{c}{} & \multicolumn{1}{c}{} & \multicolumn{1}{c}{$0.5$}
& \multicolumn{1}{c}{} & \multicolumn{1}{c}{} & \multicolumn{1}{c}{$0$} & 
\multicolumn{1}{c}{} & \multicolumn{1}{c}{$0.25$} & \multicolumn{1}{c}{} & 
\multicolumn{1}{c}{$0.45$} \\ 
&  &  &  &  &  & \multicolumn{1}{c}{} & \multicolumn{1}{c}{} & 
\multicolumn{1}{c}{} & \multicolumn{1}{c}{} & \multicolumn{1}{c}{} & 
\multicolumn{1}{c}{} & \multicolumn{1}{c}{$c_{\alpha ,0.5}$} & 
\multicolumn{1}{c}{} & \multicolumn{1}{c}{$\widehat{c}_{\alpha ,0.5}$} & 
\multicolumn{1}{c}{} & \multicolumn{1}{c}{} & \multicolumn{1}{c}{} & 
\multicolumn{1}{c}{} & \multicolumn{1}{c}{} & \multicolumn{1}{c}{} \\ 
&  &  &  &  &  & \multicolumn{1}{c}{} & \multicolumn{1}{c}{} & 
\multicolumn{1}{c}{} & \multicolumn{1}{c}{} & \multicolumn{1}{c}{} & 
\multicolumn{1}{c}{} & \multicolumn{1}{c}{} & \multicolumn{1}{c}{} & 
\multicolumn{1}{c}{} & \multicolumn{1}{c}{} & \multicolumn{1}{c}{} & 
\multicolumn{1}{c}{} & \multicolumn{1}{c}{} & \multicolumn{1}{c}{} & 
\multicolumn{1}{c}{} \\ 
\multicolumn{1}{c}{} & \multicolumn{1}{c}{$m$} & \multicolumn{1}{c}{} & 
\multicolumn{1}{c}{$m^{\ast }$} & \multicolumn{1}{c}{} & \multicolumn{1}{c}{}
& \multicolumn{1}{c}{} & \multicolumn{1}{c}{} & \multicolumn{1}{c}{} & 
\multicolumn{1}{c}{} & \multicolumn{1}{c}{} & \multicolumn{1}{c}{} & 
\multicolumn{1}{c}{} & \multicolumn{1}{c}{} & \multicolumn{1}{c}{} & 
\multicolumn{1}{c}{} & \multicolumn{1}{c}{} & \multicolumn{1}{c}{} & 
\multicolumn{1}{c}{} & \multicolumn{1}{c}{} & \multicolumn{1}{c}{} \\ 
\multicolumn{1}{c}{} & \multicolumn{1}{c}{} & \multicolumn{1}{c}{} & 
\multicolumn{1}{c}{} & \multicolumn{1}{c}{} & \multicolumn{1}{c}{} & 
\multicolumn{1}{c}{} & \multicolumn{1}{c}{} & \multicolumn{1}{c}{} & 
\multicolumn{1}{c}{} & \multicolumn{1}{c}{} & \multicolumn{1}{c}{} & 
\multicolumn{1}{c}{} & \multicolumn{1}{c}{} & \multicolumn{1}{c}{} & 
\multicolumn{1}{c}{} & \multicolumn{1}{c}{} & \multicolumn{1}{c}{} & 
\multicolumn{1}{c}{} & \multicolumn{1}{c}{} & \multicolumn{1}{c}{} \\ 
\multicolumn{1}{c}{} & \multicolumn{1}{c}{} & \multicolumn{1}{c}{} & 
\multicolumn{1}{c}{$25$} & \multicolumn{1}{c}{} & \multicolumn{1}{c}{} & 
\multicolumn{1}{c}{$0.047$} & \multicolumn{1}{c}{} & \multicolumn{1}{c}{$%
0.057$} & \multicolumn{1}{c}{} & \multicolumn{1}{c}{$0.045$} & 
\multicolumn{1}{c}{} & \multicolumn{1}{c}{$0.023$} & \multicolumn{1}{c}{} & 
\multicolumn{1}{c}{$0.045$} & \multicolumn{1}{c}{} & \multicolumn{1}{c}{$%
0.045$} & \multicolumn{1}{c}{} & \multicolumn{1}{c}{$0.057$} & 
\multicolumn{1}{c}{} & \multicolumn{1}{c}{$0.054$} \\ 
\multicolumn{1}{c}{} & \multicolumn{1}{c}{$50$} & \multicolumn{1}{c}{} & 
\multicolumn{1}{c}{$50$} & \multicolumn{1}{c}{} & \multicolumn{1}{c}{} & 
\multicolumn{1}{c}{$0.068$} & \multicolumn{1}{c}{} & \multicolumn{1}{c}{$%
0.081$} & \multicolumn{1}{c}{} & \multicolumn{1}{c}{$0.066$} & 
\multicolumn{1}{c}{} & \multicolumn{1}{c}{$0.034$} & \multicolumn{1}{c}{} & 
\multicolumn{1}{c}{$0.055$} & \multicolumn{1}{c}{} & \multicolumn{1}{c}{$%
0.060$} & \multicolumn{1}{c}{} & \multicolumn{1}{c}{$0.080$} & 
\multicolumn{1}{c}{} & \multicolumn{1}{c}{$0.069$} \\ 
\multicolumn{1}{c}{} & \multicolumn{1}{c}{} & \multicolumn{1}{c}{} & 
\multicolumn{1}{c}{$100$} & \multicolumn{1}{c}{} & \multicolumn{1}{c}{} & 
\multicolumn{1}{c}{$0.064$} & \multicolumn{1}{c}{} & \multicolumn{1}{c}{$%
0.079$} & \multicolumn{1}{c}{} & \multicolumn{1}{c}{$0.070$} & 
\multicolumn{1}{c}{} & \multicolumn{1}{c}{$0.036$} & \multicolumn{1}{c}{} & 
\multicolumn{1}{c}{$0.053$} & \multicolumn{1}{c}{} & \multicolumn{1}{c}{$%
0.058$} & \multicolumn{1}{c}{} & \multicolumn{1}{c}{$0.052$} & 
\multicolumn{1}{c}{} & \multicolumn{1}{c}{$0.061$} \\ 
\multicolumn{1}{c}{} & \multicolumn{1}{c}{} & \multicolumn{1}{c}{} & 
\multicolumn{1}{c}{$200$} & \multicolumn{1}{c}{} & \multicolumn{1}{c}{} & 
\multicolumn{1}{c}{$0.092$} & \multicolumn{1}{c}{} & \multicolumn{1}{c}{$%
0.107$} & \multicolumn{1}{c}{} & \multicolumn{1}{c}{$0.086$} & 
\multicolumn{1}{c}{} & \multicolumn{1}{c}{$0.038$} & \multicolumn{1}{c}{} & 
\multicolumn{1}{c}{$0.057$} & \multicolumn{1}{c}{} & \multicolumn{1}{c}{$%
0.087$} & \multicolumn{1}{c}{} & \multicolumn{1}{c}{$0.089$} & 
\multicolumn{1}{c}{} & \multicolumn{1}{c}{$0.083$} \\ 
\multicolumn{1}{c}{} & \multicolumn{1}{c}{} & \multicolumn{1}{c}{} & 
\multicolumn{1}{c}{} & \multicolumn{1}{c}{} & \multicolumn{1}{c}{} & 
\multicolumn{1}{c}{} & \multicolumn{1}{c}{} & \multicolumn{1}{c}{} & 
\multicolumn{1}{c}{} & \multicolumn{1}{c}{} & \multicolumn{1}{c}{} & 
\multicolumn{1}{c}{} & \multicolumn{1}{c}{} & \multicolumn{1}{c}{} & 
\multicolumn{1}{c}{} & \multicolumn{1}{c}{} & \multicolumn{1}{c}{} & 
\multicolumn{1}{c}{} & \multicolumn{1}{c}{} & \multicolumn{1}{c}{} \\ 
\multicolumn{1}{c}{} & \multicolumn{1}{c}{} & \multicolumn{1}{c}{} & 
\multicolumn{1}{c}{} & \multicolumn{1}{c}{} & \multicolumn{1}{c}{} & 
\multicolumn{1}{c}{} & \multicolumn{1}{c}{} & \multicolumn{1}{c}{} & 
\multicolumn{1}{c}{} & \multicolumn{1}{c}{} & \multicolumn{1}{c}{} & 
\multicolumn{1}{c}{} & \multicolumn{1}{c}{} & \multicolumn{1}{c}{} & 
\multicolumn{1}{c}{} & \multicolumn{1}{c}{} & \multicolumn{1}{c}{} & 
\multicolumn{1}{c}{} & \multicolumn{1}{c}{} & \multicolumn{1}{c}{} \\ 
\multicolumn{1}{c}{} & \multicolumn{1}{c}{} & \multicolumn{1}{c}{} & 
\multicolumn{1}{c}{$50$} & \multicolumn{1}{c}{} & \multicolumn{1}{c}{} & 
\multicolumn{1}{c}{$0.063$} & \multicolumn{1}{c}{} & \multicolumn{1}{c}{$%
0.068$} & \multicolumn{1}{c}{} & \multicolumn{1}{c}{$0.055$} & 
\multicolumn{1}{c}{} & \multicolumn{1}{c}{$0.025$} & \multicolumn{1}{c}{} & 
\multicolumn{1}{c}{$0.048$} & \multicolumn{1}{c}{} & \multicolumn{1}{c}{$%
0.061$} & \multicolumn{1}{c}{} & \multicolumn{1}{c}{$0.066$} & 
\multicolumn{1}{c}{} & \multicolumn{1}{c}{$0.054$} \\ 
\multicolumn{1}{c}{} & \multicolumn{1}{c}{$100$} & \multicolumn{1}{c}{} & 
\multicolumn{1}{c}{$100$} & \multicolumn{1}{c}{} & \multicolumn{1}{c}{} & 
\multicolumn{1}{c}{$0.057$} & \multicolumn{1}{c}{} & \multicolumn{1}{c}{$%
0.067$} & \multicolumn{1}{c}{} & \multicolumn{1}{c}{$0.075$} & 
\multicolumn{1}{c}{} & \multicolumn{1}{c}{$0.027$} & \multicolumn{1}{c}{} & 
\multicolumn{1}{c}{$0.056$} & \multicolumn{1}{c}{} & \multicolumn{1}{c}{$%
0.059$} & \multicolumn{1}{c}{} & \multicolumn{1}{c}{$0.066$} & 
\multicolumn{1}{c}{} & \multicolumn{1}{c}{$0.063$} \\ 
\multicolumn{1}{c}{} & \multicolumn{1}{c}{} & \multicolumn{1}{c}{} & 
\multicolumn{1}{c}{$200$} & \multicolumn{1}{c}{} & \multicolumn{1}{c}{} & 
\multicolumn{1}{c}{$0.062$} & \multicolumn{1}{c}{} & \multicolumn{1}{c}{$%
0.062$} & \multicolumn{1}{c}{} & \multicolumn{1}{c}{$0.066$} & 
\multicolumn{1}{c}{} & \multicolumn{1}{c}{$0.029$} & \multicolumn{1}{c}{} & 
\multicolumn{1}{c}{$0.048$} & \multicolumn{1}{c}{} & \multicolumn{1}{c}{$%
0.059$} & \multicolumn{1}{c}{} & \multicolumn{1}{c}{$0.055$} & 
\multicolumn{1}{c}{} & \multicolumn{1}{c}{$0.060$} \\ 
\multicolumn{1}{c}{} & \multicolumn{1}{c}{} & \multicolumn{1}{c}{} & 
\multicolumn{1}{c}{$400$} & \multicolumn{1}{c}{} & \multicolumn{1}{c}{} & 
\multicolumn{1}{c}{$0.062$} & \multicolumn{1}{c}{} & \multicolumn{1}{c}{$%
0.061$} & \multicolumn{1}{c}{} & \multicolumn{1}{c}{$0.054$} & 
\multicolumn{1}{c}{} & \multicolumn{1}{c}{$0.028$} & \multicolumn{1}{c}{} & 
\multicolumn{1}{c}{$0.050$} & \multicolumn{1}{c}{} & \multicolumn{1}{c}{$%
0.062$} & \multicolumn{1}{c}{} & \multicolumn{1}{c}{$0.064$} & 
\multicolumn{1}{c}{} & \multicolumn{1}{c}{$0.061$} \\ 
\multicolumn{1}{c}{} & \multicolumn{1}{c}{} & \multicolumn{1}{c}{} & 
\multicolumn{1}{c}{} & \multicolumn{1}{c}{} & \multicolumn{1}{c}{} & 
\multicolumn{1}{c}{} & \multicolumn{1}{c}{} & \multicolumn{1}{c}{} & 
\multicolumn{1}{c}{} & \multicolumn{1}{c}{} & \multicolumn{1}{c}{} & 
\multicolumn{1}{c}{} & \multicolumn{1}{c}{} & \multicolumn{1}{c}{} & 
\multicolumn{1}{c}{} & \multicolumn{1}{c}{} & \multicolumn{1}{c}{} & 
\multicolumn{1}{c}{} & \multicolumn{1}{c}{} & \multicolumn{1}{c}{} \\ 
\multicolumn{1}{c}{} & \multicolumn{1}{c}{} & \multicolumn{1}{c}{} & 
\multicolumn{1}{c}{} & \multicolumn{1}{c}{} & \multicolumn{1}{c}{} & 
\multicolumn{1}{c}{} & \multicolumn{1}{c}{} & \multicolumn{1}{c}{} & 
\multicolumn{1}{c}{} & \multicolumn{1}{c}{} & \multicolumn{1}{c}{} & 
\multicolumn{1}{c}{} & \multicolumn{1}{c}{} & \multicolumn{1}{c}{} & 
\multicolumn{1}{c}{} & \multicolumn{1}{c}{} & \multicolumn{1}{c}{} & 
\multicolumn{1}{c}{} & \multicolumn{1}{c}{} & \multicolumn{1}{c}{} \\ 
\multicolumn{1}{c}{} & \multicolumn{1}{c}{} & \multicolumn{1}{c}{} & 
\multicolumn{1}{c}{$100$} & \multicolumn{1}{c}{} & \multicolumn{1}{c}{} & 
\multicolumn{1}{c}{$0.060$} & \multicolumn{1}{c}{} & \multicolumn{1}{c}{$%
0.064$} & \multicolumn{1}{c}{} & \multicolumn{1}{c}{$0.053$} & 
\multicolumn{1}{c}{} & \multicolumn{1}{c}{$0.020$} & \multicolumn{1}{c}{} & 
\multicolumn{1}{c}{$0.034$} & \multicolumn{1}{c}{} & \multicolumn{1}{c}{$%
0.050$} & \multicolumn{1}{c}{} & \multicolumn{1}{c}{$0.056$} & 
\multicolumn{1}{c}{} & \multicolumn{1}{c}{$0.058$} \\ 
\multicolumn{1}{c}{} & \multicolumn{1}{c}{$200$} & \multicolumn{1}{c}{} & 
\multicolumn{1}{c}{$200$} & \multicolumn{1}{c}{} & \multicolumn{1}{c}{} & 
\multicolumn{1}{c}{$0.049$} & \multicolumn{1}{c}{} & \multicolumn{1}{c}{$%
0.063$} & \multicolumn{1}{c}{} & \multicolumn{1}{c}{$0.058$} & 
\multicolumn{1}{c}{} & \multicolumn{1}{c}{$0.023$} & \multicolumn{1}{c}{} & 
\multicolumn{1}{c}{$0.044$} & \multicolumn{1}{c}{} & \multicolumn{1}{c}{$%
0.045$} & \multicolumn{1}{c}{} & \multicolumn{1}{c}{$0.058$} & 
\multicolumn{1}{c}{} & \multicolumn{1}{c}{$0.055$} \\ 
\multicolumn{1}{c}{} & \multicolumn{1}{c}{} & \multicolumn{1}{c}{} & 
\multicolumn{1}{c}{$400$} & \multicolumn{1}{c}{} & \multicolumn{1}{c}{} & 
\multicolumn{1}{c}{$0.057$} & \multicolumn{1}{c}{} & \multicolumn{1}{c}{$%
0.057$} & \multicolumn{1}{c}{} & \multicolumn{1}{c}{$0.060$} & 
\multicolumn{1}{c}{} & \multicolumn{1}{c}{$0.023$} & \multicolumn{1}{c}{} & 
\multicolumn{1}{c}{$0.042$} & \multicolumn{1}{c}{} & \multicolumn{1}{c}{$%
0.055$} & \multicolumn{1}{c}{} & \multicolumn{1}{c}{$0.057$} & 
\multicolumn{1}{c}{} & \multicolumn{1}{c}{$0.052$} \\ 
\multicolumn{1}{c}{} & \multicolumn{1}{c}{} & \multicolumn{1}{c}{} & 
\multicolumn{1}{c}{$800$} & \multicolumn{1}{c}{} & \multicolumn{1}{c}{} & 
\multicolumn{1}{c}{$0.049$} & \multicolumn{1}{c}{} & \multicolumn{1}{c}{$%
0.050$} & \multicolumn{1}{c}{} & \multicolumn{1}{c}{$0.056$} & 
\multicolumn{1}{c}{} & \multicolumn{1}{c}{$0.024$} & \multicolumn{1}{c}{} & 
\multicolumn{1}{c}{$0.048$} & \multicolumn{1}{c}{} & \multicolumn{1}{c}{$%
0.047$} & \multicolumn{1}{c}{} & \multicolumn{1}{c}{$0.062$} & 
\multicolumn{1}{c}{} & \multicolumn{1}{c}{$0.067$} \\ 
\multicolumn{1}{c}{} & \multicolumn{1}{c}{} & \multicolumn{1}{c}{} & 
\multicolumn{1}{c}{} & \multicolumn{1}{c}{} & \multicolumn{1}{c}{} & 
\multicolumn{1}{c}{} & \multicolumn{1}{c}{} & \multicolumn{1}{c}{} & 
\multicolumn{1}{c}{} & \multicolumn{1}{c}{} & \multicolumn{1}{c}{} & 
\multicolumn{1}{c}{} & \multicolumn{1}{c}{} & \multicolumn{1}{c}{} &  &  & 
&  &  &  \\ \hline\hline
\end{tabular}
}
\par
{\scriptsize {\footnotesize 
\begin{tablenotes}
      \tiny
            \item 
            
\end{tablenotes}
} }
\end{table*}

\begin{table*}[h!]
\caption{{\protect\footnotesize {Empirical rejection frequencies under the
null of no changepoint and no covariates - Case II, $\protect\beta_0=1.05$}}}
\label{tab:ERF2}\centering
\par
{\scriptsize 
\begin{tabular}{lllllllllllllllllllll}
\hline\hline
&  &  &  &  &  &  &  &  &  &  &  &  &  &  &  &  &  &  &  &  \\ 
&  &  &  &  &  & \multicolumn{5}{c}{Weighted CUSUM} & \multicolumn{1}{c}{} & 
\multicolumn{3}{c}{Standardised CUSUM} & \multicolumn{1}{c}{} & 
\multicolumn{5}{c}{Weighted Page-CUSUM} \\ 
&  &  &  & $\psi $ &  & \multicolumn{1}{c}{$0$} & \multicolumn{1}{c}{} & 
\multicolumn{1}{c}{$0.25$} & \multicolumn{1}{c}{} & \multicolumn{1}{c}{$0.45$%
} & \multicolumn{1}{c}{} & \multicolumn{1}{c}{} & \multicolumn{1}{c}{$0.5$}
& \multicolumn{1}{c}{} & \multicolumn{1}{c}{} & \multicolumn{1}{c}{$0$} & 
\multicolumn{1}{c}{} & \multicolumn{1}{c}{$0.25$} & \multicolumn{1}{c}{} & 
\multicolumn{1}{c}{$0.45$} \\ 
&  &  &  &  &  & \multicolumn{1}{c}{} & \multicolumn{1}{c}{} & 
\multicolumn{1}{c}{} & \multicolumn{1}{c}{} & \multicolumn{1}{c}{} & 
\multicolumn{1}{c}{} & \multicolumn{1}{c}{$c_{\alpha ,0.5}$} & 
\multicolumn{1}{c}{} & \multicolumn{1}{c}{$\widehat{c}_{\alpha ,0.5}$} & 
\multicolumn{1}{c}{} & \multicolumn{1}{c}{} & \multicolumn{1}{c}{} & 
\multicolumn{1}{c}{} & \multicolumn{1}{c}{} & \multicolumn{1}{c}{} \\ 
&  &  &  &  &  & \multicolumn{1}{c}{} & \multicolumn{1}{c}{} & 
\multicolumn{1}{c}{} & \multicolumn{1}{c}{} & \multicolumn{1}{c}{} & 
\multicolumn{1}{c}{} & \multicolumn{1}{c}{} & \multicolumn{1}{c}{} & 
\multicolumn{1}{c}{} & \multicolumn{1}{c}{} & \multicolumn{1}{c}{} & 
\multicolumn{1}{c}{} & \multicolumn{1}{c}{} & \multicolumn{1}{c}{} & 
\multicolumn{1}{c}{} \\ 
\multicolumn{1}{c}{} & \multicolumn{1}{c}{$m$} & \multicolumn{1}{c}{} & 
\multicolumn{1}{c}{$m^{\ast }$} & \multicolumn{1}{c}{} & \multicolumn{1}{c}{}
& \multicolumn{1}{c}{} & \multicolumn{1}{c}{} & \multicolumn{1}{c}{} & 
\multicolumn{1}{c}{} & \multicolumn{1}{c}{} & \multicolumn{1}{c}{} & 
\multicolumn{1}{c}{} & \multicolumn{1}{c}{} & \multicolumn{1}{c}{} & 
\multicolumn{1}{c}{} & \multicolumn{1}{c}{} & \multicolumn{1}{c}{} & 
\multicolumn{1}{c}{} & \multicolumn{1}{c}{} & \multicolumn{1}{c}{} \\ 
\multicolumn{1}{c}{} & \multicolumn{1}{c}{} & \multicolumn{1}{c}{} & 
\multicolumn{1}{c}{} & \multicolumn{1}{c}{} & \multicolumn{1}{c}{} & 
\multicolumn{1}{c}{} & \multicolumn{1}{c}{} & \multicolumn{1}{c}{} & 
\multicolumn{1}{c}{} & \multicolumn{1}{c}{} & \multicolumn{1}{c}{} & 
\multicolumn{1}{c}{} & \multicolumn{1}{c}{} & \multicolumn{1}{c}{} & 
\multicolumn{1}{c}{} & \multicolumn{1}{c}{} & \multicolumn{1}{c}{} & 
\multicolumn{1}{c}{} & \multicolumn{1}{c}{} & \multicolumn{1}{c}{} \\ 
\multicolumn{1}{c}{} & \multicolumn{1}{c}{} & \multicolumn{1}{c}{} & 
\multicolumn{1}{c}{$25$} & \multicolumn{1}{c}{} & \multicolumn{1}{c}{} & 
\multicolumn{1}{c}{$0.035$} & \multicolumn{1}{c}{} & \multicolumn{1}{c}{$%
0.042$} & \multicolumn{1}{c}{} & \multicolumn{1}{c}{$0.026$} & 
\multicolumn{1}{c}{} & \multicolumn{1}{c}{$0.011$} & \multicolumn{1}{c}{} & 
\multicolumn{1}{c}{$0.022$} & \multicolumn{1}{c}{} & \multicolumn{1}{c}{$%
0.027$} & \multicolumn{1}{c}{} & \multicolumn{1}{c}{$0.037$} & 
\multicolumn{1}{c}{} & \multicolumn{1}{c}{$0.030$} \\ 
\multicolumn{1}{c}{} & \multicolumn{1}{c}{$50$} & \multicolumn{1}{c}{} & 
\multicolumn{1}{c}{$50$} & \multicolumn{1}{c}{} & \multicolumn{1}{c}{} & 
\multicolumn{1}{c}{$0.046$} & \multicolumn{1}{c}{} & \multicolumn{1}{c}{$%
0.051$} & \multicolumn{1}{c}{} & \multicolumn{1}{c}{$0.040$} & 
\multicolumn{1}{c}{} & \multicolumn{1}{c}{$0.012$} & \multicolumn{1}{c}{} & 
\multicolumn{1}{c}{$0.029$} & \multicolumn{1}{c}{} & \multicolumn{1}{c}{$%
0.044$} & \multicolumn{1}{c}{} & \multicolumn{1}{c}{$0.054$} & 
\multicolumn{1}{c}{} & \multicolumn{1}{c}{$0.039$} \\ 
\multicolumn{1}{c}{} & \multicolumn{1}{c}{} & \multicolumn{1}{c}{} & 
\multicolumn{1}{c}{$100$} & \multicolumn{1}{c}{} & \multicolumn{1}{c}{} & 
\multicolumn{1}{c}{$0.033$} & \multicolumn{1}{c}{} & \multicolumn{1}{c}{$%
0.049$} & \multicolumn{1}{c}{} & \multicolumn{1}{c}{$0.027$} & 
\multicolumn{1}{c}{} & \multicolumn{1}{c}{$0.009$} & \multicolumn{1}{c}{} & 
\multicolumn{1}{c}{$0.023$} & \multicolumn{1}{c}{} & \multicolumn{1}{c}{$%
0.028$} & \multicolumn{1}{c}{} & \multicolumn{1}{c}{$0.026$} & 
\multicolumn{1}{c}{} & \multicolumn{1}{c}{$0.025$} \\ 
\multicolumn{1}{c}{} & \multicolumn{1}{c}{} & \multicolumn{1}{c}{} & 
\multicolumn{1}{c}{$200$} & \multicolumn{1}{c}{} & \multicolumn{1}{c}{} & 
\multicolumn{1}{c}{$0.055$} & \multicolumn{1}{c}{} & \multicolumn{1}{c}{$%
0.063$} & \multicolumn{1}{c}{} & \multicolumn{1}{c}{$0.044$} & 
\multicolumn{1}{c}{} & \multicolumn{1}{c}{$0.018$} & \multicolumn{1}{c}{} & 
\multicolumn{1}{c}{$0.034$} & \multicolumn{1}{c}{} & \multicolumn{1}{c}{$%
0.051$} & \multicolumn{1}{c}{} & \multicolumn{1}{c}{$0.056$} & 
\multicolumn{1}{c}{} & \multicolumn{1}{c}{$0.047$} \\ 
\multicolumn{1}{c}{} & \multicolumn{1}{c}{} & \multicolumn{1}{c}{} & 
\multicolumn{1}{c}{} & \multicolumn{1}{c}{} & \multicolumn{1}{c}{} & 
\multicolumn{1}{c}{} & \multicolumn{1}{c}{} & \multicolumn{1}{c}{} & 
\multicolumn{1}{c}{} & \multicolumn{1}{c}{} & \multicolumn{1}{c}{} & 
\multicolumn{1}{c}{} & \multicolumn{1}{c}{} & \multicolumn{1}{c}{} & 
\multicolumn{1}{c}{} & \multicolumn{1}{c}{} & \multicolumn{1}{c}{} & 
\multicolumn{1}{c}{} & \multicolumn{1}{c}{} & \multicolumn{1}{c}{} \\ 
\multicolumn{1}{c}{} & \multicolumn{1}{c}{} & \multicolumn{1}{c}{} & 
\multicolumn{1}{c}{} & \multicolumn{1}{c}{} & \multicolumn{1}{c}{} & 
\multicolumn{1}{c}{} & \multicolumn{1}{c}{} & \multicolumn{1}{c}{} & 
\multicolumn{1}{c}{} & \multicolumn{1}{c}{} & \multicolumn{1}{c}{} & 
\multicolumn{1}{c}{} & \multicolumn{1}{c}{} & \multicolumn{1}{c}{} & 
\multicolumn{1}{c}{} & \multicolumn{1}{c}{} & \multicolumn{1}{c}{} & 
\multicolumn{1}{c}{} & \multicolumn{1}{c}{} & \multicolumn{1}{c}{} \\ 
\multicolumn{1}{c}{} & \multicolumn{1}{c}{} & \multicolumn{1}{c}{} & 
\multicolumn{1}{c}{$50$} & \multicolumn{1}{c}{} & \multicolumn{1}{c}{} & 
\multicolumn{1}{c}{$0.044$} & \multicolumn{1}{c}{} & \multicolumn{1}{c}{$%
0.043$} & \multicolumn{1}{c}{} & \multicolumn{1}{c}{$0.030$} & 
\multicolumn{1}{c}{} & \multicolumn{1}{c}{$0.008$} & \multicolumn{1}{c}{} & 
\multicolumn{1}{c}{$0.020$} & \multicolumn{1}{c}{} & \multicolumn{1}{c}{$%
0.033$} & \multicolumn{1}{c}{} & \multicolumn{1}{c}{$0.038$} & 
\multicolumn{1}{c}{} & \multicolumn{1}{c}{$0.027$} \\ 
\multicolumn{1}{c}{} & \multicolumn{1}{c}{$100$} & \multicolumn{1}{c}{} & 
\multicolumn{1}{c}{$100$} & \multicolumn{1}{c}{} & \multicolumn{1}{c}{} & 
\multicolumn{1}{c}{$0.033$} & \multicolumn{1}{c}{} & \multicolumn{1}{c}{$%
0.036$} & \multicolumn{1}{c}{} & \multicolumn{1}{c}{$0.039$} & 
\multicolumn{1}{c}{} & \multicolumn{1}{c}{$0.015$} & \multicolumn{1}{c}{} & 
\multicolumn{1}{c}{$0.025$} & \multicolumn{1}{c}{} & \multicolumn{1}{c}{$%
0.034$} & \multicolumn{1}{c}{} & \multicolumn{1}{c}{$0.034$} & 
\multicolumn{1}{c}{} & \multicolumn{1}{c}{$0.028$} \\ 
\multicolumn{1}{c}{} & \multicolumn{1}{c}{} & \multicolumn{1}{c}{} & 
\multicolumn{1}{c}{$200$} & \multicolumn{1}{c}{} & \multicolumn{1}{c}{} & 
\multicolumn{1}{c}{$0.054$} & \multicolumn{1}{c}{} & \multicolumn{1}{c}{$%
0.052$} & \multicolumn{1}{c}{} & \multicolumn{1}{c}{$0.041$} & 
\multicolumn{1}{c}{} & \multicolumn{1}{c}{$0.009$} & \multicolumn{1}{c}{} & 
\multicolumn{1}{c}{$0.019$} & \multicolumn{1}{c}{} & \multicolumn{1}{c}{$%
0.051$} & \multicolumn{1}{c}{} & \multicolumn{1}{c}{$0.040$} & 
\multicolumn{1}{c}{} & \multicolumn{1}{c}{$0.024$} \\ 
\multicolumn{1}{c}{} & \multicolumn{1}{c}{} & \multicolumn{1}{c}{} & 
\multicolumn{1}{c}{$400$} & \multicolumn{1}{c}{} & \multicolumn{1}{c}{} & 
\multicolumn{1}{c}{$0.059$} & \multicolumn{1}{c}{} & \multicolumn{1}{c}{$%
0.050$} & \multicolumn{1}{c}{} & \multicolumn{1}{c}{$0.046$} & 
\multicolumn{1}{c}{} & \multicolumn{1}{c}{$0.014$} & \multicolumn{1}{c}{} & 
\multicolumn{1}{c}{$0.035$} & \multicolumn{1}{c}{} & \multicolumn{1}{c}{$%
0.052$} & \multicolumn{1}{c}{} & \multicolumn{1}{c}{$0.055$} & 
\multicolumn{1}{c}{} & \multicolumn{1}{c}{$0.044$} \\ 
\multicolumn{1}{c}{} & \multicolumn{1}{c}{} & \multicolumn{1}{c}{} & 
\multicolumn{1}{c}{} & \multicolumn{1}{c}{} & \multicolumn{1}{c}{} & 
\multicolumn{1}{c}{} & \multicolumn{1}{c}{} & \multicolumn{1}{c}{} & 
\multicolumn{1}{c}{} & \multicolumn{1}{c}{} & \multicolumn{1}{c}{} & 
\multicolumn{1}{c}{} & \multicolumn{1}{c}{} & \multicolumn{1}{c}{} & 
\multicolumn{1}{c}{} & \multicolumn{1}{c}{} & \multicolumn{1}{c}{} & 
\multicolumn{1}{c}{} & \multicolumn{1}{c}{} & \multicolumn{1}{c}{} \\ 
\multicolumn{1}{c}{} & \multicolumn{1}{c}{} & \multicolumn{1}{c}{} & 
\multicolumn{1}{c}{} & \multicolumn{1}{c}{} & \multicolumn{1}{c}{} & 
\multicolumn{1}{c}{} & \multicolumn{1}{c}{} & \multicolumn{1}{c}{} & 
\multicolumn{1}{c}{} & \multicolumn{1}{c}{} & \multicolumn{1}{c}{} & 
\multicolumn{1}{c}{} & \multicolumn{1}{c}{} & \multicolumn{1}{c}{} & 
\multicolumn{1}{c}{} & \multicolumn{1}{c}{} & \multicolumn{1}{c}{} & 
\multicolumn{1}{c}{} & \multicolumn{1}{c}{} & \multicolumn{1}{c}{} \\ 
\multicolumn{1}{c}{} & \multicolumn{1}{c}{} & \multicolumn{1}{c}{} & 
\multicolumn{1}{c}{$100$} & \multicolumn{1}{c}{} & \multicolumn{1}{c}{} & 
\multicolumn{1}{c}{$0.047$} & \multicolumn{1}{c}{} & \multicolumn{1}{c}{$%
0.049$} & \multicolumn{1}{c}{} & \multicolumn{1}{c}{$0.034$} & 
\multicolumn{1}{c}{} & \multicolumn{1}{c}{$0.015$} & \multicolumn{1}{c}{} & 
\multicolumn{1}{c}{$0.027$} & \multicolumn{1}{c}{} & \multicolumn{1}{c}{$%
0.039$} & \multicolumn{1}{c}{} & \multicolumn{1}{c}{$0.040$} & 
\multicolumn{1}{c}{} & \multicolumn{1}{c}{$0.038$} \\ 
\multicolumn{1}{c}{} & \multicolumn{1}{c}{$200$} & \multicolumn{1}{c}{} & 
\multicolumn{1}{c}{$200$} & \multicolumn{1}{c}{} & \multicolumn{1}{c}{} & 
\multicolumn{1}{c}{$0.041$} & \multicolumn{1}{c}{} & \multicolumn{1}{c}{$%
0.055$} & \multicolumn{1}{c}{} & \multicolumn{1}{c}{$0.050$} & 
\multicolumn{1}{c}{} & \multicolumn{1}{c}{$0.022$} & \multicolumn{1}{c}{} & 
\multicolumn{1}{c}{$0.036$} & \multicolumn{1}{c}{} & \multicolumn{1}{c}{$%
0.041$} & \multicolumn{1}{c}{} & \multicolumn{1}{c}{$0.055$} & 
\multicolumn{1}{c}{} & \multicolumn{1}{c}{$0.047$} \\ 
\multicolumn{1}{c}{} & \multicolumn{1}{c}{} & \multicolumn{1}{c}{} & 
\multicolumn{1}{c}{$400$} & \multicolumn{1}{c}{} & \multicolumn{1}{c}{} & 
\multicolumn{1}{c}{$0.051$} & \multicolumn{1}{c}{} & \multicolumn{1}{c}{$%
0.054$} & \multicolumn{1}{c}{} & \multicolumn{1}{c}{$0.053$} & 
\multicolumn{1}{c}{} & \multicolumn{1}{c}{$0.015$} & \multicolumn{1}{c}{} & 
\multicolumn{1}{c}{$0.037$} & \multicolumn{1}{c}{} & \multicolumn{1}{c}{$%
0.049$} & \multicolumn{1}{c}{} & \multicolumn{1}{c}{$0.046$} & 
\multicolumn{1}{c}{} & \multicolumn{1}{c}{$0.046$} \\ 
\multicolumn{1}{c}{} & \multicolumn{1}{c}{} & \multicolumn{1}{c}{} & 
\multicolumn{1}{c}{$800$} & \multicolumn{1}{c}{} & \multicolumn{1}{c}{} & 
\multicolumn{1}{c}{$0.042$} & \multicolumn{1}{c}{} & \multicolumn{1}{c}{$%
0.042$} & \multicolumn{1}{c}{} & \multicolumn{1}{c}{$0.047$} & 
\multicolumn{1}{c}{} & \multicolumn{1}{c}{$0.013$} & \multicolumn{1}{c}{} & 
\multicolumn{1}{c}{$0.023$} & \multicolumn{1}{c}{} & \multicolumn{1}{c}{$%
0.037$} & \multicolumn{1}{c}{} & \multicolumn{1}{c}{$0.048$} & 
\multicolumn{1}{c}{} & \multicolumn{1}{c}{$0.047$} \\ 
\multicolumn{1}{c}{} & \multicolumn{1}{c}{} & \multicolumn{1}{c}{} & 
\multicolumn{1}{c}{} & \multicolumn{1}{c}{} & \multicolumn{1}{c}{} & 
\multicolumn{1}{c}{} & \multicolumn{1}{c}{} & \multicolumn{1}{c}{} & 
\multicolumn{1}{c}{} & \multicolumn{1}{c}{} & \multicolumn{1}{c}{} & 
\multicolumn{1}{c}{} & \multicolumn{1}{c}{} & \multicolumn{1}{c}{} &  &  & 
&  &  &  \\ \hline\hline
\end{tabular}
}
\par
{\scriptsize {\footnotesize 
\begin{tablenotes}
      \tiny
            \item 
            
\end{tablenotes}
} }
\end{table*}

\begin{table*}[h]
\caption{{\protect\footnotesize {Empirical rejection frequencies under the
null of no changepoint and no covariates - Case III, $\protect\beta_0=1$}}}
\label{tab:ERF3}\centering
\par
{\scriptsize 
\begin{tabular}{lllllllllllllllllllll}
\hline\hline
&  &  &  &  &  &  &  &  &  &  &  &  &  &  &  &  &  &  &  &  \\ 
&  &  &  &  &  & \multicolumn{5}{c}{Weighted CUSUM} & \multicolumn{1}{c}{} & 
\multicolumn{3}{c}{Standardised CUSUM} & \multicolumn{1}{c}{} & 
\multicolumn{5}{c}{Weighted Page-CUSUM} \\ 
&  &  &  & $\psi $ &  & \multicolumn{1}{c}{$0$} & \multicolumn{1}{c}{} & 
\multicolumn{1}{c}{$0.25$} & \multicolumn{1}{c}{} & \multicolumn{1}{c}{$0.45$%
} & \multicolumn{1}{c}{} & \multicolumn{1}{c}{} & \multicolumn{1}{c}{$0.5$}
& \multicolumn{1}{c}{} & \multicolumn{1}{c}{} & \multicolumn{1}{c}{$0$} & 
\multicolumn{1}{c}{} & \multicolumn{1}{c}{$0.25$} & \multicolumn{1}{c}{} & 
\multicolumn{1}{c}{$0.45$} \\ 
&  &  &  &  &  & \multicolumn{1}{c}{} & \multicolumn{1}{c}{} & 
\multicolumn{1}{c}{} & \multicolumn{1}{c}{} & \multicolumn{1}{c}{} & 
\multicolumn{1}{c}{} & \multicolumn{1}{c}{$c_{\alpha ,0.5}$} & 
\multicolumn{1}{c}{} & \multicolumn{1}{c}{$\widehat{c}_{\alpha ,0.5}$} & 
\multicolumn{1}{c}{} & \multicolumn{1}{c}{} & \multicolumn{1}{c}{} & 
\multicolumn{1}{c}{} & \multicolumn{1}{c}{} & \multicolumn{1}{c}{} \\ 
&  &  &  &  &  & \multicolumn{1}{c}{} & \multicolumn{1}{c}{} & 
\multicolumn{1}{c}{} & \multicolumn{1}{c}{} & \multicolumn{1}{c}{} & 
\multicolumn{1}{c}{} & \multicolumn{1}{c}{} & \multicolumn{1}{c}{} & 
\multicolumn{1}{c}{} & \multicolumn{1}{c}{} & \multicolumn{1}{c}{} & 
\multicolumn{1}{c}{} & \multicolumn{1}{c}{} & \multicolumn{1}{c}{} & 
\multicolumn{1}{c}{} \\ 
\multicolumn{1}{c}{} & \multicolumn{1}{c}{$m$} & \multicolumn{1}{c}{} & 
\multicolumn{1}{c}{$m^{\ast }$} & \multicolumn{1}{c}{} & \multicolumn{1}{c}{}
& \multicolumn{1}{c}{} & \multicolumn{1}{c}{} & \multicolumn{1}{c}{} & 
\multicolumn{1}{c}{} & \multicolumn{1}{c}{} & \multicolumn{1}{c}{} & 
\multicolumn{1}{c}{} & \multicolumn{1}{c}{} & \multicolumn{1}{c}{} & 
\multicolumn{1}{c}{} & \multicolumn{1}{c}{} & \multicolumn{1}{c}{} & 
\multicolumn{1}{c}{} & \multicolumn{1}{c}{} & \multicolumn{1}{c}{} \\ 
\multicolumn{1}{c}{} & \multicolumn{1}{c}{} & \multicolumn{1}{c}{} & 
\multicolumn{1}{c}{} & \multicolumn{1}{c}{} & \multicolumn{1}{c}{} & 
\multicolumn{1}{c}{} & \multicolumn{1}{c}{} & \multicolumn{1}{c}{} & 
\multicolumn{1}{c}{} & \multicolumn{1}{c}{} & \multicolumn{1}{c}{} & 
\multicolumn{1}{c}{} & \multicolumn{1}{c}{} & \multicolumn{1}{c}{} & 
\multicolumn{1}{c}{} & \multicolumn{1}{c}{} & \multicolumn{1}{c}{} & 
\multicolumn{1}{c}{} & \multicolumn{1}{c}{} & \multicolumn{1}{c}{} \\ 
\multicolumn{1}{c}{} & \multicolumn{1}{c}{} & \multicolumn{1}{c}{} & 
\multicolumn{1}{c}{$25$} & \multicolumn{1}{c}{} & \multicolumn{1}{c}{} & 
\multicolumn{1}{c}{$0.070$} & \multicolumn{1}{c}{} & \multicolumn{1}{c}{$%
0.069$} & \multicolumn{1}{c}{} & \multicolumn{1}{c}{$0.053$} & 
\multicolumn{1}{c}{} & \multicolumn{1}{c}{$0.028$} & \multicolumn{1}{c}{} & 
\multicolumn{1}{c}{$0.047$} & \multicolumn{1}{c}{} & \multicolumn{1}{c}{$%
0.062$} & \multicolumn{1}{c}{} & \multicolumn{1}{c}{$0.072$} & 
\multicolumn{1}{c}{} & \multicolumn{1}{c}{$0.058$} \\ 
\multicolumn{1}{c}{} & \multicolumn{1}{c}{$50$} & \multicolumn{1}{c}{} & 
\multicolumn{1}{c}{$50$} & \multicolumn{1}{c}{} & \multicolumn{1}{c}{} & 
\multicolumn{1}{c}{$0.068$} & \multicolumn{1}{c}{} & \multicolumn{1}{c}{$%
0.073$} & \multicolumn{1}{c}{} & \multicolumn{1}{c}{$0.055$} & 
\multicolumn{1}{c}{} & \multicolumn{1}{c}{$0.031$} & \multicolumn{1}{c}{} & 
\multicolumn{1}{c}{$0.044$} & \multicolumn{1}{c}{} & \multicolumn{1}{c}{$%
0.066$} & \multicolumn{1}{c}{} & \multicolumn{1}{c}{$0.076$} & 
\multicolumn{1}{c}{} & \multicolumn{1}{c}{$0.058$} \\ 
\multicolumn{1}{c}{} & \multicolumn{1}{c}{} & \multicolumn{1}{c}{} & 
\multicolumn{1}{c}{$100$} & \multicolumn{1}{c}{} & \multicolumn{1}{c}{} & 
\multicolumn{1}{c}{$0.092$} & \multicolumn{1}{c}{} & \multicolumn{1}{c}{$%
0.100$} & \multicolumn{1}{c}{} & \multicolumn{1}{c}{$0.078$} & 
\multicolumn{1}{c}{} & \multicolumn{1}{c}{$0.053$} & \multicolumn{1}{c}{} & 
\multicolumn{1}{c}{$0.070$} & \multicolumn{1}{c}{} & \multicolumn{1}{c}{$%
0.089$} & \multicolumn{1}{c}{} & \multicolumn{1}{c}{$0.076$} & 
\multicolumn{1}{c}{} & \multicolumn{1}{c}{$0.076$} \\ 
\multicolumn{1}{c}{} & \multicolumn{1}{c}{} & \multicolumn{1}{c}{} & 
\multicolumn{1}{c}{$200$} & \multicolumn{1}{c}{} & \multicolumn{1}{c}{} & 
\multicolumn{1}{c}{$0.111$} & \multicolumn{1}{c}{} & \multicolumn{1}{c}{$%
0.114$} & \multicolumn{1}{c}{} & \multicolumn{1}{c}{$0.092$} & 
\multicolumn{1}{c}{} & \multicolumn{1}{c}{$0.055$} & \multicolumn{1}{c}{} & 
\multicolumn{1}{c}{$0.073$} & \multicolumn{1}{c}{} & \multicolumn{1}{c}{$%
0.109$} & \multicolumn{1}{c}{} & \multicolumn{1}{c}{$0.098$} & 
\multicolumn{1}{c}{} & \multicolumn{1}{c}{$0.091$} \\ 
\multicolumn{1}{c}{} & \multicolumn{1}{c}{} & \multicolumn{1}{c}{} & 
\multicolumn{1}{c}{} & \multicolumn{1}{c}{} & \multicolumn{1}{c}{} & 
\multicolumn{1}{c}{} & \multicolumn{1}{c}{} & \multicolumn{1}{c}{} & 
\multicolumn{1}{c}{} & \multicolumn{1}{c}{} & \multicolumn{1}{c}{} & 
\multicolumn{1}{c}{} & \multicolumn{1}{c}{} & \multicolumn{1}{c}{} & 
\multicolumn{1}{c}{} & \multicolumn{1}{c}{} & \multicolumn{1}{c}{} & 
\multicolumn{1}{c}{} & \multicolumn{1}{c}{} & \multicolumn{1}{c}{} \\ 
\multicolumn{1}{c}{} & \multicolumn{1}{c}{} & \multicolumn{1}{c}{} & 
\multicolumn{1}{c}{} & \multicolumn{1}{c}{} & \multicolumn{1}{c}{} & 
\multicolumn{1}{c}{} & \multicolumn{1}{c}{} & \multicolumn{1}{c}{} & 
\multicolumn{1}{c}{} & \multicolumn{1}{c}{} & \multicolumn{1}{c}{} & 
\multicolumn{1}{c}{} & \multicolumn{1}{c}{} & \multicolumn{1}{c}{} & 
\multicolumn{1}{c}{} & \multicolumn{1}{c}{} & \multicolumn{1}{c}{} & 
\multicolumn{1}{c}{} & \multicolumn{1}{c}{} & \multicolumn{1}{c}{} \\ 
\multicolumn{1}{c}{} & \multicolumn{1}{c}{} & \multicolumn{1}{c}{} & 
\multicolumn{1}{c}{$50$} & \multicolumn{1}{c}{} & \multicolumn{1}{c}{} & 
\multicolumn{1}{c}{$0.070$} & \multicolumn{1}{c}{} & \multicolumn{1}{c}{$%
0.068$} & \multicolumn{1}{c}{} & \multicolumn{1}{c}{$0.049$} & 
\multicolumn{1}{c}{} & \multicolumn{1}{c}{$0.020$} & \multicolumn{1}{c}{} & 
\multicolumn{1}{c}{$0.037$} & \multicolumn{1}{c}{} & \multicolumn{1}{c}{$%
0.064$} & \multicolumn{1}{c}{} & \multicolumn{1}{c}{$0.070$} & 
\multicolumn{1}{c}{} & \multicolumn{1}{c}{$0.049$} \\ 
\multicolumn{1}{c}{} & \multicolumn{1}{c}{$100$} & \multicolumn{1}{c}{} & 
\multicolumn{1}{c}{$100$} & \multicolumn{1}{c}{} & \multicolumn{1}{c}{} & 
\multicolumn{1}{c}{$0.075$} & \multicolumn{1}{c}{} & \multicolumn{1}{c}{$%
0.082$} & \multicolumn{1}{c}{} & \multicolumn{1}{c}{$0.072$} & 
\multicolumn{1}{c}{} & \multicolumn{1}{c}{$0.038$} & \multicolumn{1}{c}{} & 
\multicolumn{1}{c}{$0.057$} & \multicolumn{1}{c}{} & \multicolumn{1}{c}{$%
0.077$} & \multicolumn{1}{c}{} & \multicolumn{1}{c}{$0.078$} & 
\multicolumn{1}{c}{} & \multicolumn{1}{c}{$0.069$} \\ 
\multicolumn{1}{c}{} & \multicolumn{1}{c}{} & \multicolumn{1}{c}{} & 
\multicolumn{1}{c}{$200$} & \multicolumn{1}{c}{} & \multicolumn{1}{c}{} & 
\multicolumn{1}{c}{$0.095$} & \multicolumn{1}{c}{} & \multicolumn{1}{c}{$%
0.099$} & \multicolumn{1}{c}{} & \multicolumn{1}{c}{$0.085$} & 
\multicolumn{1}{c}{} & \multicolumn{1}{c}{$0.047$} & \multicolumn{1}{c}{} & 
\multicolumn{1}{c}{$0.067$} & \multicolumn{1}{c}{} & \multicolumn{1}{c}{$%
0.092$} & \multicolumn{1}{c}{} & \multicolumn{1}{c}{$0.083$} & 
\multicolumn{1}{c}{} & \multicolumn{1}{c}{$0.077$} \\ 
\multicolumn{1}{c}{} & \multicolumn{1}{c}{} & \multicolumn{1}{c}{} & 
\multicolumn{1}{c}{$400$} & \multicolumn{1}{c}{} & \multicolumn{1}{c}{} & 
\multicolumn{1}{c}{$0.111$} & \multicolumn{1}{c}{} & \multicolumn{1}{c}{$%
0.100$} & \multicolumn{1}{c}{} & \multicolumn{1}{c}{$0.089$} & 
\multicolumn{1}{c}{} & \multicolumn{1}{c}{$0.057$} & \multicolumn{1}{c}{} & 
\multicolumn{1}{c}{$0.076$} & \multicolumn{1}{c}{} & \multicolumn{1}{c}{$%
0.106$} & \multicolumn{1}{c}{} & \multicolumn{1}{c}{$0.106$} & 
\multicolumn{1}{c}{} & \multicolumn{1}{c}{$0.092$} \\ 
\multicolumn{1}{c}{} & \multicolumn{1}{c}{} & \multicolumn{1}{c}{} & 
\multicolumn{1}{c}{} & \multicolumn{1}{c}{} & \multicolumn{1}{c}{} & 
\multicolumn{1}{c}{} & \multicolumn{1}{c}{} & \multicolumn{1}{c}{} & 
\multicolumn{1}{c}{} & \multicolumn{1}{c}{} & \multicolumn{1}{c}{} & 
\multicolumn{1}{c}{} & \multicolumn{1}{c}{} & \multicolumn{1}{c}{} & 
\multicolumn{1}{c}{} & \multicolumn{1}{c}{} & \multicolumn{1}{c}{} & 
\multicolumn{1}{c}{} & \multicolumn{1}{c}{} & \multicolumn{1}{c}{} \\ 
\multicolumn{1}{c}{} & \multicolumn{1}{c}{} & \multicolumn{1}{c}{} & 
\multicolumn{1}{c}{} & \multicolumn{1}{c}{} & \multicolumn{1}{c}{} & 
\multicolumn{1}{c}{} & \multicolumn{1}{c}{} & \multicolumn{1}{c}{} & 
\multicolumn{1}{c}{} & \multicolumn{1}{c}{} & \multicolumn{1}{c}{} & 
\multicolumn{1}{c}{} & \multicolumn{1}{c}{} & \multicolumn{1}{c}{} & 
\multicolumn{1}{c}{} & \multicolumn{1}{c}{} & \multicolumn{1}{c}{} & 
\multicolumn{1}{c}{} & \multicolumn{1}{c}{} & \multicolumn{1}{c}{} \\ 
\multicolumn{1}{c}{} & \multicolumn{1}{c}{} & \multicolumn{1}{c}{} & 
\multicolumn{1}{c}{$100$} & \multicolumn{1}{c}{} & \multicolumn{1}{c}{} & 
\multicolumn{1}{c}{$0.060$} & \multicolumn{1}{c}{} & \multicolumn{1}{c}{$%
0.066$} & \multicolumn{1}{c}{} & \multicolumn{1}{c}{$0.051$} & 
\multicolumn{1}{c}{} & \multicolumn{1}{c}{$0.024$} & \multicolumn{1}{c}{} & 
\multicolumn{1}{c}{$0.040$} & \multicolumn{1}{c}{} & \multicolumn{1}{c}{$%
0.057$} & \multicolumn{1}{c}{} & \multicolumn{1}{c}{$0.058$} & 
\multicolumn{1}{c}{} & \multicolumn{1}{c}{$0.052$} \\ 
\multicolumn{1}{c}{} & \multicolumn{1}{c}{$200$} & \multicolumn{1}{c}{} & 
\multicolumn{1}{c}{$200$} & \multicolumn{1}{c}{} & \multicolumn{1}{c}{} & 
\multicolumn{1}{c}{$0.093$} & \multicolumn{1}{c}{} & \multicolumn{1}{c}{$%
0.102$} & \multicolumn{1}{c}{} & \multicolumn{1}{c}{$0.096$} & 
\multicolumn{1}{c}{} & \multicolumn{1}{c}{$0.045$} & \multicolumn{1}{c}{} & 
\multicolumn{1}{c}{$0.059$} & \multicolumn{1}{c}{} & \multicolumn{1}{c}{$%
0.090$} & \multicolumn{1}{c}{} & \multicolumn{1}{c}{$0.097$} & 
\multicolumn{1}{c}{} & \multicolumn{1}{c}{$0.083$} \\ 
\multicolumn{1}{c}{} & \multicolumn{1}{c}{} & \multicolumn{1}{c}{} & 
\multicolumn{1}{c}{$400$} & \multicolumn{1}{c}{} & \multicolumn{1}{c}{} & 
\multicolumn{1}{c}{$0.076$} & \multicolumn{1}{c}{} & \multicolumn{1}{c}{$%
0.079$} & \multicolumn{1}{c}{} & \multicolumn{1}{c}{$0.071$} & 
\multicolumn{1}{c}{} & \multicolumn{1}{c}{$0.032$} & \multicolumn{1}{c}{} & 
\multicolumn{1}{c}{$0.060$} & \multicolumn{1}{c}{} & \multicolumn{1}{c}{$%
0.078$} & \multicolumn{1}{c}{} & \multicolumn{1}{c}{$0.076$} & 
\multicolumn{1}{c}{} & \multicolumn{1}{c}{$0.072$} \\ 
\multicolumn{1}{c}{} & \multicolumn{1}{c}{} & \multicolumn{1}{c}{} & 
\multicolumn{1}{c}{$800$} & \multicolumn{1}{c}{} & \multicolumn{1}{c}{} & 
\multicolumn{1}{c}{$0.093$} & \multicolumn{1}{c}{} & \multicolumn{1}{c}{$%
0.087$} & \multicolumn{1}{c}{} & \multicolumn{1}{c}{$0.082$} & 
\multicolumn{1}{c}{} & \multicolumn{1}{c}{$0.040$} & \multicolumn{1}{c}{} & 
\multicolumn{1}{c}{$0.057$} & \multicolumn{1}{c}{} & \multicolumn{1}{c}{$%
0.085$} & \multicolumn{1}{c}{} & \multicolumn{1}{c}{$0.096$} & 
\multicolumn{1}{c}{} & \multicolumn{1}{c}{$0.088$} \\ 
\multicolumn{1}{c}{} & \multicolumn{1}{c}{} & \multicolumn{1}{c}{} & 
\multicolumn{1}{c}{} & \multicolumn{1}{c}{} & \multicolumn{1}{c}{} & 
\multicolumn{1}{c}{} & \multicolumn{1}{c}{} & \multicolumn{1}{c}{} & 
\multicolumn{1}{c}{} & \multicolumn{1}{c}{} & \multicolumn{1}{c}{} & 
\multicolumn{1}{c}{} & \multicolumn{1}{c}{} & \multicolumn{1}{c}{} &  &  & 
&  &  &  \\ \hline\hline
\end{tabular}
}
\par
{\scriptsize {\footnotesize 
\begin{tablenotes}
      \tiny
            \item 
            
\end{tablenotes}
} }
\end{table*}

Under the alternative, monitoring is carried out for $m+1\leq i\leq
m+m^{\ast }$, with the same specifications as in Cases I-III and 
\begin{equation}
y_{i}=\left( \beta _{0}+\Delta I\left( i\geq m+1\right) +\epsilon
_{i,1}\right) y_{i-1}+\lambda _{0}x_{i}+\epsilon _{i,2}.  \label{gdp-power}
\end{equation}%
We set: $\Delta =0.5$ under Case I, with $\beta _{A}=0.75$ and $E\log
\left\vert \beta _{A}+\epsilon _{i,1}\right\vert =-0.298$, thus having a
change in persistence with $y_{i}$ stationary under the null and the
alternative; $\Delta =-0.1$ under Case II, with $\beta _{A}=0.969$ and $%
E\log \left\vert \beta _{A}+\epsilon _{i,1}\right\vert =-0.063$, thus having
a change from nonstationarity to stationarity; and $\Delta =0.1$ under Case III, with $\beta _{A}=1.1$ and $%
E\log \left\vert \beta _{A}+\epsilon _{i,1}\right\vert =0.089$, thus having
a change from stationarity to nonstationarity. Median delays and empirical rejection frequencies for $m=200$,
under the case of no covariates ($\lambda _{0}=0$), are in Table \ref%
{tab:Power1}.

\begin{table*}[h]
\caption{{\protect\footnotesize {Median delays and empirical rejection
frequencies under alternatives - no covariates}}}
\label{tab:Power1}\centering
\par
\resizebox{\textwidth}{!}{

\begin{tabular}{cccccccccccccccccccccc}
\hline\hline
&  &  &  &  &  &  &  &  &  &  &  &  &  &  &  &  &  &  &  &  &  \\ 
&  &  &  &  &  &  & \multicolumn{5}{c}{Weighted CUSUM} &  & 
\multicolumn{3}{c}{Standardised CUSUM} &  & \multicolumn{5}{c}{Weighted
Page-CUSUM } \\ 
& DGP &  &  &  & $\psi $ &  & $0$ &  & $0.25$ &  & $0.45$ &  &  & $0.5$ &  & 
& $0$ &  & $0.25$ &  & $0.45$ \\ 
&  &  &  &  &  &  &  &  &  &  &  &  & $c_{\alpha ,0.5}$ &  & $\widehat{c}_{\alpha ,0.5}$ &  &  &  &  &  &  \\ 
&  &  &  &  &  &  &  &  &  &  &  &  &  &  &  &  &  &  &  &  &  \\ 
\cline{2-22}
&  &  &  &  &  &  &  &  &  &  &  &  &  &  &  &  &  &  &  &  &  \\ 
&  &  &  & $100$ &  &  & $\underset{\left( 0.705\right) }{54}$ &  & $\underset{\left( 0.695\right) }{44}$ &  & $\underset{\left( 0.618\right) }{35.5}$ &  & $\underset{\left( 0.465\right) }{37}$ &  & $\underset{\left(
0.552\right) }{34}$ &  & $\underset{\left( 0.700\right) }{55}$ &  & $\underset{\left( 0.690\right) }{44}$ &  & $\underset{\left( 0.644\right) }{35.5}$ \\ 
& $\underset{\left( \beta _{0}=0.5\right) }{\text{\textbf{Case I}}}$ &  & $m^{\ast }$ & $200$ &  &  & $\underset{\left( 0.821\right) }{82}$ &  & $\underset{\left( 0.822\right) }{63}$ &  & $\underset{\left( 0.769\right) }{49}$ &  & $\underset{\left( 0.633\right) }{57}$ &  & $\underset{\left(
0.704\right) }{51}$ &  & $\underset{\left( 0.818\right) }{79}$ &  & $\underset{\left( 0.819\right) }{62}$ &  & $\underset{\left( 0.779\right) }{47}$ \\ 
&  &  &  & $400$ &  &  & $\underset{\left( 0.947\right) }{105}$ &  & $\underset{\left( 0.941\right) }{80}$ &  & $\underset{\left( 0.894\right) }{61}$ &  & $\underset{\left( 0.784\right) }{74}$ &  & $\underset{\left(
0.847\right) }{62}$ &  & $\underset{\left( 0.940\right) }{105}$ &  & $\underset{\left( 0.926\right) }{80.5}$ &  & $\underset{\left( 0.885\right) }{59}$ \\ 
&  &  &  & $800$ &  &  & $\underset{\left( 0.964\right) }{135}$ &  & $\underset{\left( 0.952\right) }{103}$ &  & $\underset{\left( 0.923\right) }{69}$ &  & $\underset{\left( 0.843\right) }{90}$ &  & $\underset{\left(
0.889\right) }{75.5}$ &  & $\underset{\left( 0.959\right) }{135}$ &  & $\underset{\left( 0.958\right) }{89}$ &  & $\underset{\left( 0.933\right) }{60}$ \\ 
&  &  &  &  &  &  &  &  &  &  &  &  &  &  &  &  &  &  &  &  &  \\ 
\cline{2-22}
&  &  &  &  &  &  &  &  &  &  &  &  &  &  &  &  &  &  &  &  &  \\ 
&  &  &  & $100$ &  &  & $\underset{\left( 1.000\right) }{20}$ &  & $\underset{\left( 1.000\right) }{14}$ &  & $\underset{\left( 1.000\right) }{11}$ &  & $\underset{\left( 1.000\right) }{12}$ &  & $\underset{\left(
1.000\right) }{11}$ &  & $\underset{\left( 1.000\right) }{20}$ &  & $\underset{\left( 1.000\right) }{14}$ &  & $\underset{\left( 1.000\right) }{11}$ \\ 
& $\underset{\left( \beta _{0}=1.05\right) }{\text{\textbf{Case II}}}$ &  & $m^{\ast }$ & $200$ &  &  & $\underset{\left( 1.000\right) }{26}$ &  & $\underset{\left( 1.000\right) }{16}$ &  & $\underset{\left( 1.000\right) }{11}$ &  & $\underset{\left( 1.000\right) }{13}$ &  & $\underset{\left(
1.000\right) }{11}$ &  & $\underset{\left( 1.000\right) }{25}$ &  & $\underset{\left( 1.000\right) }{16}$ &  & $\underset{\left( 1.000\right) }{10}$ \\ 
&  &  &  & $400$ &  &  & $\underset{\left( 1.000\right) }{29}$ &  & $\underset{\left( 1.000\right) }{18}$ &  & $\underset{\left( 1.000\right) }{11}$ &  & $\underset{\left( 1.000\right) }{13}$ &  & $\underset{\left(
1.000\right) }{11}$ &  & $\underset{\left( 1.000\right) }{30}$ &  & $\underset{\left( 1.000\right) }{18}$ &  & $\underset{\left( 1.000\right) }{11}$ \\ 
&  &  &  & $800$ &  &  & $\underset{\left( 1.000\right) }{32}$ &  & $\underset{\left( 1.000\right) }{19}$ &  & $\underset{\left( 1.000\right) }{12}$ &  & $\underset{\left( 1.000\right) }{13}$ &  & $\underset{\left(
1.000\right) }{12}$ &  & $\underset{\left( 1.000\right) }{32}$ &  & $\underset{\left( 1.000\right) }{18}$ &  & $\underset{\left( 1.000\right) }{10}$ \\ 
&  &  &  &  &  &  &  &  &  &  &  &  &  &  &  &  &  &  &  &  &  \\ 
\cline{2-22}
&  &  &  &  &  &  &  &  &  &  &  &  &  &  &  &  &  &  &  &  &  \\ 
&  &  &  & $100$ &  &  & $\underset{\left( 0.994\right) }{29}$ &  & $\underset{\left( 0.993\right) }{23}$ &  & $\underset{\left( 0.991\right) }{20}$ &  & $\underset{\left( 0.989\right) }{24}$ &  & $\underset{\left(
0.990\right) }{21}$ &  & $\underset{\left( 0.996\right) }{30}$ &  & $\underset{\left( 0.995\right) }{23}$ &  & $\underset{\left( 0.993\right) }{19}$ \\ 
& $\underset{\left( \beta _{0}=1\right) }{\text{\textbf{Case III}}}$ &  & $m^{\ast }$ & $200$ &  &  & $\underset{\left( 1.000\right) }{37}$ &  & $\underset{\left( 1.000\right) }{26}$ &  & $\underset{\left( 1.000\right) }{19}$ &  & $\underset{\left( 1.000\right) }{24}$ &  & $\underset{\left(
1.000\right) }{21}$ &  & $\underset{\left( 1.000\right) }{37}$ &  & $\underset{\left( 1.000\right) }{26}$ &  & $\underset{\left( 1.000\right) }{19}$ \\ 
&  &  &  & $400$ &  &  & $\underset{\left( 1.000\right) }{42}$ &  & $\underset{\left( 1.000\right) }{29}$ &  & $\underset{\left( 1.000\right) }{20}$ &  & $\underset{\left( 1.000\right) }{25}$ &  & $\underset{\left(
1.000\right) }{21}$ &  & $\underset{\left( 1.000\right) }{42}$ &  & $\underset{\left( 1.000\right) }{29}$ &  & $\underset{\left( 1.000\right) }{20}$ \\ 
&  &  &  & $800$ &  &  & $\underset{\left( 1.000\right) }{47}$ &  & $\underset{\left( 1.000\right) }{31}$ &  & $\underset{\left( 1.000\right) }{21}$ &  & $\underset{\left( 1.000\right) }{25}$ &  & $\underset{\left(
1.000\right) }{22}$ &  & $\underset{\left( 1.000\right) }{47}$ &  & $\underset{\left( 1.000\right) }{29}$ &  & $\underset{\left( 1.000\right) }{21}$ \\ 
&  &  &  &  &  &  &  &  &  &  &  &  &  &  &  &  &  &  &  &  &  \\ 
\hline\hline
\end{tabular}

}
\par
{\scriptsize {\ 
\begin{tablenotes}
      \tiny
            \item For each DGP, we report the \textit{median} detection delay for only the cases where a changepoint is detected (thus leaving out the cases where no changepoint is detected). Numbers in round brackets represent the empirical rejection frequencies.
            
\end{tablenotes}
} }
\end{table*}

Under stationarity, using $\psi =0.45$ yields the best results in terms of
delay and power, with little difference between the CUSUM and the Page-CUSUM
(with the latter occasionally delivering higher power, if not shorter
detection delays); results are anyway satisfactory also when using $\psi
=1/2 $ with $\widehat{c}_{\alpha ,0.5}$, which is always either the best or
the second best in terms of detection delay. The same results are observed
in the explosive and in the STUR cases. The latter is particularly
remarkable: using $\psi =1/2$ with $\widehat{c}_{\alpha ,0.5}$ delivers
essentialy the same performance as using $\psi =0.45$; however, according to
Table \ref{tab:ERF3}, the latter choice yields an oversized procedure, thus
making its performance under the alternative less reliable. Hence, Table \ref%
{tab:Power1} essentially indicates that using $\psi =1/2$ with $\widehat{c}%
_{\alpha ,0.5}$ yields the best results in terms of timeliness of detection
and power. As a final remark, detection delays worsen as $m^{\ast }$
increases; this is a natural phenomenon, since critical values (see e.g.
Corollary \ref{th-2}) increase with $m^{\ast }$, thus making detection, 
\textit{ceteris paribus}, more infrequent. In Table \ref{tab:Power2} in the
Supplement, we report median delays and empirical rejection frequencies in
the presence of covariates in the basic RCA specification. Our guidelines
based on size control suggest using $\psi =1/2$ with $c_{\alpha ,0.5}$ in
this case; median delays appear very satisfactory, although the cases $\psi
=1/2$ with $\widehat{c}_{\alpha ,0.5}$ and especially $\psi =0.45$ perform
better in several cases.

\medskip

In conclusion, the results in this section show that online changepoint
detection in an RCA\ model based on weighted CUSUM statistics seems to work
very well both under stationarity and nonstationarity. Broadly speaking, in
the latter case using the \textit{standardised} CUSUM (i.e., $\psi =1/2$) is
the preferred choice; in the former case, choosing $\psi $ close to (but
smaller than) $1/2$, appears to yield the best compromise between size
control and timely detection.

\section{Empirical applications: finding bubbles in RCA\ models\label%
{empirics}}

We validate our approach through two empirical applications, both involving
(potentially) nonstationary data. In Section \ref{covid}, we consider a
\textquotedblleft pure\textquotedblright\ RCA model with no covariates, and
use it to detect daily Covid-19 hospitalisations; in Section \ref{housing},
we consider and RCA model with covariates, and use it to detect changes in
the dynamics of house prices.

\subsection{Sequential monitoring of Covid-19 hospitalisations\label{covid}}

We consider daily data on Covid-19 hospitalisations. There is a huge
literature on the application of time series methods to the reduced form of
epidemiological models; \citet{shtatland}, \textit{inter alia}, advocate
using a low-order autoregression as an approximation of the popular SIR\
model, especially as a methodology for the early detection of outbreaks; and %
\citet{jeremy} apply the RCA model to a similar dataset, detecting several
changepoints. In this context, it is important to check whether the
observations change from an explosive to a stationary regime (meaning that
the epidemic is slowing down), or vice versa whether they change from a
stationary to an explosive regime (corresponding to a surge in the
epidemic). Our dataset consists of daily data for England recorded between
March 19th, 2020, and January 30th, 2021, i.e. spanning the period from
approximately the first lockdown (which was announced on March 23rd, 2020),
until after the third and last lockdown of January 6th, 2021. We transform
the series into logs (plus one since on some days the number of
hospitalisation is zero); no further pre-processing is applied.\footnote{%
The data are available from
\par
\url{https://ourworldindata.org/grapher/uk-daily-covid-admissions?tab=chart%
\&stackMode=absolute\&time=2020-03-29..latest\&region=World}
\par
See also %
\url{https://www.instituteforgovernment.org.uk/sites/default/files/2022-12/timeline-coronavirus-lockdown-december-2021.pdf}%
, for a timeline of the UK Government decisions on lockdown and closures.}

We use a \textquotedblleft pure\textquotedblright\ RCA\ specification with
no covariates. The purpose of our exercise is to ascertain whether public
health authorities, back in 2020, could have benefited from the use of a
sequential monitoring procedure to flag changes in the dynamics of daily
hospitalisations, thus informing their decisions. We have used $\psi =1/2$
in $g_{m,\psi }\left( k\right) $ in (\ref{boundary}); based on the empirical
rejection frequencies reported in Tables \ref{tab:ERF1}-\ref{tab:ERF3}, we
use the critical values $\widehat{c}_{\alpha ,0.5}$ defined in (\ref%
{vostr-cv});\footnote{%
Results with different values of $\psi $ - and with $\psi =1/2$ and
asymptotic critical values - are available upon request.} in all scenarios
we set the length of the monitoring horizon always equal to the training
sample size, i.e. $m=m^{\ast }$, which, according to the results in Tables %
\ref{tab:ERF1}-\ref{tab:ERF3}, always ensures size control.

\medskip

\begin{table*}[h]
\caption{Online changepoint detection for Covid-19 daily hospitalisation -
England data.}
\label{tab:TabCovidM}
\par
\centering
\par
{\scriptsize {\ 
\begin{tabular}{lllll}
\hline\hline
&  &  &  &  \\ 
&  & \multicolumn{1}{c}{Changepoint 1} & \multicolumn{1}{c}{Changepoint 2} & 
\multicolumn{1}{c}{Changepoint 3} \\ 
&  & \multicolumn{1}{c}{} & \multicolumn{1}{c}{} & \multicolumn{1}{c}{} \\ 
&  & \multicolumn{1}{c}{Aug 28th, 2020} & \multicolumn{1}{c}{Nov 4th, 2020}
& \multicolumn{1}{c}{Jan 27th, 2021} \\ 
&  & \multicolumn{1}{c}{} & \multicolumn{1}{c}{} & \multicolumn{1}{c}{} \\ 
&  & \multicolumn{1}{c}{$\underset{\left[ \text{Apr 11th, 2020 - Aug 15th,
2020}\right] }{\widehat{\beta }=0.995}$} & \multicolumn{1}{c}{$\underset{%
\left[ \text{Aug 29th, 2020 - Oct 29th, 2020}\right] }{\widehat{\beta }=1.010%
}$} & \multicolumn{1}{c}{$\underset{\left[ \text{Nov 5th, 2020 - Dec 31st,
2020}\right] }{\widehat{\beta }=1.002}$} \\ 
&  &  &  &  \\ \hline\hline
\end{tabular}
}}
\par
{\scriptsize 
\begin{tablenotes}
      \tiny
            \item The series ends at 30 January 2021. We use the logs of the original data (plus one, given that, in some days, hospitalisations are equal to zero): no further transformations are used.
        
            \item For each changepoint, we report the left WLS estimates of $\beta_0$ - i.e., the value of $\beta_0$ \textit{prior} to the breakdate; we report the sample on which estimation was performed in square brackets, based on the changepoints identified using the test by \citet{horvath2022changepoint}.
            
\end{tablenotes}
}
\end{table*}

By way of preliminary analysis (and also to assess how our online detection
methodology fares), in Section \ref{further-covid} in the Supplement we
apply the test by \citet{horvath2022changepoint} to the whole sample (Table %
\ref{tab:TabCovid}). In particular, after the closure of the education and
hospitality sectors in all UK nations announced on March 23rd, 2020 (and
implemented 3 days later), a changepoint is found at April 10th, 2020,
indicating that the lockdown had started to \textquotedblleft
bite\textquotedblright\ after that date; subsequently, another changepoint
is found towards the end of August, which can be naturally interpreted as
the beginning of the second wave in the UK, after an increase in travelling
during the holiday season. Hence, we use, as training sample, the period
between April 11th, 2020 and August 15th, 2020, with $m=127$. We find
evidence of a changepoint on August 28th, 2020. As can be seen in Table \ref%
{tab:TabCovid} in the Supplement, the ex-post test by %
\citet{horvath2022changepoint} finds a break at August 26th. Clearly, a
direct comparison between the two tests is not meaningful, since ex-post
detection uses the information contained in the full sample, whereas online
detection cannot make use of it - thus putting the latter at a disadvantage
compared to the former. Still, there is (only) a two days delay between the
two procedures. Moreover, the August changepoint was officially acknowledged
by the PM on September 18th, 2020; hence, our sequential monitoring
procedure could have brought forward public health decisions. The value of $%
\beta $ after the break is above $1$, indicating an explosive dynamics in
Covid-19 hospitalisations. Table \ref{tab:TabCovid} in the Supplement
indicates that there was a further break on October 29th, 2020,
corresponding to the lockdown at the end of that month. Hence, we restart
our monitoring procedure, and we use the training sample August 29th, 2020
till October 29th, 2020, with $m=61$; under an explosive regime, this should
suffice to ensure size control and short detection delay, according to the
results in Tables \ref{tab:ERF3} and \ref{tab:Power1}. A changepoint is
flagged on November 4th, 2020, with a decline in $\beta $ indicating the
effect of the closures and also of the growing concerns about a second wave.
The break is detected very close to (and indeed \textit{before}) the
lockdown, suggesting that lockdowns tended to occur when a turning point in
hospitalisations had already occurred, or was \textquotedblleft in the
making\textquotedblright . Such evidence can be read in conjunction with the
results in Table \ref{tab:TabCovid} in the Supplement, and also with %
\citet{wood2021}, who, albeit with a different methodology and focus, finds
similar results for hospital deaths. Finally, we carry out monitoring using
a training sample between November 5th, 2020, till December 31st, 2020, thus
having $m=58$. We find a changepoint on January 27th, 2021, i.e. two weeks
later the break found with the ex-post test by \citet{horvath2022changepoint}%
, and three weeks later after the national lockdown announced on January
6th, 2021.

In conclusion, the empirical evidence presented above shows that our
RCA-based approach to online changepoint detection is suitable to detect the
onset (and the receding) of a pandemic with short delays, thus being a
recommended item in the toolbox of public health decision-makers.

\subsection{Sequential monitoring of housing prices\label{housing}}

We apply our sequential detection procedures to the online detection of a
change in housing prices in Los Angeles at the end of the first decade of
the century. \citet{horvath2022changepoint} apply an RCA-based, ex-post
changepoint test, and find evidence of such a bubble starting around
February 4th, 2009, when, after a period of \textquotedblleft hard
landing\textquotedblright , prices stabilised;\footnote{%
See Section \ref{further-housing} for a plot of the data.} we also refer to
a related paper by \citet{zenhya}, who use monthly data. We use Los Angeles
as a case study to check how timely online detection is, and also to assess
the robustness of our results to the choice of the training and monitoring
sample sizes, and the benefit of adding covariates. In our application,
following \citet{horvath2022changepoint}, we use (logs of) daily housing
prices.\footnote{%
We use the daily data constructed by \citet{bollerslev}, and we refer to
that paper for a description of the datasets.}

All our monitoring exercises start at January 15th, 2009; %
\citet{horvath2022changepoint} find no changepoints between May 4th, 2006,
and February 3rd, 2009, which entails that the non-contamination assumption
during the training period is satisfied in all cases considered. We monitor
for changepoints using several alternative models. In addition to the basic
RCA specification, with no covariates, we also use different combinations of
regressors, including: two variables that are closely related to the
risk-free interest rate, taking into account the opportunity cost of capital
(the Moody's Seasoned Aaa Corporate Bond Yield, and the 10 Year US Treasury
Constant Maturity Rate); a measure of volatility (namely, the VXO volatility
index); and a measure of real economic activity (we use the
Lewis-Mertens-Stock Weekly Economic Indicator, WEI\footnote{%
See \citet{lewis2022measuring} for a thorough description.}). All regressors
are taken from the FRED St Louis website. The risk-free interest rate
proxies and the volatility measure are transformed into logs. Applying
standard unit root tests to these explanatory variables, we find
overwhelming evidence of a unit root in all of them during the period March
28th, 2008, corresponding to the earliest starting point for the training
sample, and October 30th, 2009, corresponding to the latest ending point for
the monitoring horizon (see Section \ref{further-housing} in the Supplement
for details); hence, we employ their first differences. As far as the WEI is
concerned, it comes at a weekly frequency, and we use its weekly value for
each day by way of disaggregation; other disaggregation methods could be
employed following the Mixed-Data Sampling literature (MIDAS; see e.g. %
\citealp{ghysels2018mixed}), e.g. using a weighted average interpolation,
but we found that these did not make virtually any difference in our results
(in fact, marginally worsening the detection delay). No further
pre-processing is carried out. As far as implementation is concerned, we
have used $\psi =1/2$ in $g_{m,\psi }^{\left( x\right) }\left( k\right) $ in
(\ref{boundary-x}); in the case considered, the data are \textquotedblleft
nearly\textquotedblright\ non-stationary during the training period, and
therefore, based on Table \ref{tab:ERF3c}, we use the asymptotic critical
values $c_{\alpha ,0.5}$.

\begin{table*}[t]
\caption{Online changepoint detection for Los Angeles daily housing prices.}
\label{tab:TabHousing}
\par
\centering
\par
\resizebox{\textwidth}{!}{

\begin{tabular}{cccccccccccccc}
\hline\hline
&  &  &  &  &  &  &  &  &  &  &  &  &  \\ 
\multicolumn{6}{c}{Model: $y_{i}=\left( \beta _{i}+\epsilon _{i,1}\right)
y_{i-1}+\epsilon _{i,2}$} &  & \multicolumn{1}{|c}{} & \multicolumn{6}{c}{
Model: $y_{i}=\left( \beta _{i}+\epsilon _{i,1}\right) y_{i-1}+\lambda
_{1}x_{1,i}+\lambda _{2}x_{2,i}+\epsilon _{i,2}$} \\ 
&  &  &  &  &  &  & \multicolumn{1}{|c}{} &  &  &  &  &  &  \\ 
& $m^{\ast }$ &  & $100$ &  & $200$ &  & \multicolumn{1}{|c}{} &  & $m^{\ast
}$ &  & $100$ &  & $200$ \\ 
$m$ &  &  &  &  &  &  & \multicolumn{1}{|c}{} & $m$ &  &  &  &  &  \\ 
&  &  &  &  &  &  & \multicolumn{1}{|c}{} &  &  &  &  &  &  \\ 
$100$ &  &  & $\underset{\left[ \text{no changepoint found}\right] }{\text{Jun 9th, 2009}}$ &  & Jun 15th, 2009 &  & \multicolumn{1}{|c}{} & $100$ &  & 
& $\text{Jun 2nd, 2009}$ &  & $\text{Jun 2nd, 2009}$ \\ 
&  &  &  &  &  &  & \multicolumn{1}{|c}{} &  &  &  &  &  &  \\ 
$200$ &  &  & $\underset{\left[ \text{no changepoint found}\right] }{\text{Jun 9th, 2009}}$ &  & Jun 15th, 2009 &  & \multicolumn{1}{|c}{} & $200$ &  & 
& J$\text{un 3rd, 2009}$ &  & Jun 10th, 2009 \\ 
&  &  &  &  &  &  & \multicolumn{1}{|c}{} &  &  &  &  &  &  \\ \hline\hline
&  &  &  &  &  &  &  &  &  &  &  &  &  \\ 
\multicolumn{6}{c}{Model: $y_{i}=\left( \beta _{i}+\epsilon _{i,1}\right)
y_{i-1}+\lambda _{1}x_{1,i}+\lambda _{2}x_{2,i}+\lambda _{3}x_{3,i}+\epsilon
_{i,2}$} &  & \multicolumn{1}{|c}{} & \multicolumn{6}{c}{Model: $y_{i}=\left( \beta _{i}+\epsilon _{i,1}\right) y_{i-1}+\lambda
_{1}x_{1,i}+\lambda _{2}x_{2,i}+\lambda _{3}x_{3,i}+\lambda
_{4}x_{4,i}+\epsilon _{i,2}$} \\ 
&  &  &  &  &  &  & \multicolumn{1}{|c}{} &  &  &  &  &  &  \\ 
& $m^{\ast }$ &  & $100$ &  & $200$ &  & \multicolumn{1}{|c}{} &  & $m^{\ast
}$ &  & $100$ &  & $200$ \\ 
$m$ &  &  &  &  &  &  & \multicolumn{1}{|c}{} & $m$ &  &  &  &  &  \\ 
&  &  &  &  &  &  & \multicolumn{1}{|c}{} &  &  &  &  &  &  \\ 
$100$ &  &  & $\text{Jun 2nd, 2009}$ &  & $\text{Jun 2nd, 2009}$ &  & 
\multicolumn{1}{|c}{} & $100$ &  &  & May 18th, 2009 &  & May 18th, 2009 \\ 
&  &  &  &  &  &  & \multicolumn{1}{|c}{} &  &  &  &  &  &  \\ 
$200$ &  &  & J$\text{un 3rd, 2009}$ &  & Jun 10th, 2009 &  & 
\multicolumn{1}{|c}{} & $200$ &  &  & $\text{Jun 2nd, 2009}$ &  & $\text{Jun
2nd, 2009}$ \\ 
&  &  &  &  &  &  & \multicolumn{1}{|c}{} &  &  &  &  &  &  \\ \hline\hline
\end{tabular}
}
\par
{\scriptsize {\ 
\begin{tablenotes}
      \tiny
            \item For each combination of $m$ and $m^\ast$, we report the estimated breakdate. For all combinations of $m$ and $m^{\ast}$, monitoring starts on January 15th, 2009. When $m=100$, the training sample covers the period August 20th, 2008, till January 14th, 2009; when $m=200$, the training sample covers the period March 28th, 2008, till January 14th, 2009. Similarly, when $m^{\ast}=100$, the monitoring horizon stops at June 9th, 2009; when $m^{\ast}=200$, the monitoring horizon stops at October 30th, 2009.
            \item We have used the following notation for the regressors: $x_{1,i}$ denotes the 10 Year US Treasury Constant Maturity Rate, $x_{2,i}$ denotes the Moody's Seasoned Aaa Corporate Bond Yield, $x_{3,i}$ is the VXO volatility index, and $x_{4,i}$ is the WEI.
            \item \citet{horvath2022changepoint} find evidence of a changepoint on February 3rd, 2009, applying ex-post changepoint detection to the period January 5th, 1995 to October 23rd, 2012. The deterministic part of the autoregressive coefficient, $\beta$, is found to be equal to $0.99931$ in the period before the changepoint, and $1.00007$ afterwards.
            
\end{tablenotes}
} }
\end{table*}

Results in Table \ref{tab:TabHousing} show that adding covariates can
potentially lead to meaningful improvements in terms of detection delay: a
measure of economic activity such as the WEI index, in particular, seems to
contain relevant information to model the dynamics of house prices. We also
note that detection delays seem to be around $4$ months, which may be a
consequence of the relatively small change in $\beta $ before and after the
change. Interestingly, results are unaffected by the size of the training
sample $m$, but they do differ if $m^{\ast }$ changes. This is a purely
mechanical effect, due to the increase in the asymptotic critical values as $%
m^{\ast }$ increases, as also noted in Section \ref{simulations}.

\section{Discussion and conclusions\label{conclusions}}

We propose a family of weighted CUSUM statistics for the online detection of
changepoints in a Random Coefficient Autoregressive model. Our statistics
are, in particular, based on the CUSUM\ process of the WLS\ residuals, and
we study both the standard CUSUM, and the so-called Page-CUSUM monitoring
schemes, which is designed to offer higher power/shorter detection delay. In
the case of the standard CUSUM, we also study the standardised version using 
$\psi =1/2$, obtaining, under the null, a Darling-Erd\H{o}s limit theorem.
In this case, seeing as the asymptotic critical values are a poor
approximation leading to an overly conservative procedure, we also propose
an approximation which works well in finite samples, avoiding
overrejections. As our simulations show, the use of weighted statistics is
particularly beneficial under the alternative, with detection delays
decreasing as $\psi $ increases. Whilst, for the ease of exposition, we
focus on the RCA\ case with no covariates, we also extend our theory to
include exogenous regressors, which are allowed to be weakly dependent
according to a very general definition of dependence. Simulations show that
our procedures broadly guarantee size control; indeed, our experiments
indicate that, for any given training sample size $m$ and monitoring horizon 
$m^{\ast }$, it is always possible to choose the appropriate weighing scheme
to ensure the best balance between size control and timely detection. This
is reinforced by our empirical illustrations, showing that when our
methodology is applied to the online detection of epidemiological and
housing data, it manages to find breaks very quickly, even when compared
against ex-post detection methodologies.

Importantly, building on the well-known fact that, in an RCA\ model, WLS\
inference on the deterministic part of the autoregressive parameter is
always Gaussian, irrespective of whether the observations form a stationary
or nonstationary sequence, our monitoring procedures can be applied to
virtually any type of data: stationary, or with an explosive dynamics, or on
the boundary between the two regimes, with no modifications required. Hence,
our methodology is particularly suited, as a leading example, to the
detection of both the onset, or the collapse, of a bubble (when using
financial data), or of a pandemic (when using epidemiological data). Other
cases also can be considered, e.g. monitoring for changes in the persistence
of a stationary series such as inflation, and the extension of the basic RCA
specification to include covariates should enhance the applicability of our
methodology.

\begin{adjustwidth}{-5pt}{-5pt}

{\footnotesize {\ 
\bibliographystyle{chicago}
\bibliography{LTbiblio}
} }

\end{adjustwidth}

\newpage

\clearpage
\renewcommand*{\thesection}{\Alph{section}}

\setcounter{section}{0} \setcounter{subsection}{-1} %
\setcounter{subsubsection}{-1} \setcounter{equation}{0} \setcounter{lemma}{0}
\setcounter{theorem}{0} \renewcommand{\theassumption}{A.\arabic{assumption}} 
\renewcommand{\thetheorem}{A.\arabic{theorem}} \renewcommand{\thelemma}{A.%
\arabic{lemma}} \renewcommand{\thecorollary}{A.\arabic{corollary}} %
\renewcommand{\theequation}{A.\arabic{equation}}

\section{Further Monte Carlo evidence and guidelines\label{furtherMC}}

\begin{table*}[h!]
\caption{{\protect\footnotesize {Empirical rejection frequencies under the
null of no changepoint with covariates - Case I, $\protect\beta_0=0.5$}}}
\label{tab:ERF1c}\centering
{\footnotesize {\ }}
\par
{\footnotesize {\ }}
\par
{\scriptsize {\ }}
\par
{\scriptsize 
\begin{tabular}{lllllllllllllllllllll}
\hline\hline
&  &  &  &  &  &  &  &  &  &  &  &  &  &  &  &  &  &  &  &  \\ 
&  &  &  &  &  & \multicolumn{5}{c}{Weighted CUSUM} & \multicolumn{1}{c}{} & 
\multicolumn{3}{c}{Standardised CUSUM} & \multicolumn{1}{c}{} & 
\multicolumn{5}{c}{Weighted Page-CUSUM} \\ 
&  &  &  & $\psi $ & \multicolumn{1}{c}{} & \multicolumn{1}{c}{$0$} & 
\multicolumn{1}{c}{} & \multicolumn{1}{c}{$0.25$} & \multicolumn{1}{c}{} & 
\multicolumn{1}{c}{$0.45$} & \multicolumn{1}{c}{} & \multicolumn{1}{c}{} & 
\multicolumn{1}{c}{$0.5$} & \multicolumn{1}{c}{} & \multicolumn{1}{c}{} & 
\multicolumn{1}{c}{$0$} & \multicolumn{1}{c}{} & \multicolumn{1}{c}{$0.25$}
& \multicolumn{1}{c}{} & \multicolumn{1}{c}{$0.45$} \\ 
&  &  &  &  & \multicolumn{1}{c}{} & \multicolumn{1}{c}{} & 
\multicolumn{1}{c}{} & \multicolumn{1}{c}{} & \multicolumn{1}{c}{} & 
\multicolumn{1}{c}{} & \multicolumn{1}{c}{} & \multicolumn{1}{c}{$c_{\alpha
,0.5}$} & \multicolumn{1}{c}{} & \multicolumn{1}{c}{$\widehat{c}_{\alpha
,0.5}$} & \multicolumn{1}{c}{} & \multicolumn{1}{c}{} & \multicolumn{1}{c}{}
& \multicolumn{1}{c}{} & \multicolumn{1}{c}{} & \multicolumn{1}{c}{} \\ 
&  &  &  &  & \multicolumn{1}{c}{} & \multicolumn{1}{c}{} & 
\multicolumn{1}{c}{} & \multicolumn{1}{c}{} & \multicolumn{1}{c}{} & 
\multicolumn{1}{c}{} & \multicolumn{1}{c}{} & \multicolumn{1}{c}{} & 
\multicolumn{1}{c}{} & \multicolumn{1}{c}{} & \multicolumn{1}{c}{} & 
\multicolumn{1}{c}{} & \multicolumn{1}{c}{} & \multicolumn{1}{c}{} & 
\multicolumn{1}{c}{} & \multicolumn{1}{c}{} \\ 
\multicolumn{1}{c}{} & \multicolumn{1}{c}{$m$} & \multicolumn{1}{c}{} & 
\multicolumn{1}{c}{$m^{\ast }$} & \multicolumn{1}{c}{} & \multicolumn{1}{c}{}
& \multicolumn{1}{c}{} & \multicolumn{1}{c}{} & \multicolumn{1}{c}{} & 
\multicolumn{1}{c}{} & \multicolumn{1}{c}{} & \multicolumn{1}{c}{} & 
\multicolumn{1}{c}{} & \multicolumn{1}{c}{} & \multicolumn{1}{c}{} & 
\multicolumn{1}{c}{} & \multicolumn{1}{c}{} & \multicolumn{1}{c}{} & 
\multicolumn{1}{c}{} & \multicolumn{1}{c}{} & \multicolumn{1}{c}{} \\ 
\multicolumn{1}{c}{} & \multicolumn{1}{c}{} & \multicolumn{1}{c}{} & 
\multicolumn{1}{c}{} & \multicolumn{1}{c}{} & \multicolumn{1}{c}{} & 
\multicolumn{1}{c}{} & \multicolumn{1}{c}{} & \multicolumn{1}{c}{} & 
\multicolumn{1}{c}{} & \multicolumn{1}{c}{} & \multicolumn{1}{c}{} & 
\multicolumn{1}{c}{} & \multicolumn{1}{c}{} & \multicolumn{1}{c}{} & 
\multicolumn{1}{c}{} & \multicolumn{1}{c}{} & \multicolumn{1}{c}{} & 
\multicolumn{1}{c}{} & \multicolumn{1}{c}{} & \multicolumn{1}{c}{} \\ 
\multicolumn{1}{c}{} & \multicolumn{1}{c}{} & \multicolumn{1}{c}{} & 
\multicolumn{1}{c}{$25$} & \multicolumn{1}{c}{} & \multicolumn{1}{c}{} & 
\multicolumn{1}{c}{$0.060$} & \multicolumn{1}{c}{} & \multicolumn{1}{c}{$%
0.062$} & \multicolumn{1}{c}{} & \multicolumn{1}{c}{$0.046$} & 
\multicolumn{1}{c}{} & \multicolumn{1}{c}{$0.030$} & \multicolumn{1}{c}{} & 
\multicolumn{1}{c}{$0.052$} & \multicolumn{1}{c}{} & \multicolumn{1}{c}{$%
0.051$} & \multicolumn{1}{c}{} & \multicolumn{1}{c}{$0.059$} & 
\multicolumn{1}{c}{} & \multicolumn{1}{c}{$0.046$} \\ 
\multicolumn{1}{c}{} & \multicolumn{1}{c}{$50$} & \multicolumn{1}{c}{} & 
\multicolumn{1}{c}{$50$} & \multicolumn{1}{c}{} & \multicolumn{1}{c}{} & 
\multicolumn{1}{c}{$0.068$} & \multicolumn{1}{c}{} & \multicolumn{1}{c}{$%
0.078$} & \multicolumn{1}{c}{} & \multicolumn{1}{c}{$0.071$} & 
\multicolumn{1}{c}{} & \multicolumn{1}{c}{$0.037$} & \multicolumn{1}{c}{} & 
\multicolumn{1}{c}{$0.066$} & \multicolumn{1}{c}{} & \multicolumn{1}{c}{$%
0.068$} & \multicolumn{1}{c}{} & \multicolumn{1}{c}{$0.067$} & 
\multicolumn{1}{c}{} & \multicolumn{1}{c}{$0.066$} \\ 
\multicolumn{1}{c}{} & \multicolumn{1}{c}{} & \multicolumn{1}{c}{} & 
\multicolumn{1}{c}{$100$} & \multicolumn{1}{c}{} & \multicolumn{1}{c}{} & 
\multicolumn{1}{c}{$0.072$} & \multicolumn{1}{c}{} & \multicolumn{1}{c}{$%
0.078$} & \multicolumn{1}{c}{} & \multicolumn{1}{c}{$0.077$} & 
\multicolumn{1}{c}{} & \multicolumn{1}{c}{$0.043$} & \multicolumn{1}{c}{} & 
\multicolumn{1}{c}{$0.064$} & \multicolumn{1}{c}{} & \multicolumn{1}{c}{$%
0.070$} & \multicolumn{1}{c}{} & \multicolumn{1}{c}{$0.072$} & 
\multicolumn{1}{c}{} & \multicolumn{1}{c}{$0.072$} \\ 
\multicolumn{1}{c}{} & \multicolumn{1}{c}{} & \multicolumn{1}{c}{} & 
\multicolumn{1}{c}{$200$} & \multicolumn{1}{c}{} & \multicolumn{1}{c}{} & 
\multicolumn{1}{c}{$0.087$} & \multicolumn{1}{c}{} & \multicolumn{1}{c}{$%
0.082$} & \multicolumn{1}{c}{} & \multicolumn{1}{c}{$0.082$} & 
\multicolumn{1}{c}{} & \multicolumn{1}{c}{$0.044$} & \multicolumn{1}{c}{} & 
\multicolumn{1}{c}{$0.066$} & \multicolumn{1}{c}{} & \multicolumn{1}{c}{$%
0.081$} & \multicolumn{1}{c}{} & \multicolumn{1}{c}{$0.102$} & 
\multicolumn{1}{c}{} & \multicolumn{1}{c}{$0.093$} \\ 
\multicolumn{1}{c}{} & \multicolumn{1}{c}{} & \multicolumn{1}{c}{} & 
\multicolumn{1}{c}{} & \multicolumn{1}{c}{} & \multicolumn{1}{c}{} & 
\multicolumn{1}{c}{} & \multicolumn{1}{c}{} & \multicolumn{1}{c}{} & 
\multicolumn{1}{c}{} & \multicolumn{1}{c}{} & \multicolumn{1}{c}{} & 
\multicolumn{1}{c}{} & \multicolumn{1}{c}{} & \multicolumn{1}{c}{} & 
\multicolumn{1}{c}{} & \multicolumn{1}{c}{} & \multicolumn{1}{c}{} & 
\multicolumn{1}{c}{} & \multicolumn{1}{c}{} & \multicolumn{1}{c}{} \\ 
\multicolumn{1}{c}{} & \multicolumn{1}{c}{} & \multicolumn{1}{c}{} & 
\multicolumn{1}{c}{} & \multicolumn{1}{c}{} & \multicolumn{1}{c}{} & 
\multicolumn{1}{c}{} & \multicolumn{1}{c}{} & \multicolumn{1}{c}{} & 
\multicolumn{1}{c}{} & \multicolumn{1}{c}{} & \multicolumn{1}{c}{} & 
\multicolumn{1}{c}{} & \multicolumn{1}{c}{} & \multicolumn{1}{c}{} & 
\multicolumn{1}{c}{} & \multicolumn{1}{c}{} & \multicolumn{1}{c}{} & 
\multicolumn{1}{c}{} & \multicolumn{1}{c}{} & \multicolumn{1}{c}{} \\ 
\multicolumn{1}{c}{} & \multicolumn{1}{c}{} & \multicolumn{1}{c}{} & 
\multicolumn{1}{c}{$50$} & \multicolumn{1}{c}{} & \multicolumn{1}{c}{} & 
\multicolumn{1}{c}{$0.075$} & \multicolumn{1}{c}{} & \multicolumn{1}{c}{$%
0.069$} & \multicolumn{1}{c}{} & \multicolumn{1}{c}{$0.075$} & 
\multicolumn{1}{c}{} & \multicolumn{1}{c}{$0.030$} & \multicolumn{1}{c}{} & 
\multicolumn{1}{c}{$0.057$} & \multicolumn{1}{c}{} & \multicolumn{1}{c}{$%
0.063$} & \multicolumn{1}{c}{} & \multicolumn{1}{c}{$0.067$} & 
\multicolumn{1}{c}{} & \multicolumn{1}{c}{$0.059$} \\ 
\multicolumn{1}{c}{} & \multicolumn{1}{c}{$100$} & \multicolumn{1}{c}{} & 
\multicolumn{1}{c}{$100$} & \multicolumn{1}{c}{} & \multicolumn{1}{c}{} & 
\multicolumn{1}{c}{$0.060$} & \multicolumn{1}{c}{} & \multicolumn{1}{c}{$%
0.061$} & \multicolumn{1}{c}{} & \multicolumn{1}{c}{$0.050$} & 
\multicolumn{1}{c}{} & \multicolumn{1}{c}{$0.034$} & \multicolumn{1}{c}{} & 
\multicolumn{1}{c}{$0.050$} & \multicolumn{1}{c}{} & \multicolumn{1}{c}{$%
0.060$} & \multicolumn{1}{c}{} & \multicolumn{1}{c}{$0.062$} & 
\multicolumn{1}{c}{} & \multicolumn{1}{c}{$0.058$} \\ 
\multicolumn{1}{c}{} & \multicolumn{1}{c}{} & \multicolumn{1}{c}{} & 
\multicolumn{1}{c}{$200$} & \multicolumn{1}{c}{} & \multicolumn{1}{c}{} & 
\multicolumn{1}{c}{$0.075$} & \multicolumn{1}{c}{} & \multicolumn{1}{c}{$%
0.085$} & \multicolumn{1}{c}{} & \multicolumn{1}{c}{$0.078$} & 
\multicolumn{1}{c}{} & \multicolumn{1}{c}{$0.026$} & \multicolumn{1}{c}{} & 
\multicolumn{1}{c}{$0.057$} & \multicolumn{1}{c}{} & \multicolumn{1}{c}{$%
0.074$} & \multicolumn{1}{c}{} & \multicolumn{1}{c}{$0.088$} & 
\multicolumn{1}{c}{} & \multicolumn{1}{c}{$0.078$} \\ 
\multicolumn{1}{c}{} & \multicolumn{1}{c}{} & \multicolumn{1}{c}{} & 
\multicolumn{1}{c}{$400$} & \multicolumn{1}{c}{} & \multicolumn{1}{c}{} & 
\multicolumn{1}{c}{$0.062$} & \multicolumn{1}{c}{} & \multicolumn{1}{c}{$%
0.066$} & \multicolumn{1}{c}{} & \multicolumn{1}{c}{$0.065$} & 
\multicolumn{1}{c}{} & \multicolumn{1}{c}{$0.022$} & \multicolumn{1}{c}{} & 
\multicolumn{1}{c}{$0.047$} & \multicolumn{1}{c}{} & \multicolumn{1}{c}{$%
0.056$} & \multicolumn{1}{c}{} & \multicolumn{1}{c}{$0.065$} & 
\multicolumn{1}{c}{} & \multicolumn{1}{c}{$0.062$} \\ 
\multicolumn{1}{c}{} & \multicolumn{1}{c}{} & \multicolumn{1}{c}{} & 
\multicolumn{1}{c}{} & \multicolumn{1}{c}{} & \multicolumn{1}{c}{} & 
\multicolumn{1}{c}{} & \multicolumn{1}{c}{} & \multicolumn{1}{c}{} & 
\multicolumn{1}{c}{} & \multicolumn{1}{c}{} & \multicolumn{1}{c}{} & 
\multicolumn{1}{c}{} & \multicolumn{1}{c}{} & \multicolumn{1}{c}{} & 
\multicolumn{1}{c}{} & \multicolumn{1}{c}{} & \multicolumn{1}{c}{} & 
\multicolumn{1}{c}{} & \multicolumn{1}{c}{} & \multicolumn{1}{c}{} \\ 
\multicolumn{1}{c}{} & \multicolumn{1}{c}{} & \multicolumn{1}{c}{} & 
\multicolumn{1}{c}{} & \multicolumn{1}{c}{} & \multicolumn{1}{c}{} & 
\multicolumn{1}{c}{} & \multicolumn{1}{c}{} & \multicolumn{1}{c}{} & 
\multicolumn{1}{c}{} & \multicolumn{1}{c}{} & \multicolumn{1}{c}{} & 
\multicolumn{1}{c}{} & \multicolumn{1}{c}{} & \multicolumn{1}{c}{} & 
\multicolumn{1}{c}{} & \multicolumn{1}{c}{} & \multicolumn{1}{c}{} & 
\multicolumn{1}{c}{} & \multicolumn{1}{c}{} & \multicolumn{1}{c}{} \\ 
\multicolumn{1}{c}{} & \multicolumn{1}{c}{} & \multicolumn{1}{c}{} & 
\multicolumn{1}{c}{$100$} & \multicolumn{1}{c}{} & \multicolumn{1}{c}{} & 
\multicolumn{1}{c}{$0.051$} & \multicolumn{1}{c}{} & \multicolumn{1}{c}{$%
0.041$} & \multicolumn{1}{c}{} & \multicolumn{1}{c}{$0.055$} & 
\multicolumn{1}{c}{} & \multicolumn{1}{c}{$0.029$} & \multicolumn{1}{c}{} & 
\multicolumn{1}{c}{$0.046$} & \multicolumn{1}{c}{} & \multicolumn{1}{c}{$%
0.048$} & \multicolumn{1}{c}{} & \multicolumn{1}{c}{$0.047$} & 
\multicolumn{1}{c}{} & \multicolumn{1}{c}{$0.058$} \\ 
\multicolumn{1}{c}{} & \multicolumn{1}{c}{$200$} & \multicolumn{1}{c}{} & 
\multicolumn{1}{c}{$200$} & \multicolumn{1}{c}{} & \multicolumn{1}{c}{} & 
\multicolumn{1}{c}{$0.056$} & \multicolumn{1}{c}{} & \multicolumn{1}{c}{$%
0.047$} & \multicolumn{1}{c}{} & \multicolumn{1}{c}{$0.051$} & 
\multicolumn{1}{c}{} & \multicolumn{1}{c}{$0.030$} & \multicolumn{1}{c}{} & 
\multicolumn{1}{c}{$0.044$} & \multicolumn{1}{c}{} & \multicolumn{1}{c}{$%
0.058$} & \multicolumn{1}{c}{} & \multicolumn{1}{c}{$0.061$} & 
\multicolumn{1}{c}{} & \multicolumn{1}{c}{$0.066$} \\ 
\multicolumn{1}{c}{} & \multicolumn{1}{c}{} & \multicolumn{1}{c}{} & 
\multicolumn{1}{c}{$400$} & \multicolumn{1}{c}{} & \multicolumn{1}{c}{} & 
\multicolumn{1}{c}{$0.056$} & \multicolumn{1}{c}{} & \multicolumn{1}{c}{$%
0.062$} & \multicolumn{1}{c}{} & \multicolumn{1}{c}{$0.055$} & 
\multicolumn{1}{c}{} & \multicolumn{1}{c}{$0.033$} & \multicolumn{1}{c}{} & 
\multicolumn{1}{c}{$0.053$} & \multicolumn{1}{c}{} & \multicolumn{1}{c}{$%
0.049$} & \multicolumn{1}{c}{} & \multicolumn{1}{c}{$0.060$} & 
\multicolumn{1}{c}{} & \multicolumn{1}{c}{$0.065$} \\ 
\multicolumn{1}{c}{} & \multicolumn{1}{c}{} & \multicolumn{1}{c}{} & 
\multicolumn{1}{c}{$800$} & \multicolumn{1}{c}{} & \multicolumn{1}{c}{} & 
\multicolumn{1}{c}{$0.054$} & \multicolumn{1}{c}{} & \multicolumn{1}{c}{$%
0.064$} & \multicolumn{1}{c}{} & \multicolumn{1}{c}{$0.056$} & 
\multicolumn{1}{c}{} & \multicolumn{1}{c}{$0.025$} & \multicolumn{1}{c}{} & 
\multicolumn{1}{c}{$0.042$} & \multicolumn{1}{c}{} & \multicolumn{1}{c}{$%
0.049$} & \multicolumn{1}{c}{} & \multicolumn{1}{c}{$0.059$} & 
\multicolumn{1}{c}{} & \multicolumn{1}{c}{$0.057$} \\ 
\multicolumn{1}{c}{} & \multicolumn{1}{c}{} & \multicolumn{1}{c}{} & 
\multicolumn{1}{c}{} & \multicolumn{1}{c}{} & \multicolumn{1}{c}{} & 
\multicolumn{1}{c}{} & \multicolumn{1}{c}{} & \multicolumn{1}{c}{} & 
\multicolumn{1}{c}{} & \multicolumn{1}{c}{} & \multicolumn{1}{c}{} & 
\multicolumn{1}{c}{} & \multicolumn{1}{c}{} & \multicolumn{1}{c}{} &  &  & 
&  &  &  \\ \hline\hline
\end{tabular}
}
\par
{\scriptsize \ }
\par
{\scriptsize {\footnotesize \ } }
\par
{\scriptsize {\footnotesize 
\begin{tablenotes}
      \tiny
            \item 
            
\end{tablenotes}
} }
\end{table*}

\begin{table*}[h!]
\caption{{\protect\footnotesize {Empirical rejection frequencies under the
null of no changepoint with covariates - Case II, $\protect\beta_0=1.05$}}}
\label{tab:ERF2c}\centering
{\footnotesize {\ }}
\par
{\footnotesize {\ }}
\par
{\scriptsize {\ }}
\par
{\scriptsize 
\begin{tabular}{lllllllllllllllllllll}
\hline\hline
&  &  &  &  &  &  &  &  &  &  &  &  &  &  &  &  &  &  &  &  \\ 
&  &  &  &  &  & \multicolumn{5}{c}{Weighted CUSUM} & \multicolumn{1}{c}{} & 
\multicolumn{3}{c}{Standardised CUSUM} & \multicolumn{1}{c}{} & 
\multicolumn{5}{c}{Weighted Page-CUSUM} \\ 
&  &  &  & $\psi $ &  & \multicolumn{1}{c}{$0$} & \multicolumn{1}{c}{} & 
\multicolumn{1}{c}{$0.25$} & \multicolumn{1}{c}{} & \multicolumn{1}{c}{$0.45$%
} & \multicolumn{1}{c}{} & \multicolumn{1}{c}{} & \multicolumn{1}{c}{$0.5$}
& \multicolumn{1}{c}{} & \multicolumn{1}{c}{} & \multicolumn{1}{c}{$0$} & 
\multicolumn{1}{c}{} & \multicolumn{1}{c}{$0.25$} & \multicolumn{1}{c}{} & 
\multicolumn{1}{c}{$0.45$} \\ 
&  &  &  &  &  & \multicolumn{1}{c}{} & \multicolumn{1}{c}{} & 
\multicolumn{1}{c}{} & \multicolumn{1}{c}{} & \multicolumn{1}{c}{} & 
\multicolumn{1}{c}{} & \multicolumn{1}{c}{$c_{\alpha ,0.5}$} & 
\multicolumn{1}{c}{} & \multicolumn{1}{c}{$\widehat{c}_{\alpha ,0.5}$} & 
\multicolumn{1}{c}{} & \multicolumn{1}{c}{} & \multicolumn{1}{c}{} & 
\multicolumn{1}{c}{} & \multicolumn{1}{c}{} & \multicolumn{1}{c}{} \\ 
&  &  &  &  &  & \multicolumn{1}{c}{} & \multicolumn{1}{c}{} & 
\multicolumn{1}{c}{} & \multicolumn{1}{c}{} & \multicolumn{1}{c}{} & 
\multicolumn{1}{c}{} & \multicolumn{1}{c}{} & \multicolumn{1}{c}{} & 
\multicolumn{1}{c}{} & \multicolumn{1}{c}{} & \multicolumn{1}{c}{} & 
\multicolumn{1}{c}{} & \multicolumn{1}{c}{} & \multicolumn{1}{c}{} & 
\multicolumn{1}{c}{} \\ 
\multicolumn{1}{c}{} & \multicolumn{1}{c}{$m$} & \multicolumn{1}{c}{} & 
\multicolumn{1}{c}{$m^{\ast }$} & \multicolumn{1}{c}{} & \multicolumn{1}{c}{}
& \multicolumn{1}{c}{} & \multicolumn{1}{c}{} & \multicolumn{1}{c}{} & 
\multicolumn{1}{c}{} & \multicolumn{1}{c}{} & \multicolumn{1}{c}{} & 
\multicolumn{1}{c}{} & \multicolumn{1}{c}{} & \multicolumn{1}{c}{} & 
\multicolumn{1}{c}{} & \multicolumn{1}{c}{} & \multicolumn{1}{c}{} & 
\multicolumn{1}{c}{} & \multicolumn{1}{c}{} & \multicolumn{1}{c}{} \\ 
\multicolumn{1}{c}{} & \multicolumn{1}{c}{} & \multicolumn{1}{c}{} & 
\multicolumn{1}{c}{} & \multicolumn{1}{c}{} & \multicolumn{1}{c}{} & 
\multicolumn{1}{c}{} & \multicolumn{1}{c}{} & \multicolumn{1}{c}{} & 
\multicolumn{1}{c}{} & \multicolumn{1}{c}{} & \multicolumn{1}{c}{} & 
\multicolumn{1}{c}{} & \multicolumn{1}{c}{} & \multicolumn{1}{c}{} & 
\multicolumn{1}{c}{} & \multicolumn{1}{c}{} & \multicolumn{1}{c}{} & 
\multicolumn{1}{c}{} & \multicolumn{1}{c}{} & \multicolumn{1}{c}{} \\ 
\multicolumn{1}{c}{} & \multicolumn{1}{c}{} & \multicolumn{1}{c}{} & 
\multicolumn{1}{c}{$25$} & \multicolumn{1}{c}{} & \multicolumn{1}{c}{} & 
\multicolumn{1}{c}{$0.046$} & \multicolumn{1}{c}{} & \multicolumn{1}{c}{$%
0.045$} & \multicolumn{1}{c}{} & \multicolumn{1}{c}{$0.029$} & 
\multicolumn{1}{c}{} & \multicolumn{1}{c}{$0.014$} & \multicolumn{1}{c}{} & 
\multicolumn{1}{c}{$0.029$} & \multicolumn{1}{c}{} & \multicolumn{1}{c}{$%
0.041$} & \multicolumn{1}{c}{} & \multicolumn{1}{c}{$0.037$} & 
\multicolumn{1}{c}{} & \multicolumn{1}{c}{$0.025$} \\ 
\multicolumn{1}{c}{} & \multicolumn{1}{c}{$50$} & \multicolumn{1}{c}{} & 
\multicolumn{1}{c}{$50$} & \multicolumn{1}{c}{} & \multicolumn{1}{c}{} & 
\multicolumn{1}{c}{$0.070$} & \multicolumn{1}{c}{} & \multicolumn{1}{c}{$%
0.073$} & \multicolumn{1}{c}{} & \multicolumn{1}{c}{$0.046$} & 
\multicolumn{1}{c}{} & \multicolumn{1}{c}{$0.016$} & \multicolumn{1}{c}{} & 
\multicolumn{1}{c}{$0.041$} & \multicolumn{1}{c}{} & \multicolumn{1}{c}{$%
0.063$} & \multicolumn{1}{c}{} & \multicolumn{1}{c}{$0.061$} & 
\multicolumn{1}{c}{} & \multicolumn{1}{c}{$0.041$} \\ 
\multicolumn{1}{c}{} & \multicolumn{1}{c}{} & \multicolumn{1}{c}{} & 
\multicolumn{1}{c}{$100$} & \multicolumn{1}{c}{} & \multicolumn{1}{c}{} & 
\multicolumn{1}{c}{$0.062$} & \multicolumn{1}{c}{} & \multicolumn{1}{c}{$%
0.061$} & \multicolumn{1}{c}{} & \multicolumn{1}{c}{$0.046$} & 
\multicolumn{1}{c}{} & \multicolumn{1}{c}{$0.013$} & \multicolumn{1}{c}{} & 
\multicolumn{1}{c}{$0.029$} & \multicolumn{1}{c}{} & \multicolumn{1}{c}{$%
0.055$} & \multicolumn{1}{c}{} & \multicolumn{1}{c}{$0.042$} & 
\multicolumn{1}{c}{} & \multicolumn{1}{c}{$0.040$} \\ 
\multicolumn{1}{c}{} & \multicolumn{1}{c}{} & \multicolumn{1}{c}{} & 
\multicolumn{1}{c}{$200$} & \multicolumn{1}{c}{} & \multicolumn{1}{c}{} & 
\multicolumn{1}{c}{$0.070$} & \multicolumn{1}{c}{} & \multicolumn{1}{c}{$%
0.052$} & \multicolumn{1}{c}{} & \multicolumn{1}{c}{$0.032$} & 
\multicolumn{1}{c}{} & \multicolumn{1}{c}{$0.016$} & \multicolumn{1}{c}{} & 
\multicolumn{1}{c}{$0.024$} & \multicolumn{1}{c}{} & \multicolumn{1}{c}{$%
0.063$} & \multicolumn{1}{c}{} & \multicolumn{1}{c}{$0.065$} & 
\multicolumn{1}{c}{} & \multicolumn{1}{c}{$0.038$} \\ 
\multicolumn{1}{c}{} & \multicolumn{1}{c}{} & \multicolumn{1}{c}{} & 
\multicolumn{1}{c}{} & \multicolumn{1}{c}{} & \multicolumn{1}{c}{} & 
\multicolumn{1}{c}{} & \multicolumn{1}{c}{} & \multicolumn{1}{c}{} & 
\multicolumn{1}{c}{} & \multicolumn{1}{c}{} & \multicolumn{1}{c}{} & 
\multicolumn{1}{c}{} & \multicolumn{1}{c}{} & \multicolumn{1}{c}{} & 
\multicolumn{1}{c}{} & \multicolumn{1}{c}{} & \multicolumn{1}{c}{} & 
\multicolumn{1}{c}{} & \multicolumn{1}{c}{} & \multicolumn{1}{c}{} \\ 
\multicolumn{1}{c}{} & \multicolumn{1}{c}{} & \multicolumn{1}{c}{} & 
\multicolumn{1}{c}{} & \multicolumn{1}{c}{} & \multicolumn{1}{c}{} & 
\multicolumn{1}{c}{} & \multicolumn{1}{c}{} & \multicolumn{1}{c}{} & 
\multicolumn{1}{c}{} & \multicolumn{1}{c}{} & \multicolumn{1}{c}{} & 
\multicolumn{1}{c}{} & \multicolumn{1}{c}{} & \multicolumn{1}{c}{} & 
\multicolumn{1}{c}{} & \multicolumn{1}{c}{} & \multicolumn{1}{c}{} & 
\multicolumn{1}{c}{} & \multicolumn{1}{c}{} & \multicolumn{1}{c}{} \\ 
\multicolumn{1}{c}{} & \multicolumn{1}{c}{} & \multicolumn{1}{c}{} & 
\multicolumn{1}{c}{$50$} & \multicolumn{1}{c}{} & \multicolumn{1}{c}{} & 
\multicolumn{1}{c}{$0.039$} & \multicolumn{1}{c}{} & \multicolumn{1}{c}{$%
0.040$} & \multicolumn{1}{c}{} & \multicolumn{1}{c}{$0.033$} & 
\multicolumn{1}{c}{} & \multicolumn{1}{c}{$0.009$} & \multicolumn{1}{c}{} & 
\multicolumn{1}{c}{$0.027$} & \multicolumn{1}{c}{} & \multicolumn{1}{c}{$%
0.038$} & \multicolumn{1}{c}{} & \multicolumn{1}{c}{$0.036$} & 
\multicolumn{1}{c}{} & \multicolumn{1}{c}{$0.033$} \\ 
\multicolumn{1}{c}{} & \multicolumn{1}{c}{$100$} & \multicolumn{1}{c}{} & 
\multicolumn{1}{c}{$100$} & \multicolumn{1}{c}{} & \multicolumn{1}{c}{} & 
\multicolumn{1}{c}{$0.052$} & \multicolumn{1}{c}{} & \multicolumn{1}{c}{$%
0.045$} & \multicolumn{1}{c}{} & \multicolumn{1}{c}{$0.034$} & 
\multicolumn{1}{c}{} & \multicolumn{1}{c}{$0.012$} & \multicolumn{1}{c}{} & 
\multicolumn{1}{c}{$0.029$} & \multicolumn{1}{c}{} & \multicolumn{1}{c}{$%
0.048$} & \multicolumn{1}{c}{} & \multicolumn{1}{c}{$0.048$} & 
\multicolumn{1}{c}{} & \multicolumn{1}{c}{$0.037$} \\ 
\multicolumn{1}{c}{} & \multicolumn{1}{c}{} & \multicolumn{1}{c}{} & 
\multicolumn{1}{c}{$200$} & \multicolumn{1}{c}{} & \multicolumn{1}{c}{} & 
\multicolumn{1}{c}{$0.048$} & \multicolumn{1}{c}{} & \multicolumn{1}{c}{$%
0.046$} & \multicolumn{1}{c}{} & \multicolumn{1}{c}{$0.037$} & 
\multicolumn{1}{c}{} & \multicolumn{1}{c}{$0.011$} & \multicolumn{1}{c}{} & 
\multicolumn{1}{c}{$0.023$} & \multicolumn{1}{c}{} & \multicolumn{1}{c}{$%
0.042$} & \multicolumn{1}{c}{} & \multicolumn{1}{c}{$0.047$} & 
\multicolumn{1}{c}{} & \multicolumn{1}{c}{$0.036$} \\ 
\multicolumn{1}{c}{} & \multicolumn{1}{c}{} & \multicolumn{1}{c}{} & 
\multicolumn{1}{c}{$400$} & \multicolumn{1}{c}{} & \multicolumn{1}{c}{} & 
\multicolumn{1}{c}{$0.066$} & \multicolumn{1}{c}{} & \multicolumn{1}{c}{$%
0.077$} & \multicolumn{1}{c}{} & \multicolumn{1}{c}{$0.061$} & 
\multicolumn{1}{c}{} & \multicolumn{1}{c}{$0.021$} & \multicolumn{1}{c}{} & 
\multicolumn{1}{c}{$0.049$} & \multicolumn{1}{c}{} & \multicolumn{1}{c}{$%
0.051$} & \multicolumn{1}{c}{} & \multicolumn{1}{c}{$0.068$} & 
\multicolumn{1}{c}{} & \multicolumn{1}{c}{$0.063$} \\ 
\multicolumn{1}{c}{} & \multicolumn{1}{c}{} & \multicolumn{1}{c}{} & 
\multicolumn{1}{c}{} & \multicolumn{1}{c}{} & \multicolumn{1}{c}{} & 
\multicolumn{1}{c}{} & \multicolumn{1}{c}{} & \multicolumn{1}{c}{} & 
\multicolumn{1}{c}{} & \multicolumn{1}{c}{} & \multicolumn{1}{c}{} & 
\multicolumn{1}{c}{} & \multicolumn{1}{c}{} & \multicolumn{1}{c}{} & 
\multicolumn{1}{c}{} & \multicolumn{1}{c}{} & \multicolumn{1}{c}{} & 
\multicolumn{1}{c}{} & \multicolumn{1}{c}{} & \multicolumn{1}{c}{} \\ 
\multicolumn{1}{c}{} & \multicolumn{1}{c}{} & \multicolumn{1}{c}{} & 
\multicolumn{1}{c}{} & \multicolumn{1}{c}{} & \multicolumn{1}{c}{} & 
\multicolumn{1}{c}{} & \multicolumn{1}{c}{} & \multicolumn{1}{c}{} & 
\multicolumn{1}{c}{} & \multicolumn{1}{c}{} & \multicolumn{1}{c}{} & 
\multicolumn{1}{c}{} & \multicolumn{1}{c}{} & \multicolumn{1}{c}{} & 
\multicolumn{1}{c}{} & \multicolumn{1}{c}{} & \multicolumn{1}{c}{} & 
\multicolumn{1}{c}{} & \multicolumn{1}{c}{} & \multicolumn{1}{c}{} \\ 
\multicolumn{1}{c}{} & \multicolumn{1}{c}{} & \multicolumn{1}{c}{} & 
\multicolumn{1}{c}{$100$} & \multicolumn{1}{c}{} & \multicolumn{1}{c}{} & 
\multicolumn{1}{c}{$0.053$} & \multicolumn{1}{c}{} & \multicolumn{1}{c}{$%
0.041$} & \multicolumn{1}{c}{} & \multicolumn{1}{c}{$0.033$} & 
\multicolumn{1}{c}{} & \multicolumn{1}{c}{$0.011$} & \multicolumn{1}{c}{} & 
\multicolumn{1}{c}{$0.025$} & \multicolumn{1}{c}{} & \multicolumn{1}{c}{$%
0.048$} & \multicolumn{1}{c}{} & \multicolumn{1}{c}{$0.045$} & 
\multicolumn{1}{c}{} & \multicolumn{1}{c}{$0.037$} \\ 
\multicolumn{1}{c}{} & \multicolumn{1}{c}{$200$} & \multicolumn{1}{c}{} & 
\multicolumn{1}{c}{$200$} & \multicolumn{1}{c}{} & \multicolumn{1}{c}{} & 
\multicolumn{1}{c}{$0.059$} & \multicolumn{1}{c}{} & \multicolumn{1}{c}{$%
0.050$} & \multicolumn{1}{c}{} & \multicolumn{1}{c}{$0.045$} & 
\multicolumn{1}{c}{} & \multicolumn{1}{c}{$0.014$} & \multicolumn{1}{c}{} & 
\multicolumn{1}{c}{$0.031$} & \multicolumn{1}{c}{} & \multicolumn{1}{c}{$%
0.055$} & \multicolumn{1}{c}{} & \multicolumn{1}{c}{$0.061$} & 
\multicolumn{1}{c}{} & \multicolumn{1}{c}{$0.051$} \\ 
\multicolumn{1}{c}{} & \multicolumn{1}{c}{} & \multicolumn{1}{c}{} & 
\multicolumn{1}{c}{$400$} & \multicolumn{1}{c}{} & \multicolumn{1}{c}{} & 
\multicolumn{1}{c}{$0.046$} & \multicolumn{1}{c}{} & \multicolumn{1}{c}{$%
0.055$} & \multicolumn{1}{c}{} & \multicolumn{1}{c}{$0.043$} & 
\multicolumn{1}{c}{} & \multicolumn{1}{c}{$0.010$} & \multicolumn{1}{c}{} & 
\multicolumn{1}{c}{$0.032$} & \multicolumn{1}{c}{} & \multicolumn{1}{c}{$%
0.036$} & \multicolumn{1}{c}{} & \multicolumn{1}{c}{$0.053$} & 
\multicolumn{1}{c}{} & \multicolumn{1}{c}{$0.044$} \\ 
\multicolumn{1}{c}{} & \multicolumn{1}{c}{} & \multicolumn{1}{c}{} & 
\multicolumn{1}{c}{$800$} & \multicolumn{1}{c}{} & \multicolumn{1}{c}{} & 
\multicolumn{1}{c}{$0.046$} & \multicolumn{1}{c}{} & \multicolumn{1}{c}{$%
0.053$} & \multicolumn{1}{c}{} & \multicolumn{1}{c}{$0.038$} & 
\multicolumn{1}{c}{} & \multicolumn{1}{c}{$0.008$} & \multicolumn{1}{c}{} & 
\multicolumn{1}{c}{$0.018$} & \multicolumn{1}{c}{} & \multicolumn{1}{c}{$%
0.039$} & \multicolumn{1}{c}{} & \multicolumn{1}{c}{$0.045$} & 
\multicolumn{1}{c}{} & \multicolumn{1}{c}{$0.034$} \\ 
\multicolumn{1}{c}{} & \multicolumn{1}{c}{} & \multicolumn{1}{c}{} & 
\multicolumn{1}{c}{} & \multicolumn{1}{c}{} & \multicolumn{1}{c}{} & 
\multicolumn{1}{c}{} & \multicolumn{1}{c}{} & \multicolumn{1}{c}{} & 
\multicolumn{1}{c}{} & \multicolumn{1}{c}{} & \multicolumn{1}{c}{} & 
\multicolumn{1}{c}{} & \multicolumn{1}{c}{} & \multicolumn{1}{c}{} &  &  & 
&  &  &  \\ \hline\hline
\end{tabular}
}
\par
{\scriptsize \ }
\par
{\scriptsize {\footnotesize \ } }
\par
{\scriptsize {\footnotesize 
\begin{tablenotes}
      \tiny
            \item 
            
\end{tablenotes}
} }
\end{table*}

\begin{table*}[h!]
\caption{{\protect\footnotesize {Empirical rejection frequencies under the
null of no changepoint with covariates - Case III, $\protect\beta_0=1$}}}
\label{tab:ERF3c}\centering
{\footnotesize {\ }}
\par
{\footnotesize {\ }}
\par
{\scriptsize {\ }}
\par
{\scriptsize 
\begin{tabular}{lllllllllllllllllllll}
\hline\hline
&  &  &  &  &  &  &  &  &  &  &  &  &  &  &  &  &  &  &  &  \\ 
&  &  &  &  &  & \multicolumn{5}{c}{Weighted CUSUM} & \multicolumn{1}{c}{} & 
\multicolumn{3}{c}{Standardised CUSUM} & \multicolumn{1}{c}{} & 
\multicolumn{5}{c}{Weighted Page-CUSUM} \\ 
&  &  &  & $\psi $ &  & \multicolumn{1}{c}{$0$} & \multicolumn{1}{c}{} & 
\multicolumn{1}{c}{$0.25$} & \multicolumn{1}{c}{} & \multicolumn{1}{c}{$0.45$%
} & \multicolumn{1}{c}{} & \multicolumn{1}{c}{} & \multicolumn{1}{c}{$0.5$}
& \multicolumn{1}{c}{} & \multicolumn{1}{c}{} & \multicolumn{1}{c}{$0$} & 
\multicolumn{1}{c}{} & \multicolumn{1}{c}{$0.25$} & \multicolumn{1}{c}{} & 
\multicolumn{1}{c}{$0.45$} \\ 
&  &  &  &  &  & \multicolumn{1}{c}{} & \multicolumn{1}{c}{} & 
\multicolumn{1}{c}{} & \multicolumn{1}{c}{} & \multicolumn{1}{c}{} & 
\multicolumn{1}{c}{} & \multicolumn{1}{c}{$c_{\alpha ,0.5}$} & 
\multicolumn{1}{c}{} & \multicolumn{1}{c}{$\widehat{c}_{\alpha ,0.5}$} & 
\multicolumn{1}{c}{} & \multicolumn{1}{c}{} & \multicolumn{1}{c}{} & 
\multicolumn{1}{c}{} & \multicolumn{1}{c}{} & \multicolumn{1}{c}{} \\ 
&  &  &  &  &  & \multicolumn{1}{c}{} & \multicolumn{1}{c}{} & 
\multicolumn{1}{c}{} & \multicolumn{1}{c}{} & \multicolumn{1}{c}{} & 
\multicolumn{1}{c}{} & \multicolumn{1}{c}{} & \multicolumn{1}{c}{} & 
\multicolumn{1}{c}{} & \multicolumn{1}{c}{} & \multicolumn{1}{c}{} & 
\multicolumn{1}{c}{} & \multicolumn{1}{c}{} & \multicolumn{1}{c}{} & 
\multicolumn{1}{c}{} \\ 
\multicolumn{1}{c}{} & \multicolumn{1}{c}{$m$} & \multicolumn{1}{c}{} & 
\multicolumn{1}{c}{$m^{\ast }$} & \multicolumn{1}{c}{} & \multicolumn{1}{c}{}
& \multicolumn{1}{c}{} & \multicolumn{1}{c}{} & \multicolumn{1}{c}{} & 
\multicolumn{1}{c}{} & \multicolumn{1}{c}{} & \multicolumn{1}{c}{} & 
\multicolumn{1}{c}{} & \multicolumn{1}{c}{} & \multicolumn{1}{c}{} & 
\multicolumn{1}{c}{} & \multicolumn{1}{c}{} & \multicolumn{1}{c}{} & 
\multicolumn{1}{c}{} & \multicolumn{1}{c}{} & \multicolumn{1}{c}{} \\ 
\multicolumn{1}{c}{} & \multicolumn{1}{c}{} & \multicolumn{1}{c}{} & 
\multicolumn{1}{c}{} & \multicolumn{1}{c}{} & \multicolumn{1}{c}{} & 
\multicolumn{1}{c}{} & \multicolumn{1}{c}{} & \multicolumn{1}{c}{} & 
\multicolumn{1}{c}{} & \multicolumn{1}{c}{} & \multicolumn{1}{c}{} & 
\multicolumn{1}{c}{} & \multicolumn{1}{c}{} & \multicolumn{1}{c}{} & 
\multicolumn{1}{c}{} & \multicolumn{1}{c}{} & \multicolumn{1}{c}{} & 
\multicolumn{1}{c}{} & \multicolumn{1}{c}{} & \multicolumn{1}{c}{} \\ 
\multicolumn{1}{c}{} & \multicolumn{1}{c}{} & \multicolumn{1}{c}{} & 
\multicolumn{1}{c}{$25$} & \multicolumn{1}{c}{} & \multicolumn{1}{c}{} & 
\multicolumn{1}{c}{$0.084$} & \multicolumn{1}{c}{} & \multicolumn{1}{c}{$%
0.083$} & \multicolumn{1}{c}{} & \multicolumn{1}{c}{$0.061$} & 
\multicolumn{1}{c}{} & \multicolumn{1}{c}{$0.039$} & \multicolumn{1}{c}{} & 
\multicolumn{1}{c}{$0.066$} & \multicolumn{1}{c}{} & \multicolumn{1}{c}{$%
0.078$} & \multicolumn{1}{c}{} & \multicolumn{1}{c}{$0.074$} & 
\multicolumn{1}{c}{} & \multicolumn{1}{c}{$0.066$} \\ 
\multicolumn{1}{c}{} & \multicolumn{1}{c}{$50$} & \multicolumn{1}{c}{} & 
\multicolumn{1}{c}{$50$} & \multicolumn{1}{c}{} & \multicolumn{1}{c}{} & 
\multicolumn{1}{c}{$0.103$} & \multicolumn{1}{c}{} & \multicolumn{1}{c}{$%
0.107$} & \multicolumn{1}{c}{} & \multicolumn{1}{c}{$0.095$} & 
\multicolumn{1}{c}{} & \multicolumn{1}{c}{$0.057$} & \multicolumn{1}{c}{} & 
\multicolumn{1}{c}{$0.083$} & \multicolumn{1}{c}{} & \multicolumn{1}{c}{$%
0.103$} & \multicolumn{1}{c}{} & \multicolumn{1}{c}{$0.092$} & 
\multicolumn{1}{c}{} & \multicolumn{1}{c}{$0.084$} \\ 
\multicolumn{1}{c}{} & \multicolumn{1}{c}{} & \multicolumn{1}{c}{} & 
\multicolumn{1}{c}{$100$} & \multicolumn{1}{c}{} & \multicolumn{1}{c}{} & 
\multicolumn{1}{c}{$0.126$} & \multicolumn{1}{c}{} & \multicolumn{1}{c}{$%
0.124$} & \multicolumn{1}{c}{} & \multicolumn{1}{c}{$0.110$} & 
\multicolumn{1}{c}{} & \multicolumn{1}{c}{$0.071$} & \multicolumn{1}{c}{} & 
\multicolumn{1}{c}{$0.093$} & \multicolumn{1}{c}{} & \multicolumn{1}{c}{$%
0.124$} & \multicolumn{1}{c}{} & \multicolumn{1}{c}{$0.119$} & 
\multicolumn{1}{c}{} & \multicolumn{1}{c}{$0.101$} \\ 
\multicolumn{1}{c}{} & \multicolumn{1}{c}{} & \multicolumn{1}{c}{} & 
\multicolumn{1}{c}{$200$} & \multicolumn{1}{c}{} & \multicolumn{1}{c}{} & 
\multicolumn{1}{c}{$0.139$} & \multicolumn{1}{c}{} & \multicolumn{1}{c}{$%
0.132$} & \multicolumn{1}{c}{} & \multicolumn{1}{c}{$0.121$} & 
\multicolumn{1}{c}{} & \multicolumn{1}{c}{$0.074$} & \multicolumn{1}{c}{} & 
\multicolumn{1}{c}{$0.106$} & \multicolumn{1}{c}{} & \multicolumn{1}{c}{$%
0.136$} & \multicolumn{1}{c}{} & \multicolumn{1}{c}{$0.143$} & 
\multicolumn{1}{c}{} & \multicolumn{1}{c}{$0.125$} \\ 
\multicolumn{1}{c}{} & \multicolumn{1}{c}{} & \multicolumn{1}{c}{} & 
\multicolumn{1}{c}{} & \multicolumn{1}{c}{} & \multicolumn{1}{c}{} & 
\multicolumn{1}{c}{} & \multicolumn{1}{c}{} & \multicolumn{1}{c}{} & 
\multicolumn{1}{c}{} & \multicolumn{1}{c}{} & \multicolumn{1}{c}{} & 
\multicolumn{1}{c}{} & \multicolumn{1}{c}{} & \multicolumn{1}{c}{} & 
\multicolumn{1}{c}{} & \multicolumn{1}{c}{} & \multicolumn{1}{c}{} & 
\multicolumn{1}{c}{} & \multicolumn{1}{c}{} & \multicolumn{1}{c}{} \\ 
\multicolumn{1}{c}{} & \multicolumn{1}{c}{} & \multicolumn{1}{c}{} & 
\multicolumn{1}{c}{} & \multicolumn{1}{c}{} & \multicolumn{1}{c}{} & 
\multicolumn{1}{c}{} & \multicolumn{1}{c}{} & \multicolumn{1}{c}{} & 
\multicolumn{1}{c}{} & \multicolumn{1}{c}{} & \multicolumn{1}{c}{} & 
\multicolumn{1}{c}{} & \multicolumn{1}{c}{} & \multicolumn{1}{c}{} & 
\multicolumn{1}{c}{} & \multicolumn{1}{c}{} & \multicolumn{1}{c}{} & 
\multicolumn{1}{c}{} & \multicolumn{1}{c}{} & \multicolumn{1}{c}{} \\ 
\multicolumn{1}{c}{} & \multicolumn{1}{c}{} & \multicolumn{1}{c}{} & 
\multicolumn{1}{c}{$50$} & \multicolumn{1}{c}{} & \multicolumn{1}{c}{} & 
\multicolumn{1}{c}{$0.083$} & \multicolumn{1}{c}{} & \multicolumn{1}{c}{$%
0.077$} & \multicolumn{1}{c}{} & \multicolumn{1}{c}{$0.063$} & 
\multicolumn{1}{c}{} & \multicolumn{1}{c}{$0.036$} & \multicolumn{1}{c}{} & 
\multicolumn{1}{c}{$0.056$} & \multicolumn{1}{c}{} & \multicolumn{1}{c}{$%
0.075$} & \multicolumn{1}{c}{} & \multicolumn{1}{c}{$0.070$} & 
\multicolumn{1}{c}{} & \multicolumn{1}{c}{$0.060$} \\ 
\multicolumn{1}{c}{} & \multicolumn{1}{c}{$100$} & \multicolumn{1}{c}{} & 
\multicolumn{1}{c}{$100$} & \multicolumn{1}{c}{} & \multicolumn{1}{c}{} & 
\multicolumn{1}{c}{$0.098$} & \multicolumn{1}{c}{} & \multicolumn{1}{c}{$%
0.091$} & \multicolumn{1}{c}{} & \multicolumn{1}{c}{$0.078$} & 
\multicolumn{1}{c}{} & \multicolumn{1}{c}{$0.047$} & \multicolumn{1}{c}{} & 
\multicolumn{1}{c}{$0.074$} & \multicolumn{1}{c}{} & \multicolumn{1}{c}{$%
0.097$} & \multicolumn{1}{c}{} & \multicolumn{1}{c}{$0.094$} & 
\multicolumn{1}{c}{} & \multicolumn{1}{c}{$0.080$} \\ 
\multicolumn{1}{c}{} & \multicolumn{1}{c}{} & \multicolumn{1}{c}{} & 
\multicolumn{1}{c}{$200$} & \multicolumn{1}{c}{} & \multicolumn{1}{c}{} & 
\multicolumn{1}{c}{$0.103$} & \multicolumn{1}{c}{} & \multicolumn{1}{c}{$%
0.102$} & \multicolumn{1}{c}{} & \multicolumn{1}{c}{$0.090$} & 
\multicolumn{1}{c}{} & \multicolumn{1}{c}{$0.061$} & \multicolumn{1}{c}{} & 
\multicolumn{1}{c}{$0.072$} & \multicolumn{1}{c}{} & \multicolumn{1}{c}{$%
0.103$} & \multicolumn{1}{c}{} & \multicolumn{1}{c}{$0.111$} & 
\multicolumn{1}{c}{} & \multicolumn{1}{c}{$0.092$} \\ 
\multicolumn{1}{c}{} & \multicolumn{1}{c}{} & \multicolumn{1}{c}{} & 
\multicolumn{1}{c}{$400$} & \multicolumn{1}{c}{} & \multicolumn{1}{c}{} & 
\multicolumn{1}{c}{$0.112$} & \multicolumn{1}{c}{} & \multicolumn{1}{c}{$%
0.120$} & \multicolumn{1}{c}{} & \multicolumn{1}{c}{$0.106$} & 
\multicolumn{1}{c}{} & \multicolumn{1}{c}{$0.056$} & \multicolumn{1}{c}{} & 
\multicolumn{1}{c}{$0.084$} & \multicolumn{1}{c}{} & \multicolumn{1}{c}{$%
0.112$} & \multicolumn{1}{c}{} & \multicolumn{1}{c}{$0.118$} & 
\multicolumn{1}{c}{} & \multicolumn{1}{c}{$0.107$} \\ 
\multicolumn{1}{c}{} & \multicolumn{1}{c}{} & \multicolumn{1}{c}{} & 
\multicolumn{1}{c}{} & \multicolumn{1}{c}{} & \multicolumn{1}{c}{} & 
\multicolumn{1}{c}{} & \multicolumn{1}{c}{} & \multicolumn{1}{c}{} & 
\multicolumn{1}{c}{} & \multicolumn{1}{c}{} & \multicolumn{1}{c}{} & 
\multicolumn{1}{c}{} & \multicolumn{1}{c}{} & \multicolumn{1}{c}{} & 
\multicolumn{1}{c}{} & \multicolumn{1}{c}{} & \multicolumn{1}{c}{} & 
\multicolumn{1}{c}{} & \multicolumn{1}{c}{} & \multicolumn{1}{c}{} \\ 
\multicolumn{1}{c}{} & \multicolumn{1}{c}{} & \multicolumn{1}{c}{} & 
\multicolumn{1}{c}{} & \multicolumn{1}{c}{} & \multicolumn{1}{c}{} & 
\multicolumn{1}{c}{} & \multicolumn{1}{c}{} & \multicolumn{1}{c}{} & 
\multicolumn{1}{c}{} & \multicolumn{1}{c}{} & \multicolumn{1}{c}{} & 
\multicolumn{1}{c}{} & \multicolumn{1}{c}{} & \multicolumn{1}{c}{} & 
\multicolumn{1}{c}{} & \multicolumn{1}{c}{} & \multicolumn{1}{c}{} & 
\multicolumn{1}{c}{} & \multicolumn{1}{c}{} & \multicolumn{1}{c}{} \\ 
\multicolumn{1}{c}{} & \multicolumn{1}{c}{} & \multicolumn{1}{c}{} & 
\multicolumn{1}{c}{$100$} & \multicolumn{1}{c}{} & \multicolumn{1}{c}{} & 
\multicolumn{1}{c}{$0.068$} & \multicolumn{1}{c}{} & \multicolumn{1}{c}{$%
0.050$} & \multicolumn{1}{c}{} & \multicolumn{1}{c}{$0.043$} & 
\multicolumn{1}{c}{} & \multicolumn{1}{c}{$0.020$} & \multicolumn{1}{c}{} & 
\multicolumn{1}{c}{$0.032$} & \multicolumn{1}{c}{} & \multicolumn{1}{c}{$%
0.065$} & \multicolumn{1}{c}{} & \multicolumn{1}{c}{$0.062$} & 
\multicolumn{1}{c}{} & \multicolumn{1}{c}{$0.047$} \\ 
\multicolumn{1}{c}{} & \multicolumn{1}{c}{$200$} & \multicolumn{1}{c}{} & 
\multicolumn{1}{c}{$200$} & \multicolumn{1}{c}{} & \multicolumn{1}{c}{} & 
\multicolumn{1}{c}{$0.098$} & \multicolumn{1}{c}{} & \multicolumn{1}{c}{$%
0.087$} & \multicolumn{1}{c}{} & \multicolumn{1}{c}{$0.082$} & 
\multicolumn{1}{c}{} & \multicolumn{1}{c}{$0.041$} & \multicolumn{1}{c}{} & 
\multicolumn{1}{c}{$0.070$} & \multicolumn{1}{c}{} & \multicolumn{1}{c}{$%
0.092$} & \multicolumn{1}{c}{} & \multicolumn{1}{c}{$0.098$} & 
\multicolumn{1}{c}{} & \multicolumn{1}{c}{$0.092$} \\ 
\multicolumn{1}{c}{} & \multicolumn{1}{c}{} & \multicolumn{1}{c}{} & 
\multicolumn{1}{c}{$400$} & \multicolumn{1}{c}{} & \multicolumn{1}{c}{} & 
\multicolumn{1}{c}{$0.095$} & \multicolumn{1}{c}{} & \multicolumn{1}{c}{$%
0.096$} & \multicolumn{1}{c}{} & \multicolumn{1}{c}{$0.078$} & 
\multicolumn{1}{c}{} & \multicolumn{1}{c}{$0.044$} & \multicolumn{1}{c}{} & 
\multicolumn{1}{c}{$0.068$} & \multicolumn{1}{c}{} & \multicolumn{1}{c}{$%
0.094$} & \multicolumn{1}{c}{} & \multicolumn{1}{c}{$0.098$} & 
\multicolumn{1}{c}{} & \multicolumn{1}{c}{$0.083$} \\ 
\multicolumn{1}{c}{} & \multicolumn{1}{c}{} & \multicolumn{1}{c}{} & 
\multicolumn{1}{c}{$800$} & \multicolumn{1}{c}{} & \multicolumn{1}{c}{} & 
\multicolumn{1}{c}{$0.109$} & \multicolumn{1}{c}{} & \multicolumn{1}{c}{$%
0.114$} & \multicolumn{1}{c}{} & \multicolumn{1}{c}{$0.091$} & 
\multicolumn{1}{c}{} & \multicolumn{1}{c}{$0.056$} & \multicolumn{1}{c}{} & 
\multicolumn{1}{c}{$0.073$} & \multicolumn{1}{c}{} & \multicolumn{1}{c}{$%
0.103$} & \multicolumn{1}{c}{} & \multicolumn{1}{c}{$0.106$} & 
\multicolumn{1}{c}{} & \multicolumn{1}{c}{$0.100$} \\ 
\multicolumn{1}{c}{} & \multicolumn{1}{c}{} & \multicolumn{1}{c}{} & 
\multicolumn{1}{c}{} & \multicolumn{1}{c}{} & \multicolumn{1}{c}{} & 
\multicolumn{1}{c}{} & \multicolumn{1}{c}{} & \multicolumn{1}{c}{} & 
\multicolumn{1}{c}{} & \multicolumn{1}{c}{} & \multicolumn{1}{c}{} & 
\multicolumn{1}{c}{} & \multicolumn{1}{c}{} & \multicolumn{1}{c}{} &  &  & 
&  &  &  \\ \hline\hline
\end{tabular}
}
\par
{\scriptsize {\footnotesize 
\begin{tablenotes}
      \tiny
            \item 
            
\end{tablenotes}
} }
\end{table*}

\begin{table*}[h!]
\caption{{\protect\footnotesize {Median delays and empirical rejection
frequencies under alternatives - DGP with covariates}}}
\label{tab:Power2}\centering
\par
\resizebox{\textwidth}{!}{

\begin{tabular}{cccccccccccccccccccccc}
\hline\hline
&  &  &  &  &  &  &  &  &  &  &  &  &  &  &  &  &  &  &  &  &  \\ 
&  &  &  &  &  &  & \multicolumn{5}{c}{Weighted CUSUM} &  & 
\multicolumn{3}{c}{Standardised CUSUM} &  & \multicolumn{5}{c}{Weighted
Page-CUSUM } \\ 
& DGP &  &  &  & $\psi $ &  & $0$ &  & $0.25$ &  & $0.45$ &  &  & $0.5$ &  & 
& $0$ &  & $0.25$ &  & $0.45$ \\ 
&  &  &  &  &  &  &  &  &  &  &  &  & $c_{\alpha ,0.5}$ &  & $\widehat{c}_{\alpha ,0.5}$ &  &  &  &  &  &  \\ 
&  &  &  &  &  &  &  &  &  &  &  &  &  &  &  &  &  &  &  &  &  \\ 
\cline{2-22}
&  &  &  &  &  &  &  &  &  &  &  &  &  &  &  &  &  &  &  &  &  \\ 
&  &  &  & $100$ &  &  & $\underset{\left( 0.704\right) }{52}$ &  & $\underset{\left( 0.671\right) }{45}$ &  & $\underset{\left( 0.622\right) }{34}$ &  & $\underset{\left( 0.473\right) }{36}$ &  & $\underset{\left(
0.553\right) }{31}$ &  & $\underset{\left( 0.704\right) }{55}$ &  & $\underset{\left( 0.699\right) }{45.5}$ &  & $\underset{\left( 0.656\right) }{31}$ \\ 
& $\underset{\left( \beta _{0}=0.5\right) }{\text{\textbf{Case I}}}$ &  & $m^{\ast }$ & $200$ &  &  & $\underset{\left( 0.849\right) }{80}$ &  & $\underset{\left( 0.833\right) }{61.5}$ &  & $\underset{\left( 0.792\right) }{47}$ &  & $\underset{\left( 0.662\right) }{52}$ &  & $\underset{\left(
0.743\right) }{46}$ &  & $\underset{\left( 0.848\right) }{81}$ &  & $\underset{\left( 0.846\right) }{59}$ &  & $\underset{\left( 0.809\right) }{43.5}$ \\ 
&  &  &  & $400$ &  &  & $\underset{\left( 0.915\right) }{108.5}$ &  & $\underset{\left( 0.911\right) }{79}$ &  & $\underset{\left( 0.880\right) }{60}$ &  & $\underset{\left( 0.759\right) }{71.5}$ &  & $\underset{\left(
0.834\right) }{58}$ &  & $\underset{\left( 0.909\right) }{109}$ &  & $\underset{\left( 0.912\right) }{78}$ &  & $\underset{\left( 0.881\right) }{57}$ \\ 
&  &  &  & $800$ &  &  & $\underset{\left( 0.972\right) }{135}$ &  & $\underset{\left( 0.970\right) }{94}$ &  & $\underset{\left( 0.942\right) }{75}$ &  & $\underset{\left( 0.846\right) }{94}$ &  & $\underset{\left(
0.898\right) }{81}$ &  & $\underset{\left( 0.967\right) }{136}$ &  & $\underset{\left( 0.962\right) }{95}$ &  & $\underset{\left( 0.940\right) }{70}$ \\ 
&  &  &  &  &  &  &  &  &  &  &  &  &  &  &  &  &  &  &  &  &  \\ 
\cline{2-22}
&  &  &  &  &  &  &  &  &  &  &  &  &  &  &  &  &  &  &  &  &  \\ 
&  &  &  & $100$ &  &  & $\underset{\left( 1.000\right) }{15}$ &  & $\underset{\left( 1.000\right) }{13}$ &  & $\underset{\left( 1.000\right) }{9}
$ &  & $\underset{\left( 1.000\right) }{11}$ &  & $\underset{\left(
1.000\right) }{9}$ &  & $\underset{\left( 1.000\right) }{19}$ &  & $\underset{\left( 1.000\right) }{13}$ &  & $\underset{\left( 1.000\right) }{9}$ \\ 
& $\underset{\left( \beta _{0}=1.05\right) }{\text{\textbf{Case II}}}$ &  & $m^{\ast }$ & $200$ &  &  & $\underset{\left( 1.000\right) }{24}$ &  & $\underset{\left( 1.000\right) }{15}$ &  & $\underset{\left( 1.000\right) }{9}
$ &  & $\underset{\left( 1.000\right) }{11}$ &  & $\underset{\left(
1.000\right) }{9}$ &  & $\underset{\left( 1.000\right) }{25}$ &  & $\underset{\left( 1.000\right) }{15}$ &  & $\underset{\left( 1.000\right) }{9}$ \\ 
&  &  &  & $400$ &  &  & $\underset{\left( 1.000\right) }{27}$ &  & $\underset{\left( 1.000\right) }{16}$ &  & $\underset{\left( 1.000\right) }{10}$ &  & $\underset{\left( 1.000\right) }{11}$ &  & $\underset{\left(
1.000\right) }{9}$ &  & $\underset{\left( 1.000\right) }{28}$ &  & $\underset{\left( 1.000\right) }{16}$ &  & $\underset{\left( 1.000\right) }{10}$ \\ 
&  &  &  & $800$ &  &  & $\underset{\left( 1.000\right) }{30}$ &  & $\underset{\left( 1.000\right) }{17}$ &  & $\underset{\left( 1.000\right) }{10}$ &  & $\underset{\left( 1.000\right) }{11}$ &  & $\underset{\left(
1.000\right) }{9}$ &  & $\underset{\left( 1.000\right) }{31}$ &  & $\underset{\left( 1.000\right) }{17}$ &  & $\underset{\left( 1.000\right) }{9}$ \\ 
&  &  &  &  &  &  &  &  &  &  &  &  &  &  &  &  &  &  &  &  &  \\ 
\cline{2-22}
&  &  &  &  &  &  &  &  &  &  &  &  &  &  &  &  &  &  &  &  &  \\ 
&  &  &  & $100$ &  &  & $\underset{\left( 0.997\right) }{27}$ &  & $\underset{\left( 0.996\right) }{22}$ &  & $\underset{\left( 0.995\right) }{18}$ &  & $\underset{\left( 0.992\right) }{22}$ &  & $\underset{\left(
0.994\right) }{18}$ &  & $\underset{\left( 0.998\right) }{29}$ &  & $\underset{\left( 0.997\right) }{22}$ &  & $\underset{\left( 0.997\right) }{18}$ \\ 
& $\underset{\left( \beta _{0}=1\right) }{\text{\textbf{Case III}}}$ &  & $m^{\ast }$ & $200$ &  &  & $\underset{\left( 1.000\right) }{34}$ &  & $\underset{\left( 1.000\right) }{25}$ &  & $\underset{\left( 1.000\right) }{19}$ &  & $\underset{\left( 1.000\right) }{22}$ &  & $\underset{\left(
1.000\right) }{18}$ &  & $\underset{\left( 1.000\right) }{35}$ &  & $\underset{\left( 1.000\right) }{24}$ &  & $\underset{\left( 1.000\right) }{18}$ \\ 
&  &  &  & $400$ &  &  & $\underset{\left( 1.000\right) }{39}$ &  & $\underset{\left( 1.000\right) }{26}$ &  & $\underset{\left( 1.000\right) }{19}$ &  & $\underset{\left( 1.000\right) }{22}$ &  & $\underset{\left(
1.000\right) }{18}$ &  & $\underset{\left( 1.000\right) }{40}$ &  & $\underset{\left( 1.000\right) }{26}$ &  & $\underset{\left( 1.000\right) }{18}$ \\ 
&  &  &  & $800$ &  &  & $\underset{\left( 1.000\right) }{45}$ &  & $\underset{\left( 1.000\right) }{27}$ &  & $\underset{\left( 1.000\right) }{18}$ &  & $\underset{\left( 1.000\right) }{23}$ &  & $\underset{\left(
1.000\right) }{19}$ &  & $\underset{\left( 1.000\right) }{46}$ &  & $\underset{\left( 1.000\right) }{29}$ &  & $\underset{\left( 1.000\right) }{18}$ \\ 
&  &  &  &  &  &  &  &  &  &  &  &  &  &  &  &  &  &  &  &  &  \\ 
\hline\hline
\end{tabular}

}
\par
{\scriptsize {\ 
\begin{tablenotes}
      \tiny
            \item For each DGP, we report the \textit{mean} detection delay for only the cases where a changepoint is detected (thus leaving out the cases where no changepoint is detected). Numbers in round brackets represent the empirical rejection frequencies.
            \item We do not report median delays, unlike in Table \ref{tab:Power1}. All medians are around zero.
            
\end{tablenotes}
} }
\end{table*}

\clearpage

\newpage

\clearpage
\renewcommand*{\thesection}{\Alph{section}}

\setcounter{subsection}{-1} \setcounter{subsubsection}{-1} %
\setcounter{equation}{0} \setcounter{lemma}{0} \setcounter{theorem}{0} %
\renewcommand{\theassumption}{B.\arabic{assumption}} 
\renewcommand{\thetheorem}{B.\arabic{theorem}} \renewcommand{\thelemma}{B.%
\arabic{lemma}} \renewcommand{\thecorollary}{B.\arabic{corollary}} %
\renewcommand{\theequation}{B.\arabic{equation}}

\section{Further empirical evidence\label{empiric_further}}

\subsection{Further empirical evidence on UK Covid-19 hospitalisation data 
\label{further-covid}}

We report a graph of the logs of (one plus) the daily hospitalisation data
with the identified changepoints in Figure \ref{fig:FigCovid}.

\medskip

\begin{figure}[h]
\caption{Daily Covid-19 hospitalisations - with changepoints - for England.}
\label{fig:FigCovid}\centering
\hspace{-2.5cm}
\par
\centering
\includegraphics[scale=0.95]{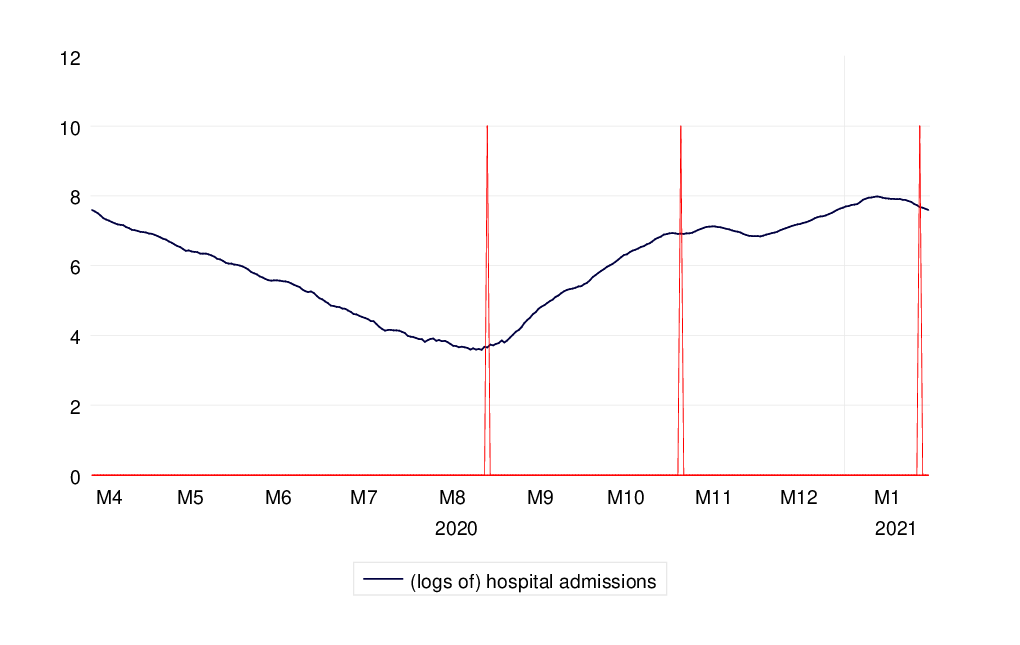}
\end{figure}

\medskip

We have conducted, by way of comparison, an ex-post changepoint detection
exercise, applying the techniques developed in \citet{horvath2022changepoint}
to the whole sample. We only report results obtained using R\'{e}nyi
statistics (corresponding to using a weighted version of the CUSUM process
with weights $\kappa =0.51$, $0.55$, $0.65$, $0.75$, $0.85$ and $1$ - see %
\citealp{horvath2022changepoint} for details).\footnote{%
Using other weighing schemes give very similar results, available upon
request.} As far as breakdates are concerned, we pick the ones corresponding
to the \textquotedblleft majority vote\textquotedblright\ across $\kappa $,
although discrepancies are, when present, in the region of few days ($2-5$
at most). We use binary segmentation, as also discussed in %
\citet{horvath2022changepoint}, to detect multiple breaks.

\medskip

\begin{table*}[h]
\caption{Ex-post changepoint detection for Covid-19 daily hospitalisation -
England data.}
\label{tab:TabCovid}
\par
\centering
{\scriptsize {\ }}
\par
{\scriptsize 
\begin{tabular}{llll}
\hline\hline
&  &  &  \\ 
\multicolumn{1}{c}{Changepoint 1} & \multicolumn{1}{c}{Changepoint 2} & 
\multicolumn{1}{c}{Changepoint 3} & \multicolumn{1}{c}{Changepoint 4} \\ 
\multicolumn{1}{c}{} & \multicolumn{1}{c}{} & \multicolumn{1}{c}{} & 
\multicolumn{1}{c}{} \\ 
\multicolumn{1}{c}{$\underset{\left[ 1.009\right] }{\text{Apr 10th, 2020}}$}
& \multicolumn{1}{c}{$\underset{\left[ 0.994\right] }{\text{Aug 26th, 2020}}$%
} & \multicolumn{1}{c}{$\underset{\left[ 1.010\right] }{\text{Oct 29th, 2020}%
}$} & \multicolumn{1}{c}{$\underset{\left[ 1.002\right] }{\text{Jan 12th,
2021}}$} \\ 
&  &  &  \\ \hline\hline
\end{tabular}
}
\par
{\scriptsize 
\begin{tablenotes}
      \tiny
            \item The series ends at 30 January 2021. We use the logs of the original data (plus one, given that, in some days, hospitalisations are equal to zero): no further transformations are used.
            \item All changepoints have been detected by all R\'{e}nyi-type tests - no discrepancies were noted. Detected changepoints, and their estimated date, are presented in \textit{chronological} order; breakdates have been estimated as the points in time where the majority of tests identifies a changepoint. Whilst details are available upon request, we note that breaks were detected with this order (from the first to be detected to the last one): break in August; break in April; break in January 2021; break in October.
            \item For each changepoint, we report in square brackets, for reference, the left WLS estimates of $\beta_0$ - i.e., the value of $\beta_0$ \textit{prior} to the breakdate.
            
\end{tablenotes}
}
\end{table*}

\subsection{Further empirical evidence on housing data\label{further-housing}%
}

We report some preliminary information on our data. In Table \ref%
{tab:TabURoot}, we report the outcome of a standard unit root test on the
three covariates used in our exercise.

\medskip

\begin{table*}[h]
\caption{Unit root tests applied to covariates}
\label{tab:TabURoot}
\par
\centering
\par
{\scriptsize {\ }}
\par
{\scriptsize {\ }}
\par
{\scriptsize 
\begin{tabular}{lllllll}
\hline\hline
&  &  &  &  &  & Notes \\ 
& \multicolumn{1}{c}{Variable} & \multicolumn{1}{c}{} & \multicolumn{1}{c}{
Period} & \multicolumn{1}{c}{} & \multicolumn{1}{c}{$t$-ADF} & 
\multicolumn{1}{c}{} \\ 
& \multicolumn{1}{c}{} & \multicolumn{1}{c}{} & \multicolumn{1}{c}{} & 
\multicolumn{1}{c}{} & \multicolumn{1}{c}{} & \multicolumn{1}{c}{} \\ 
& \multicolumn{1}{c}{AAA} & \multicolumn{1}{c}{} & \multicolumn{1}{c}{Mar
28th, 2008 - Oct 30th, 2009} & \multicolumn{1}{c}{} & \multicolumn{1}{c}{$%
\underset{\left[ 0.863\right] }{-1.387}$} & \multicolumn{1}{c}{Daily
frequency; trend and intercept used in ADF} \\ 
& \multicolumn{1}{c}{} & \multicolumn{1}{c}{} & \multicolumn{1}{c}{} & 
\multicolumn{1}{c}{} & \multicolumn{1}{c}{} & \multicolumn{1}{c}{} \\ 
& \multicolumn{1}{c}{GS10} & \multicolumn{1}{c}{} & \multicolumn{1}{c}{Mar
28th, 2008 - Oct 30th, 2009} & \multicolumn{1}{c}{} & \multicolumn{1}{c}{$%
\underset{\left[ 0.840\right] }{-1.463}$} & \multicolumn{1}{c}{Daily
frequency; trend and intercept used in ADF} \\ 
& \multicolumn{1}{c}{} & \multicolumn{1}{c}{} & \multicolumn{1}{c}{} & 
\multicolumn{1}{c}{} & \multicolumn{1}{c}{} & \multicolumn{1}{c}{} \\ 
& \multicolumn{1}{c}{VXO} & \multicolumn{1}{c}{} & \multicolumn{1}{c}{Mar
28th, 2008 - Oct 30th, 2009} & \multicolumn{1}{c}{} & \multicolumn{1}{c}{$%
\underset{\left[ 0.244\right] }{-2.682}$} & \multicolumn{1}{c}{Daily
frequency; trend and intercept used in ADF} \\ 
& \multicolumn{1}{c}{} & \multicolumn{1}{c}{} & \multicolumn{1}{c}{} & 
\multicolumn{1}{c}{} & \multicolumn{1}{c}{} & \multicolumn{1}{c}{} \\ 
& \multicolumn{1}{c}{WEI} & \multicolumn{1}{c}{} & \multicolumn{1}{c}{Jan
5th, 2008 - Dec 26th, 2009} & \multicolumn{1}{c}{} & \multicolumn{1}{c}{$%
\underset{\left[ 1.000\right] }{3.070}$} & \multicolumn{1}{c}{Weekly
frequency; trend and intercept used in ADF} \\ 
&  &  &  &  &  &  \\ \hline\hline
\end{tabular}
}
\par
{\scriptsize 
\begin{tablenotes}
      \tiny
            \item For each series, we have carried out a standard ADF test, choosing the number of lags in the Dickey-Fuller regression based on BIC. The numbers in square brackets are the p-values.
            
\end{tablenotes}
}
\end{table*}

\medskip

In Figure \ref{fig:FigLA}, we plot housing prices in Los Angeles between
March 28th, 2008, and October 30th, 2009.

\medskip

\begin{figure}[h]
\caption{Logs of daily housing prices in Los Angeles - with estimated
changepoint.}
\label{fig:FigLA}\centering
\hspace{-2.5cm}
\par
\centering
\includegraphics[scale=0.95]{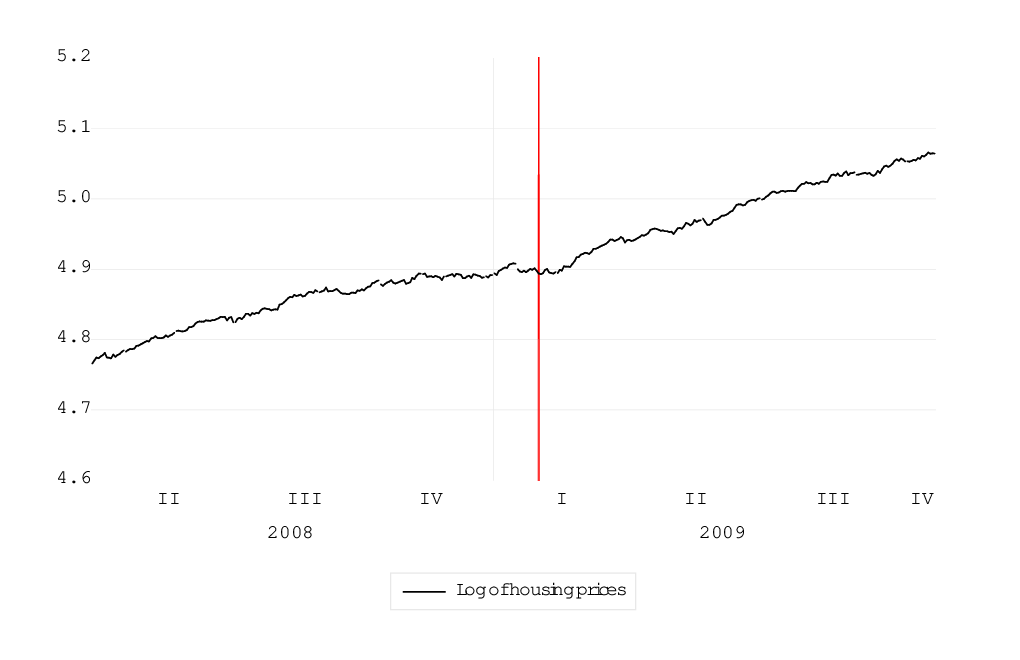}
\end{figure}

\medskip

\newpage

\clearpage
\renewcommand*{\thesection}{\Alph{section}}

\setcounter{subsection}{-1} \setcounter{subsubsection}{-1} %
\setcounter{equation}{0} \setcounter{lemma}{0} \setcounter{theorem}{0} %
\renewcommand{\theassumption}{C.\arabic{assumption}} 
\renewcommand{\thetheorem}{C.\arabic{theorem}} \renewcommand{\thelemma}{C.%
\arabic{lemma}} \renewcommand{\thecorollary}{C.\arabic{corollary}} %
\renewcommand{\theequation}{C.\arabic{equation}}

\section{Preliminary lemmas\label{lemmas}}

We will use the following facts and notation: \textquotedblleft $\overset{%
\mathcal{D}}{\rightarrow }$\textquotedblright\ denotes convergence in
distribution; \textquotedblleft $\overset{\mathcal{D}}{=}$%
\textquotedblright\ denotes equality in distribution; $c_{1}$, $c_{2}$, ...
denote positive, finite constants which do not depend on sample sizes, and
whose values may change from line to line.

Recall (\ref{rca}) under $H_{0}$%
\begin{equation*}
y_{i}=\left( \beta _{0}+\epsilon _{i,1}\right) y_{i-1}+\epsilon _{i,2}.
\end{equation*}%
If $E\log \left\vert \beta _{0}+\epsilon _{0,1}\right\vert <0$, define the
stationary solution%
\begin{equation}
\overline{y}_{i}=\sum_{\ell =0}^{\infty }\left( \prod_{j=1}^{\ell }\left(
\beta _{0}+\epsilon _{i-j,1}\right) \right) \epsilon _{i-\ell ,2},
\label{st-sol}
\end{equation}%
with the convention that $\prod_{\oslash }=0$. Finally, let%
\begin{equation}
a_{3}=E\left( \frac{\overline{y}_{0}^{2}}{1+\overline{y}_{0}^{2}}\right) .
\label{a3}
\end{equation}

\begin{lemma}
\label{lemma1}We assume that Assumption \ref{as-1} is satisfied. Under $%
H_{0} $, it holds that

(i) if $E\log \left\vert \beta _{0}+\epsilon _{0,1}\right\vert <0$, then%
\begin{align}
&\left\vert \left( \sum_{i=2}^{m}\frac{y_{i-1}^{2}}{1+y_{i-1}^{2}}\right)
^{-1}\left( \sum_{i=2}^{m}\frac{\epsilon _{i,1}y_{i-1}^{2}}{1+y_{i-1}^{2}}%
\right) -\frac{1}{ma_{3}}\left( \sum_{i=2}^{m}\frac{\epsilon _{i,1}\overline{%
y}_{i-1}^{2}}{1+\overline{y}_{i-1}^{2}}\right) \right\vert =O_{P}\left( 
\frac{1}{m}\right) ,  \label{lemma1-1} \\
&\left\vert \left( \sum_{i=2}^{m}\frac{y_{i-1}^{2}}{1+y_{i-1}^{2}}\right)
^{-1}\left( \sum_{i=2}^{m}\frac{\epsilon _{i,2}y_{i-1}}{1+y_{i-1}^{2}}%
\right) -\frac{1}{ma_{3}}\left( \sum_{i=2}^{m}\frac{\epsilon _{i,2}\overline{%
y}_{i-1}}{1+\overline{y}_{i-1}^{2}}\right) \right\vert =O_{P}\left( \frac{1}{%
m}\right) ;  \label{lemma1-2}
\end{align}

(ii) if either $E\log \left\vert \beta _{0}+\epsilon _{0,1}\right\vert =0$
and Assumption \ref{as-3} holds, or $E\log \left\vert \beta _{0}+\epsilon
_{0,1}\right\vert >0$ and Assumption \ref{as-4} holds, then%
\begin{align}
&\left\vert \left( \sum_{i=2}^{m}\frac{y_{i-1}^{2}}{1+y_{i-1}^{2}}\right)
^{-1}\left( \sum_{i=2}^{m}\frac{\epsilon _{i,1}y_{i-1}^{2}}{1+y_{i-1}^{2}}%
\right) -\frac{1}{m}\sum_{i=2}^{m}\epsilon _{i,1}\right\vert =O_{P}\left(
m^{-1/2-\zeta }\right) ,  \label{lemma1-3} \\
&\left\vert \left( \sum_{i=2}^{m}\frac{y_{i-1}^{2}}{1+y_{i-1}^{2}}\right)
^{-1}\left( \sum_{i=2}^{m}\frac{\epsilon _{i,2}y_{i-1}}{1+y_{i-1}^{2}}%
\right) \right\vert =O_{P}\left( m^{-1/2-\zeta }\right) ,  \label{lemma1-4}
\end{align}%
for some $\zeta >0$.

\begin{proof}
We begin by showing (\ref{lemma1-1})-(\ref{lemma1-2}). Under $E\log
\left\vert \beta _{0}+\epsilon _{0,1}\right\vert <0$, \citet{aue2006} show
that there exist a $\kappa >0$ and a $0<c<1$ such that%
\begin{equation*}
E\left\vert y_{i}-\overline{y}_{i}\right\vert ^{\kappa }=O\left(
c^{i}\right) ,
\end{equation*}%
as $i\rightarrow \infty $. Thus, by elementary algebra%
\begin{equation*}
\sum_{i=2}^{\infty }\left\vert \epsilon _{i,1}\right\vert \left\vert \frac{%
y_{i-1}^{2}}{1+y_{i-1}^{2}}-\frac{\overline{y}_{i-1}^{2}}{1+\overline{y}%
_{i-1}^{2}}\right\vert \leq \sum_{i=2}^{\infty }\left\vert \epsilon
_{i,1}\right\vert \left\vert y_{i-1}^{2}-\overline{y}_{i-1}^{2}\right\vert
\leq \sum_{i=2}^{\infty }\left\vert \epsilon _{i,1}\right\vert \left\vert
y_{i-1}-\overline{y}_{i-1}\right\vert \left( \left\vert y_{i-1}\right\vert
+\left\vert \overline{y}_{i-1}\right\vert \right) .
\end{equation*}%
Since we can assume that $\kappa <1$, it follows that%
\begin{align*}
& E\left( \sum_{i=2}^{\infty }\left\vert \epsilon _{i,1}\right\vert
\left\vert \frac{y_{i-1}^{2}}{1+y_{i-1}^{2}}-\frac{\overline{y}_{i-1}^{2}}{1+%
\overline{y}_{i-1}^{2}}\right\vert \right) ^{\kappa /2} \\
\leq & \sum_{i=2}^{\infty }\left( E\left\vert \epsilon _{i,1}\right\vert
^{\kappa /2}\right) E\left( \left\vert y_{i-1}-\overline{y}_{i-1}\right\vert
^{\kappa /2}\left( \left\vert y_{i-1}\right\vert +\left\vert \overline{y}%
_{i-1}\right\vert \right) ^{\kappa /2}\right) \\
\leq & \sum_{i=2}^{\infty }\left( E\left\vert \epsilon _{i,1}\right\vert
^{\kappa /2}\right) \left( E\left\vert y_{i-1}-\overline{y}_{i-1}\right\vert
^{\kappa }\right) ^{1/2}\left( E\left\vert y_{i-1}\right\vert ^{\kappa
}+E\left\vert \overline{y}_{i-1}\right\vert ^{\kappa }\right) ^{1/2} \\
<& \infty ,
\end{align*}%
having used Assumption \ref{as-1}. Hence by Markov's inequality it follows
that%
\begin{equation}
\sum_{i=2}^{\infty }\left\vert \epsilon _{i,1}\right\vert \left\vert \frac{%
y_{i-1}^{2}}{1+y_{i-1}^{2}}-\frac{\overline{y}_{i-1}^{2}}{1+\overline{y}%
_{i-1}^{2}}\right\vert =O_{P}\left( 1\right) ,  \label{markov-1}
\end{equation}%
and by the same logic it also follows that%
\begin{equation*}
\sum_{i=2}^{\infty }\left\vert \frac{y_{i-1}^{2}}{1+y_{i-1}^{2}}-\frac{%
\overline{y}_{i-1}^{2}}{1+\overline{y}_{i-1}^{2}}\right\vert =O_{P}\left(
1\right) .
\end{equation*}%
By Lemmas D.1-D.4 in \citet{horvath2022changepoint}, the sequence $\overline{%
y}_{i-1}^{2}/\left( 1+\overline{y}_{i-1}^{2}\right) $ is a decomposable
Bernoulli shift with all moments; hence, by Proposition 4.1 in %
\citet{berkeshormann} it follows that%
\begin{equation*}
\left\vert \sum_{i=2}^{m}\frac{\overline{y}_{i-1}^{2}}{1+\overline{y}%
_{i-1}^{2}}-ma_{3}\right\vert =O_{P}\left( m^{1/2}\right) ,
\end{equation*}%
where $a_{3}$ is defined in (\ref{a3}), and similarly%
\begin{equation*}
\left\vert \sum_{i=2}^{m}\frac{\epsilon _{i,1}\overline{y}_{i-1}^{2}}{1+%
\overline{y}_{i-1}^{2}}\right\vert =O_{P}\left( m^{1/2}\right) .
\end{equation*}%
Equation (\ref{lemma1-1}) now follows; (\ref{lemma1-2}) follows also from
exactly the same logic.

We now turn to showing (\ref{lemma1-3}) and (\ref{lemma1-4}). When $E\log
\left\vert \beta _{0}+\epsilon _{0,1}\right\vert =0$, Lemma A.4 in %
\citet{HT2016} implies that%
\begin{equation}
P\left\{ \left\vert y_{i}\right\vert \leq i^{\overline{\kappa }}\right\}
\leq ci^{-\overline{\kappa }};  \label{lammeht2016}
\end{equation}%
when $E\log \left\vert \beta _{0}+\epsilon _{0,1}\right\vert >0$, %
\citet{berkes2009} show that $\left\vert y_{i}\right\vert \rightarrow \infty 
$ a.s. exponentially fast, which implies (\ref{lammeht2016}). Hence by (\ref%
{lammeht2016}) it follows that%
\begin{equation*}
\sum_{i=2}^{m}E\left\vert \frac{y_{i-1}^{2}}{1+y_{i-1}^{2}}-1\right\vert
=\sum_{i=2}^{m}E\left\vert \frac{1}{1+y_{i-1}^{2}}\right\vert I\left(
\left\vert y_{i}\right\vert \leq i^{\overline{\kappa }}\right)
+\sum_{i=2}^{m}E\left\vert \frac{1}{1+y_{i-1}^{2}}\right\vert I\left(
\left\vert y_{i}\right\vert >i^{\overline{\kappa }}\right) =O\left( m^{1-%
\overline{\kappa }}\right) .
\end{equation*}%
By the independence between $\epsilon _{i,1}$ and $y_{i-1}$ and by
Assumption \ref{as-1}, it follows that%
\begin{equation*}
E\left( \sum_{i=2}^{m}\epsilon _{i,1}\left( \frac{y_{i-1}^{2}}{1+y_{i-1}^{2}}%
-1\right) \right) ^{2}=\sum_{i=2}^{m}E\epsilon _{i,1}^{2}E\left( \frac{%
y_{i-1}^{2}}{1+y_{i-1}^{2}}-1\right) ^{2}=O\left( m^{1-\overline{\kappa }%
}\right) ,
\end{equation*}%
and similarly%
\begin{equation*}
E\left( \sum_{i=2}^{m}\frac{\epsilon _{i,2}y_{i-1}}{1+y_{i-1}^{2}}\right)
^{2}=O\left( m^{1-\overline{\kappa }}\right) .
\end{equation*}%
The proof of (\ref{lemma1-3}) and (\ref{lemma1-4}) is now complete.
\end{proof}
\end{lemma}

\medskip

Consider now the decomposition%
\begin{equation}
Z_{m}\left( k\right) =\left\vert \left( \beta _{0}-\widehat{\beta }%
_{m}\right) \sum_{i=m+1}^{m+k}\frac{y_{i-1}^{2}}{1+y_{i-1}^{2}}%
+\sum_{i=m+1}^{m+k}\frac{\epsilon _{i,1}y_{i-1}^{2}}{1+y_{i-1}^{2}}%
+\sum_{i=m+1}^{m+k}\frac{\epsilon _{i,2}y_{i-1}}{1+y_{i-1}^{2}}\right\vert ,
\label{cusum-dec}
\end{equation}%
where $\widehat{\beta }_{m}$ is the WLS\ estimator computed using $\left\{
y_{1},...,y_{m}\right\} $.

In the next lemma, we obtain asymptotic representations for the terms in (%
\ref{cusum-dec}).

\begin{lemma}
\label{lemma2}We assume that Assumption \ref{as-1} is satisfied. Under $%
H_{0} $:

(i) if $E\log \left\vert \beta _{0}+\epsilon _{0,1}\right\vert <0$, then%
\begin{align}
&\max_{1\leq k<\infty }k^{-1/2-\eta }\left\vert \sum_{i=m+1}^{m+k}\frac{%
y_{i-1}^{2}}{1+y_{i-1}^{2}}-ka_{3}\right\vert =O_{P}\left( 1\right) ,
\label{lemma2-1} \\
&\max_{1\leq k<\infty }\left\vert \sum_{i=m+1}^{m+k}\frac{\epsilon
_{i,1}y_{i-1}^{2}}{1+y_{i-1}^{2}}-\sum_{i=m+1}^{m+k}\frac{\epsilon _{i,1}%
\overline{y}_{i-1}^{2}}{1+\overline{y}_{i-1}^{2}}\right\vert =O_{P}\left(
1\right) ,  \label{lemma2-2} \\
&\max_{1\leq k<\infty }\left\vert \sum_{i=m+1}^{m+k}\frac{\epsilon
_{i,2}y_{i-1}}{1+y_{i-1}^{2}}-\sum_{i=m+1}^{m+k}\frac{\epsilon _{i,2}%
\overline{y}_{i-1}}{1+\overline{y}_{i-1}^{2}}\right\vert =O_{P}\left(
1\right) ,  \label{lemma2-3}
\end{align}%
for all $\eta >0$;

(ii) if either $E\log \left\vert \beta _{0}+\epsilon _{0,1}\right\vert =0$
and Assumption \ref{as-3} holds, or $E\log \left\vert \beta _{0}+\epsilon
_{0,1}\right\vert >0$ and Assumption \ref{as-4} holds, then%
\begin{align}
&\max_{1\leq k<\infty }k^{-1}\left\vert \sum_{i=m+1}^{m+k}\frac{y_{i-1}^{2}}{%
1+y_{i-1}^{2}}-k\right\vert =O_{P}\left( m^{-\zeta }\right) ,
\label{lemma2-4} \\
&\max_{1\leq k<\infty }k^{-1/2-\eta }\left\vert \sum_{i=m+1}^{m+k}\frac{%
\epsilon _{i,1}y_{i-1}^{2}}{1+y_{i-1}^{2}}-\sum_{i=m+1}^{m+k}\epsilon
_{i,1}\right\vert =O_{P}\left( m^{-\zeta }\right) ,  \label{lemma2-5} \\
&\max_{1\leq k<\infty }k^{-1/2-\eta }\left\vert \sum_{i=m+1}^{m+k}\frac{%
\epsilon _{i,2}y_{i-1}}{1+y_{i-1}^{2}}\right\vert =O_{P}\left( m^{-\zeta
}\right) ,  \label{lemma2-6}
\end{align}%
for some $\zeta >0$ and for all $\eta >0$.

\begin{proof}
We begin by considering the stationary case $E\log \left\vert \beta
_{0}+\epsilon _{0,1}\right\vert <0$. On account of the proof of Lemma \ref%
{lemma1}, (\ref{lemma2-1}) follows if we show%
\begin{equation*}
\max_{1\leq k<\infty }k^{-1/2-\eta }\left\vert \sum_{i=m+1}^{m+k}\frac{%
\overline{y}_{i-1}^{2}}{1+\overline{y}_{i-1}^{2}}-ka_{3}\right\vert
=O_{P}\left( 1\right) .
\end{equation*}%
This follows immediately, since $\overline{y}_{i-1}^{2}/\left( 1+\overline{y}%
_{i-1}^{2}\right) $ is a decomposable Bernoulli shift with all moments, and
the result is implied by the strong approximation in \citet{aue2014}.
Further, (\ref{lemma2-2}) follows immediately from (\ref{markov-1}); (\ref%
{lemma2-3}) can be shown by the same logic.

Consider now the case $E\log \left\vert \beta _{0}+\epsilon
_{0,1}\right\vert \geq 0$. By (\ref{lammeht2016})%
\begin{equation*}
E\sum_{i=m+1}^{\infty }\frac{1}{i}\frac{1}{1+y_{i-1}^{2}}=O\left( m^{-%
\widehat{\kappa }}\right) ,
\end{equation*}%
for all $\widehat{\kappa }<\overline{\kappa }$. Hence, Abel's summation
formula yields%
\begin{equation*}
\frac{1}{k}\sum_{i=m+1}^{m+k}\frac{1}{1+y_{i-1}^{2}}=\sum_{i=m+1}^{m+k}\frac{%
1}{i}\frac{1}{1+y_{i-1}^{2}}-\frac{1}{k}\sum_{i=m+1}^{m+k-1}\left( \left(
i+1\right) -i\right) \left( \sum_{j=m+1}^{i}\frac{1}{j}\frac{1}{1+y_{j-1}^{2}%
}\right) .
\end{equation*}%
Given that%
\begin{equation*}
\frac{1}{k}\sum_{i=m+1}^{m+k-1}E\left( \sum_{j=m+1}^{i}\frac{1}{j}\frac{1}{%
1+y_{j-1}^{2}}\right) \leq E\sum_{j=m+1}^{\infty }\frac{1}{j}\frac{1}{%
1+y_{j-1}^{2}}=O\left( m^{-\widehat{\kappa }}\right) ,
\end{equation*}%
(\ref{lemma2-4}) follows from Markov's inequality. Letting $\mathcal{F}_{i}$
be the $\sigma $-field generated by $\left\{ \left( \epsilon _{j,1},\epsilon
_{j,2}\right) ,j\leq i\right\} $, note that%
\begin{equation*}
E\left( \left. \frac{\epsilon _{i,1}}{1+y_{i-1}^{2}}\right\vert \mathcal{F}%
_{i-1}\right) =0,
\end{equation*}%
and therefore the sequence $\epsilon _{i,1}/\left( 1+y_{i-1}^{2}\right) $ is
a martingale difference sequence. Also, using the Burkholder's inequality
(see e.g. Theorem 2.10 in \citealp{hallheyde}) and (\ref{lammeht2016})%
\begin{align}
&E\left\vert \sum_{i=m+1}^{m+k}E\left[ \left. \left( \frac{\epsilon _{i,1}}{%
1+y_{i-1}^{2}}\right) ^{2}\right\vert \mathcal{F}_{i-1}\right] \right\vert
^{\nu }  \label{m-ineq-1} \\
\leq &c_{1}E\left( \sum_{i=m+1}^{m+k}E\left( \frac{1}{1+y_{i-1}^{2}}\right)
^{2}\right) ^{\nu /2}\leq c_{2}k^{\nu /2-1}\sum_{i=m+1}^{m+k}\left( E\left( 
\frac{1}{\left( 1+y_{i-1}^{2}\right) ^{\nu }}\right) \right)  \notag \\
\leq &c_{2}k^{\nu /2-1}\sum_{i=m+1}^{m+k}\left( E\left( \frac{1}{\left(
1+y_{i-1}^{2}\right) ^{\nu }}I\left( \left\vert y_{i}\right\vert \leq i^{%
\overline{\kappa }}\right) \right) +E\left( \frac{1}{\left(
1+y_{i-1}^{2}\right) ^{\nu }}I\left( \left\vert y_{i}\right\vert >i^{%
\overline{\kappa }}\right) \right) \right)  \notag \\
\leq &k^{\nu /2-1}\sum_{i=m+1}^{m+k}\left( c_{3}i^{-\overline{\kappa }%
}+c_{4}i^{-2\nu \overline{\kappa }}\right) \leq c_{5}\frac{k^{\nu /2}}{m^{%
\overline{\kappa }}}.  \notag
\end{align}%
Similarly, we have%
\begin{equation}
\sum_{i=m+1}^{m+k}E\left\vert \frac{\epsilon _{i,1}}{1+y_{i-1}^{2}}%
\right\vert ^{\nu }\leq c_{6}\sum_{i=m+1}^{m+k}E\left\vert \frac{1}{%
1+y_{i-1}^{2}}\right\vert ^{\nu }\leq c_{7}\frac{k}{m^{\overline{\kappa }}}.
\label{m-ineq-2}
\end{equation}%
Using (\ref{m-ineq-1}) and (\ref{m-ineq-2}), and Rosenthal's maximal
inequality for martingale difference sequences (see e.g. Theorem 2.12 in %
\citealp{hallheyde}), it follows that 
\begin{equation}
E\max_{1\leq j\leq k}\left\vert \sum_{i=m+1}^{m+j}\frac{\epsilon _{i,1}}{%
1+y_{i-1}^{2}}\right\vert ^{\nu }\leq c_{8}\frac{k^{\nu /2}}{m^{\overline{%
\kappa }}}.  \label{rosenthal}
\end{equation}%
We now show (\ref{lemma2-5}) by noting that%
\begin{align*}
&P\left\{ \max_{1\leq k<\infty }k^{-1/2-\eta }\left\vert \sum_{i=m+1}^{m+k}%
\frac{\epsilon _{i,1}}{1+y_{i-1}^{2}}\right\vert >x\right\} \\
\leq &\sum_{\ell =0}^{\infty }P\left\{ \max_{\exp \left( \ell \right) \leq
k\leq \exp \left( \ell +1\right) }k^{-1/2-\eta }\left\vert \sum_{i=m+1}^{m+k}%
\frac{\epsilon _{i,1}}{1+y_{i-1}^{2}}\right\vert >x\right\} \\
\leq &\sum_{\ell =0}^{\infty }P\left\{ \max_{\exp \left( \ell \right) \leq
k\leq \exp \left( \ell +1\right) }\left\vert \sum_{i=m+1}^{m+k}\frac{%
\epsilon _{i,1}}{1+y_{i-1}^{2}}\right\vert >x\exp \left( \ell \left( \frac{1%
}{2}+\eta \right) \right) \right\} \\
\leq &c_{9}x^{-\nu }\sum_{\ell =0}^{\infty }\exp \left( -\nu \ell \left( 
\frac{1}{2}+\eta \right) \right) E\max_{1\leq k\leq \exp \left( \ell
+1\right) }\left\vert \sum_{i=m+1}^{m+k}\frac{\epsilon _{i,1}}{1+y_{i-1}^{2}}%
\right\vert ^{\nu } \\
\leq &c_{10}x^{-\nu }m^{-\overline{\kappa }}\sum_{\ell =0}^{\infty }\exp
\left( -\nu \ell \left( \frac{1}{2}+\eta \right) +\frac{\nu }{2}\left( \ell
+1\right) \right) ,
\end{align*}%
whence (\ref{lemma2-5}) follows immediately with $\zeta =\overline{\kappa }$%
. Equation (\ref{lemma2-6}) can be shown using the same logic.
\end{proof}
\end{lemma}

\begin{lemma}
\label{lemma3}We assume that Assumption \ref{as-1} is satisfied. Under $%
H_{0} $

(i) if $E\log \left\vert \beta _{0}+\epsilon _{0,1}\right\vert <0$, then we
can define two independent standard Wiener processes $\left\{ W_{m,1}\left(
k\right) ,1\leq k\leq m\right\} $ and $\left\{ W_{m,2}\left( k\right) ,1\leq
k<\infty \right\} $, whose distribution does not depend on $m$, such that 
\begin{align}
& \max_{1\leq k<m}k^{-1/2+\zeta }\left\vert \sum_{i=2}^{m}\left( \frac{%
\epsilon _{i,1}\overline{y}_{i-1}^{2}}{1+\overline{y}_{i-1}^{2}}+\frac{%
\epsilon _{i,2}\overline{y}_{i-1}}{1+\overline{y}_{i-1}^{2}}\right) -%
\mathscr{s}^{1/2}W_{m,1}\left( k\right) \right\vert  \label{lemma3-1} \\
=& O_{P}\left( 1\right) ,  \notag \\
& \max_{1\leq k<\infty }k^{-1/2+\zeta }\left\vert \sum_{i=m+1}^{m+k}\left( 
\frac{\epsilon _{i,1}\overline{y}_{i-1}^{2}}{1+\overline{y}_{i-1}^{2}}+\frac{%
\epsilon _{i,2}\overline{y}_{i-1}}{1+\overline{y}_{i-1}^{2}}\right) -%
\mathscr{s}^{1/2}W_{m,2}\left( k\right) \right\vert  \label{lemma3-2} \\
=& O_{P}\left( 1\right) ,  \notag
\end{align}%
for some $\zeta >0$;

(ii) if either $E\log \left\vert \beta _{0}+\epsilon _{0,1}\right\vert =0$
and Assumption \ref{as-3} holds, or $E\log \left\vert \beta _{0}+\epsilon
_{0,1}\right\vert >0$ and Assumption \ref{as-4} holds, then%
\begin{align}
& \max_{1\leq k\leq m}k^{-1/2+\zeta }\left\vert \sum_{i=1}^{k}\epsilon
_{i,1}-\sigma _{1}W_{m,1}\left( k\right) \right\vert =O_{P}\left( 1\right) ,
\label{lemma3-3} \\
& \max_{1\leq k<\infty }k^{-1/2+\zeta }\left\vert \sum_{i=m+1}^{m+k}\epsilon
_{i,1}-\sigma _{2}W_{m,2}\left( k\right) \right\vert =O_{P}\left( 1\right) ,
\label{lemma3-4}
\end{align}%
for some $\zeta >0$.

\begin{proof}
\citet{horvath2022changepoint} show that $\left( \epsilon _{i,1}\overline{y}%
_{i-1}^{2}+\epsilon _{i,2}\overline{y}_{i-1}\right) /\left( 1+\overline{y}%
_{i-1}^{2}\right) $ is a decomposable Bernoulli shift under $-\infty \leq
E\log \left\vert \beta _{0}+\epsilon _{0,1}\right\vert <0$. Hence, the
strong approximation shown in \citet{aue2014} immediately yields (\ref%
{lemma3-1}) and (\ref{lemma3-2}). Equations (\ref{lemma3-3}) and (\ref%
{lemma3-4}) follow directly from \citet{KMT1} and \citet{KMT2}.
\end{proof}
\end{lemma}

\bigskip

Let%
\begin{equation}
\Gamma _{m}\left( k\right) =\left\{ 
\begin{array}{ll}
\mathscr{s}^{1/2}\left\vert W_{m,2}\left( k\right) -kW_{m,1}\left( m\right)
\right\vert & \text{if }E\log \left\vert \beta _{0}+\epsilon
_{0,1}\right\vert <0\text{\ holds,} \\ 
\sigma _{1}\left\vert W_{m,2}\left( k\right) -kW_{m,1}\left( m\right)
\right\vert & \text{if }E\log \left\vert \beta _{0}+\epsilon
_{0,1}\right\vert \geq 0\text{\ holds.}%
\end{array}%
\right.  \label{gammak}
\end{equation}

\begin{lemma}
\label{gombay}Let $\beta _{m}$\ be a sequence such that, as $m\rightarrow
\infty $, $\gamma _{m}\rightarrow \infty $ with $\gamma _{m}=o\left(
m\right) $. Then, if assumptions of Theorem \ref{de} are satisfied, it holds
that under $H_{0}$ 
\begin{equation*}
\max_{\gamma _{m}\leq k\leq m^{\ast }}\frac{\left\vert Z_{m}\left( k\right)
-\Gamma _{m}\left( k\right) \right\vert }{m^{1/2}\left( 1+\displaystyle\frac{%
k}{m}\right) \left( \displaystyle\frac{k}{m+k}\right) ^{1/2}}=O_{P}\left(
\beta _{m}^{\zeta -1/2}\right) ,
\end{equation*}%
for some $0<\zeta <1/2$.

\begin{proof}
Upon following the proof of (\ref{str-approx-1}), it can be shown that there
exists a $0<\zeta <1/2$ such that%
\begin{equation*}
\max_{1\leq k\leq m^{\ast }}\left( k^{\zeta }+\frac{k}{m}m^{\zeta }\right)
^{-1}\left\vert Z_{m}\left( k\right) -\Gamma _{m}\left( k\right) \right\vert
=O_{P}\left( 1\right) .
\end{equation*}%
Hence we have%
\begin{align*}
& \max_{\gamma _{m}\leq k\leq m^{\ast }}\frac{\left\vert Z_{m}\left(
k\right) -\Gamma _{m}\left( k\right) \right\vert }{m^{1/2}\left( 1+%
\displaystyle\frac{k}{m}\right) \left( \displaystyle\frac{k}{m+k}\right)
^{1/2}} \\
=& O_{P}\left( 1\right) \max_{\gamma _{m}\leq k\leq m^{\ast }}\frac{k^{\zeta
}+\frac{k}{m}m^{\zeta }}{m^{1/2}\left( 1+\displaystyle\frac{k}{m}\right)
\left( \displaystyle\frac{k}{m+k}\right) ^{1/2}} \\
=& O_{P}\left( \gamma _{m}^{\zeta -1/2}\right) +O_{P}\left( m^{\zeta
-1/2}\right) =O_{P}\left( \gamma _{m}^{\zeta -1/2}\right) .
\end{align*}
\end{proof}
\end{lemma}

\begin{lemma}
\label{gombay2}We assume that the conditions of Theorem \ref{de} hold. Let 
\begin{equation*}
\gamma _{m^{\ast }}=O\left( \exp \left( \log \left( m^{\ast }\right)
^{1-\epsilon }\right) \right) ,
\end{equation*}%
with $\epsilon >0$ and arbitrarily small. Then, under $H_{0}$, on a suitably
enlarged space, it is possible to construct two independent Wiener processes 
$\left\{ W_{m,1}\left( k\right) ,1\leq k\leq T_{m}\right\} $ and $\left\{
W_{m,2}\left( k\right) ,1\leq k\leq m\right\} $\ whose distribution does not
depend on $m$, such that%
\begin{equation*}
\max_{\gamma _{m^{\ast }}\leq k\leq m^{\ast }}\left\vert \frac{\left\vert
Z_{m}\left( k\right) \right\vert }{g_{m,0.5}\left( k\right) }-\frac{%
\left\vert W_{m,2}\left( k\right) -\displaystyle\frac{k}{m}W_{m,1}\left(
m\right) \right\vert }{g_{m,0.5}\left( k\right) }\right\vert =O_{P}\left(
\exp \left( -c_{0}\log \left( m^{\ast }\right) ^{1-\epsilon }\right) \right)
,
\end{equation*}%
for some $0<c_{0}<1/2$.

\begin{proof}
The proof follows from the same arguments as Lemma \ref{gombay}, of which
this lemma is a special case.
\end{proof}
\end{lemma}

\medskip

We now report some preliminary results to prove the main results in Section %
\ref{covariates}. We begin by studying the WLS\ estimator $\widehat{\beta }%
_{m}$. Define%
\begin{equation*}
\mathbf{Q}_{m}=\left[ 
\begin{array}{cc}
y_{1} & \mathbf{x}_{2}^{\intercal } \\ 
y_{2} & \mathbf{x}_{3}^{\intercal } \\ 
. & . \\ 
y_{m-1} & \mathbf{x}_{m}^{\intercal }%
\end{array}%
\right] ,
\end{equation*}%
and the diagonal matrix%
\begin{equation*}
\mathbf{W}_{m}=\text{diag}\left\{ \frac{1}{1+y_{1}^{2}},\frac{1}{1+y_{2}^{2}}%
,...,\frac{1}{1+y_{m-1}^{2}}\right\} .
\end{equation*}%
Then it holds that%
\begin{equation*}
\widehat{\mathbf{b}}_{m}=\left( \mathbf{Q}_{m}^{\intercal }\mathbf{W}_{m}%
\mathbf{Q}_{m}\right) ^{-1}\mathbf{Q}_{m}^{\intercal }\mathbf{W}_{m}\mathbf{Y%
}_{m},
\end{equation*}%
where $\mathbf{Y}_{m}=\left( y_{2},y_{3},...,y_{m}\right) ^{\intercal }$.
Using the recursion defined in (\ref{rca-x}), we obtain%
\begin{equation*}
\widehat{\mathbf{b}}_{m}-\mathbf{b}_{0}=\left( \mathbf{Q}_{m}^{\intercal }%
\mathbf{W}_{m}\mathbf{Q}_{m}\right) ^{-1}\mathbf{Q}_{m}^{\intercal }\mathbf{W%
}_{m}\mathbf{E}_{m},
\end{equation*}%
having defined $\mathbf{b}_{0}=\left( \beta _{0},\mathbf{\lambda }%
_{0}^{\intercal }\right) ^{\intercal }$ and%
\begin{equation*}
\mathbf{E}_{m}=\left( \epsilon _{1,2}y_{1}+\epsilon _{2,2},\epsilon
_{1,3}y_{2}+\epsilon _{2,3},...,\epsilon _{1,m}y_{m-1}+\epsilon
_{2,m}\right) ^{\intercal }.
\end{equation*}%
Based on the above, under the condition for stationarity $E\log \left\vert
\beta _{0}+\epsilon _{0,1}\right\vert <0$, it can be verified that%
\begin{equation*}
\overline{y}_{i}=\left( \beta _{0}+\epsilon _{i,1}\right) \overline{y}_{i-1}+%
\mathbf{\lambda }_{0}^{\intercal }\mathbf{x}_{i}+\epsilon _{i,2},\text{ \ \ }%
-\infty <i<\infty ,
\end{equation*}%
has a unique stationary, causal solution. Consider the variables%
\begin{align}
&\mathbf{z}_{i} =\frac{1}{\left( 1+\overline{y}_{i-1}^{2}\right) ^{1/2}}%
\left( \overline{y}_{i-1},\mathbf{x}_{i}^{\intercal }\right) ^{\intercal },
\label{zi} \\
&\mathbf{\eta }_{i} =\left( \frac{\left( \epsilon _{i,1}\overline{y}%
_{i-1}+\epsilon _{i,2}\right) \overline{y}_{i-1}}{1+\overline{y}_{i-1}^{2}},%
\frac{\left( \epsilon _{i,1}\overline{y}_{i-1}+\epsilon _{i,2}\right) 
\mathbf{x}_{i}^{\intercal }}{1+\overline{y}_{i-1}^{2}}\right) ^{\intercal },
\label{etai}
\end{align}%
and define 
\begin{equation}
\mathbf{Q}=E\left( \mathbf{z}_{1}\mathbf{z}_{1}^{\intercal }\right) ,
\label{q}
\end{equation}%
\begin{equation}
\mathbf{C}=E\left( \mathbf{\eta }_{0}\mathbf{\eta }_{0}^{\intercal }\right) ,
\label{c}
\end{equation}%
and 
\begin{equation}
\mathbf{a}=E\left( \frac{\left( \overline{y}_{1},\mathbf{x}_{2}^{\intercal
}\right) ^{\intercal }\overline{y}_{1}}{1+\overline{y}_{1}^{2}}\right) .
\label{a}
\end{equation}

\begin{lemma}
\label{beta-x}We assume that $E\log \left\vert \beta _{0}+\epsilon
_{0,1}\right\vert <0$, and that Assumptions \ref{as-1}, \ref{as-x-1} and \ref%
{as-x-2} are satisfied. Then it holds that%
\begin{equation*}
\widehat{\mathbf{b}}_{m}-\mathbf{b}_{0}=\frac{1}{m}\mathbf{Q}%
^{-1}\sum_{i=2}^{m}\mathbf{\eta }_{i}+o_{P}\left( m^{-1/2-\zeta }\right) ,
\end{equation*}%
for some $\zeta >0$, where $\mathbf{\eta }_{i}$ and $\mathbf{Q}$ are defined
in (\ref{etai}) and (\ref{q}).

\begin{proof}
We begin by writing explicitly the solution of (\ref{rca-x}) as%
\begin{equation}
\overline{y}_{i}=\sum_{\ell =0}^{\infty }\left( \prod_{j=1}^{\ell }\left(
\beta _{0}+\epsilon _{i-j,1}\right) \right) \left( \mathbf{x}_{i-\ell
}^{\intercal }\mathbf{\lambda }_{0}+\epsilon _{i-\ell ,2}\right) ,
\label{rec-x}
\end{equation}%
with the convention that $\prod_{\oslash }=0$. With minor modifications of
the arguments in \citet{aue2006}, it can be shown that there exist a $\kappa
>0$ and a constant $0<c<1$\ such that 
\begin{equation}
E\left\vert y_{i}-\overline{y}_{i}\right\vert ^{\kappa }=O\left(
c^{i}\right) ,  \label{aue2006-x}
\end{equation}%
as $i\rightarrow \infty $. Using (\ref{aue2006-x}) it can be shown that 
\begin{equation*}
\left\Vert \mathbf{Q}_{m}^{\intercal }\mathbf{W}_{m}\mathbf{Q}_{m}-\overline{%
\mathbf{Q}}_{m}\right\Vert =O_{P}\left( 1\right) ,
\end{equation*}%
where 
\begin{equation*}
\overline{\mathbf{Q}}_{m}=\sum_{i=2}^{m}\overline{\mathbf{z}}_{i}\overline{%
\mathbf{z}}_{i}^{\intercal },
\end{equation*}%
and $\overline{\mathbf{z}}_{i}$ is constructed in the same way as $\mathbf{z}%
_{i}$ defined in (\ref{zi}), replacing $y_{i}$ with $\overline{y}_{i}$. We
now show that $\overline{y}_{i}$ defined in (\ref{rec-x}) is a decomposable
Bernoulli shift. Indeed, let%
\begin{equation*}
\overline{y}_{i,k}=\sum_{\ell =0}^{k}\left( \prod_{j=1}^{\ell }\left( \beta
_{0}+\epsilon _{i-j,1}\right) \right) \left( \mathbf{x}_{i-\ell }^{\intercal
}\mathbf{\lambda }_{0}+\epsilon _{i-\ell ,2}\right) +\sum_{\ell
=k+1}^{\infty }\left( \prod_{j=1}^{\ell }\left( \beta _{0}+\epsilon
_{i-j,1}^{\ast }\right) \right) \left( \mathbf{x}_{i-\ell ,k}^{\intercal }%
\mathbf{\lambda }_{0}+\epsilon _{i-\ell ,2}^{\ast }\right) ,
\end{equation*}%
where for $\ell \geq 0$%
\begin{equation*}
\mathbf{x}_{i-\ell ,k}=\left\{ 
\begin{array}{ll}
\mathbf{g}\left( \eta _{i-\ell },\eta _{i-\ell -1},...,\eta _{k},\eta
_{k-1}^{\ast },\eta _{k-2}^{\ast },...\right) & \text{if }\ell <i-k, \\ 
\mathbf{g}\left( \eta _{i-\ell }^{\ast },\eta _{i-\ell -1}^{\ast },...\right)
& \text{if }\ell \geq i-k,%
\end{array}%
\right.
\end{equation*}%
the $\eta _{j}^{\ast }$s are independent copies of $\eta _{0}$, $\left(
\epsilon _{j,1}^{\ast },\epsilon _{j,2}^{\ast }\right) $ are independent
copies of $\left( \epsilon _{0,1},\epsilon _{0,2}\right) $, and the
sequences $\left\{ \eta _{j},-\infty <j<\infty \right\} $, $\left\{ \eta
_{j}^{\ast },-\infty <j<\infty \right\} $, $\left\{ \left( \epsilon
_{j,1},\epsilon _{j,2}\right) ,-\infty <j<\infty \right\} $ and $\left\{
\left( \epsilon _{j,1}^{\ast },\epsilon _{j,2}^{\ast }\right) ,\right. $ $%
\left. -\infty <j<\infty \right\} $ are independent. Defining $u_{i-\ell }=%
\mathbf{x}_{i-\ell }^{\intercal }\mathbf{\lambda }_{0}+\epsilon _{i-\ell ,2}$%
, and recalling Assumptions \ref{as-x-1} and \ref{as-x-2}, by the same
arguments as in \citet{aue2006} we can show that, under $E\log \left\vert
\beta _{0}+\epsilon _{0,1}\right\vert <0$, there exists a $\kappa >0$ such
that%
\begin{equation*}
E\left\vert \beta _{0}+\epsilon _{0,1}\right\vert ^{\kappa }<1.
\end{equation*}%
Using $\kappa <1$, we have%
\begin{align*}
&E\left\vert \sum_{\ell =k+1}^{\infty }\left( \prod_{j=1}^{\ell }\left(
\beta _{0}+\epsilon _{i-j,1}\right) \right) \left( \mathbf{x}_{i-\ell
}^{\intercal }\mathbf{\lambda }_{0}+\epsilon _{i-\ell ,2}\right) \right\vert
^{\kappa /2} \\
\leq &\sum_{\ell =k+1}^{\infty }E\left( \left\vert \prod_{j=1}^{\ell }\left(
\beta _{0}+\epsilon _{i-j,1}\right) \right\vert ^{\kappa /2}\left\vert 
\mathbf{x}_{i-\ell }^{\intercal }\mathbf{\lambda }_{0}+\epsilon _{i-\ell
,2}\right\vert ^{\kappa /2}\right) \\
\leq &\sum_{\ell =k+1}^{\infty }\left( E\left\vert \prod_{j=1}^{\ell }\left(
\beta _{0}+\epsilon _{i-j,1}\right) \right\vert ^{\kappa }\right)
^{1/2}\left( E\left\vert \mathbf{x}_{i-\ell }^{\intercal }\mathbf{\lambda }%
_{0}+\epsilon _{i-\ell ,2}\right\vert ^{\kappa }\right) ^{1/2} \\
=&\sum_{\ell =k+1}^{\infty }\left( E\left\vert \beta _{0}+\epsilon
_{0,1}\right\vert ^{\kappa }\right) ^{\ell /2}\left( E\left\vert \mathbf{x}%
_{0}^{\intercal }\mathbf{\lambda }_{0}+\epsilon _{0,2}\right\vert ^{\kappa
}\right) ^{1/2} \\
\leq &c_{1}\sum_{\ell =k+1}^{\infty }\rho ^{\ell /2}\leq c_{2}c_{3}^{k},
\end{align*}%
where $\rho =E\left\vert \beta _{0}+\epsilon _{0,1}\right\vert ^{\kappa }$,
and $0<c_{3}<1$, having used Assumptions \ref{as-x-1} and \ref{as-x-2}. By
the same token, we have%
\begin{align*}
&E\left\vert \sum_{\ell =k+1}^{\infty }\left( \prod_{j=1}^{\ell }\left(
\beta _{0}+\epsilon _{i-j,1}^{\ast }\right) \right) \left( \mathbf{x}%
_{i-\ell ,k}^{\intercal }\mathbf{\lambda }_{0}+\epsilon _{i-\ell ,2}^{\ast
}\right) \right\vert ^{\kappa /2} \\
\leq &\sum_{\ell =k+1}^{\infty }E\left( \left\vert \prod_{j=1}^{\ell }\left(
\beta _{0}+\epsilon _{i-j,1}^{\ast }\right) \right\vert ^{\kappa
/2}\left\vert \mathbf{x}_{i-\ell ,k}^{\intercal }\mathbf{\lambda }%
_{0}+\epsilon _{i-\ell ,2}^{\ast }\right\vert ^{\kappa /2}\right) \\
\leq &\sum_{\ell =k+1}^{\infty }\left( E\left\vert \prod_{j=1}^{\ell }\left(
\beta _{0}+\epsilon _{i-j,1}^{\ast }\right) \right\vert ^{\kappa }\right)
^{1/2}\left( E\left\vert \mathbf{x}_{i-\ell ,k}^{\intercal }\mathbf{\lambda }%
_{0}+\epsilon _{i-\ell ,2}^{\ast }\right\vert ^{\kappa }\right) ^{1/2} \\
=&\sum_{\ell =k+1}^{\infty }\left( E\left\vert \beta _{0}+\epsilon
_{0,1}\right\vert ^{\kappa }\right) ^{\ell /2}\left( E\left\vert \mathbf{x}%
_{0}^{\intercal }\mathbf{\lambda }_{0}+\epsilon _{0,2}\right\vert ^{\kappa
}\right) ^{1/2}\leq c_{4}c_{5}^{k},
\end{align*}%
with $0<c_{5}<1$. Hence it holds that%
\begin{equation}
E\left\vert \overline{y}_{i}-\overline{y}_{i,k}\right\vert ^{k}\leq
c_{6}c_{7}^{k},  \label{ht2016-x}
\end{equation}%
for some $c_{6}>0$ and $0<c_{7}<1$, which shows that $\overline{y}_{i}$ is a
decomposable Bernoulli shift. This immediately yields, by the approximations
developed in \citet{aue2014}, that%
\begin{equation}
\left\Vert \frac{1}{m}\overline{\mathbf{Q}}_{m}-\mathbf{Q}\right\Vert
=O_{P}\left( m^{-1/2}\right) .  \label{deno-x}
\end{equation}

We now turn to studying the numerator of $\widehat{\beta }_{m}-\beta _{0}$.
We begin with some preliminary results. Using (\ref{ht2016-x}), it follows
that, for all $\kappa _{3}>0$ and for $c_{7}<c_{9}<1$ 
\begin{align}
&E\left\vert \frac{\overline{y}_{i}^{2}}{1+\overline{y}_{i}^{2}}-\frac{%
\overline{y}_{i,k}^{2}}{1+\overline{y}_{i,k}^{2}}\right\vert  \label{bsx1} \\
\leq &E\left\vert \frac{\left\vert \overline{y}_{i}\right\vert +\left\vert 
\overline{y}_{i,k}\right\vert }{\left( 1+\overline{y}_{i}^{2}\right) \left(
1+\overline{y}_{i,k}^{2}\right) }\left\vert \overline{y}_{i}-\overline{y}%
_{i,k}\right\vert \right\vert ^{\kappa _{3}}  \notag \\
\leq &E\left( c_{8}\left\vert \overline{y}_{i}-\overline{y}_{i,k}\right\vert
I\left( \left\vert \overline{y}_{i}-\overline{y}_{i,k}\right\vert \leq
c_{9}^{k}\right) +c_{10}I\left( \left\vert \overline{y}_{i}-\overline{y}%
_{i,k}\right\vert \geq c_{9}^{k}\right) \right) ^{\kappa _{3}}  \notag \\
\leq &c_{11}c_{9}^{k\kappa _{3}}+c_{12}P\left( \left\vert \overline{y}_{i}-%
\overline{y}_{i,k}\right\vert \geq c_{9}^{k}\right) \leq
c_{11}c_{9}^{k\kappa _{3}}+c_{12}c_{9}^{-k}E\left\vert \overline{y}_{i}-%
\overline{y}_{i,k}\right\vert  \notag \\
\leq &c_{13}c_{14}^{k},  \notag
\end{align}%
for some $0<c_{9}<1$. Using the same arguments, for all $\kappa _{4}<\kappa
_{1}$%
\begin{equation}
E\left\Vert \frac{\mathbf{x}_{i}\overline{y}_{i-1}}{1+\overline{y}_{i-1}^{2}}%
-\frac{\mathbf{x}_{i,k}\overline{y}_{i-1,k}}{1+\overline{y}_{i-1,k}^{2}}%
\right\Vert ^{\kappa _{4}}\leq c_{15}c_{16}^{k},  \label{bsx2}
\end{equation}%
and for all $\kappa _{5}<\kappa _{2}/2$%
\begin{equation}
E\left\Vert \frac{\mathbf{x}_{i}\mathbf{x}_{i}^{\intercal }}{1+\overline{y}%
_{i-1}^{2}}-\frac{\mathbf{x}_{i,k}\mathbf{x}_{i,k}^{\intercal }}{1+\overline{%
y}_{i-1,k}^{2}}\right\Vert ^{\kappa _{5}}\leq c_{17}c_{18}^{k},  \label{bsx3}
\end{equation}%
where $0<c_{16},c_{18}<1$. This entails that all the sequences studied above
are decomposable Benrnoulli shifts. Thus%
\begin{equation*}
\left\Vert \mathbf{Q}_{m}^{\intercal }\mathbf{W}_{m}\mathbf{E}_{m}-\overline{%
\mathbf{E}}_{m}\right\Vert =O_{P}\left( 1\right) ,
\end{equation*}%
where 
\begin{equation*}
\overline{\mathbf{E}}_{m}=\sum_{i=2}^{m}\mathbf{\eta }_{i},
\end{equation*}%
and $\mathbf{\eta }_{i}$\ is defined in (\ref{etai}).Repeating the arguments
above, it can also be shown that, for all $\kappa _{6}<\kappa _{1}$%
\begin{equation*}
E\left\Vert \mathbf{\eta }_{i}-\mathbf{\eta }_{i,k}\right\Vert ^{\kappa
_{6}}\leq c_{19}c_{20}^{k},
\end{equation*}%
with $\mathbf{\eta }_{i,k}$ defined as the other coupling constructions
above and $0<c_{20}<1$. The final results now follows from the strong
approximation derived in \citet{aue2014}.
\end{proof}
\end{lemma}

\begin{lemma}
\label{beta-x-2}We assume that $E\log \left\vert \beta _{0}+\epsilon
_{0,1}\right\vert <0$, and that Assumptions \ref{as-1}, \ref{as-x-1} and \ref%
{as-x-2} are satisfied. Then it holds that%
\begin{equation}
\max_{1\leq k<\infty }k^{-1/2-\eta }\left\Vert \sum_{i=m+1}^{m+k}\left(
y_{i-1},\mathbf{x}_{i}^{\intercal }\right) \frac{y_{i-1}}{1+y_{i-1}^{2}}-k%
\mathbf{a}\right\Vert =O_{P}\left( 1\right) ,  \label{beta-x-2-1}
\end{equation}%
for all $\eta >0$, and%
\begin{equation}
\max_{1\leq k<\infty }\left\vert \sum_{i=m+1}^{m+k}\frac{\left( \epsilon
_{i,1}y_{i-1}+\epsilon _{i,2}\right) y_{i-1}}{1+y_{i-1}^{2}}%
-\sum_{i=m+1}^{m+k}\frac{\left( \epsilon _{i,1}\overline{y}_{i-1}+\epsilon
_{i,2}\right) \overline{y}_{i-1}}{1+\overline{y}_{i-1}^{2}}\right\vert
=O_{P}\left( 1\right) .  \label{beta-x-2-2}
\end{equation}

\begin{proof}
The proof is based on the results derived in the proof of Lemma \ref{beta-x}%
. Equation (\ref{beta-x-2-1}) follows from exactly the same argumens as the
proof of (\ref{lemma2-1}), and (\ref{beta-x-2-2}) follows from the proof of (%
\ref{lemma2-2})-(\ref{lemma2-3}).
\end{proof}
\end{lemma}

\begin{lemma}
\label{beta-x-3}We assume that $E\log \left\vert \beta _{0}+\epsilon
_{0,1}\right\vert <0$, and that Assumptions \ref{as-1}, \ref{as-x-1} and \ref%
{as-x-2} are satisfied. Then, on a suitably enlarged probability space, it
is possible to define two independent Wiener processes $\left\{
W_{m,1}\left( k\right) ,1\leq k\leq m\right\} $ and $\left\{ W_{m,2}\left(
k\right) ,1\leq k<\infty \right\} $, whose distributions do not depend on $m$%
, such that%
\begin{equation*}
\left\vert \sum_{i=2}^{m}\mathbf{a}^{\intercal }\mathbf{Q\eta }_{i}-%
\mathscr{s}_{x,1}W_{m,1}\left( m\right) \right\vert =O_{P}\left(
m^{1/2-\zeta }\right) ,
\end{equation*}%
and%
\begin{equation*}
\max_{1\leq k<\infty }\frac{1}{k^{1/2-\zeta }}\left\vert \sum_{i=m+1}^{m+k}%
\frac{\left( \epsilon _{i,1}\overline{y}_{i-1}+\epsilon _{i,2}\right) 
\overline{y}_{i-1}}{1+\overline{y}_{i-1}^{2}}-\mathscr{s}_{x,2}W_{m,2}\left(
k\right) \right\vert =O_{P}\left( 1\right) ,
\end{equation*}%
for some $0<\zeta <1/2$, where $\mathscr{s}_{x,1}$ and $\mathscr{s}_{x,2}$
are defined in (\ref{sig-x-2}).

\begin{proof}
In the proof of Lemma \ref{beta-x} we have shown that $\mathbf{\eta }_{i}$
is a decomposable Bernoulli shift; the desired result follows from %
\citet{aue2014}.
\end{proof}
\end{lemma}

\begin{lemma}
\label{beta-x-4}We assume that $E\log \left\vert \beta _{0}+\epsilon
_{0,1}\right\vert >0$, and that Assumptions \ref{as-1}, \ref{as-x-1}-\ref%
{as-x-3} are satisfied. Then it holds that%
\begin{equation*}
\left\Vert \widehat{\mathbf{b}}_{m}-\mathbf{b}_{0}\right\Vert =O_{P}\left(
m^{-1/2}\right) .
\end{equation*}

\begin{proof}
Recall that%
\begin{equation*}
\widehat{\mathbf{b}}_{m}-\mathbf{b}_{0}=\left( \mathbf{Q}_{m}^{\intercal }%
\mathbf{W}_{m}\mathbf{Q}_{m}\right) ^{-1}\mathbf{Q}_{m}^{\intercal }\mathbf{W%
}_{m}\mathbf{E}_{m}.
\end{equation*}%
Repeating the proof of Theorem 3.1 in \citet{berkes2009}, it follows that $%
\left\vert y_{i}\right\vert \rightarrow \infty $ a.s. exponentially fast;
this immediately yields%
\begin{equation*}
\left\Vert \mathbf{Q}_{m}^{\intercal }\mathbf{W}_{m}\mathbf{Q}_{m}-\mathbf{B}%
_{m}\right\Vert =O_{P}\left( 1\right) ,
\end{equation*}%
where $\mathbf{B}_{m}$ is a $\left( p+1\right) \times \left( p+1\right) $
symmetric matrix with elements $\left\{ B_{i,j},1\leq i,j\leq p+1\right\} $
defined as%
\begin{align*}
&B_{1,1} =m, \\
&\left( B_{1,2},...,B_{1,p+1}\right) =\sum_{i=2}^{\infty }\frac{\mathbf{x}%
_{i}^{\intercal }y_{i-1}}{1+y_{i-1}^{2}}, \\
&\left\{ B_{i,j},2\leq i,j\leq p+1\right\} =\sum_{i=2}^{\infty }\frac{%
\mathbf{x}_{i}\mathbf{x}_{i}^{\intercal }}{1+y_{i-1}^{2}}.
\end{align*}%
Since we already know from the above that $\left\Vert \mathbf{E}%
_{m}\right\Vert =O_{P}\left( m^{1/2}\right) $, the result follows
immediately.
\end{proof}
\end{lemma}

\newpage

\clearpage
\renewcommand*{\thesection}{\Alph{section}}

\setcounter{subsection}{-1} \setcounter{subsubsection}{-1} %
\setcounter{equation}{0} \setcounter{lemma}{0} \setcounter{theorem}{0} %
\renewcommand{\theassumption}{D.\arabic{assumption}} 
\renewcommand{\thetheorem}{D.\arabic{theorem}} \renewcommand{\thelemma}{D.%
\arabic{lemma}} \renewcommand{\thecorollary}{D.\arabic{corollary}} %
\renewcommand{\theequation}{D.\arabic{equation}}

\section{Proofs\label{proofs}}

\begin{proof}[Proof of Theorem \protect\ref{th-1}]
Recall (\ref{gammak}). We begin by showing that%
\begin{equation}
\max_{1\leq k<\infty }\frac{\left\vert Z_{m}\left( k\right) -\Gamma
_{m}\left( k\right) \right\vert }{m^{1/2}\left( 1+\displaystyle\frac{k}{m}%
\right) \left( \displaystyle\frac{k}{m+k}\right) ^{\psi }}=o_{P}\left(
1\right) .  \label{str-approx-1}
\end{equation}%
We begin by considering the case $E\log \left\vert \beta _{0}+\epsilon
_{0,1}\right\vert <0$; we write%
\begin{align*}
& \left( \beta _{0}-\widehat{\beta }_{m}\right) \sum_{i=m+1}^{m+k}\frac{%
y_{i-1}^{2}}{1+y_{i-1}^{2}}+\frac{k}{m}\sum_{i=2}^{m}\left( \frac{\epsilon
_{i,1}\overline{y}_{i-1}^{2}}{1+\overline{y}_{i-1}^{2}}+\frac{\epsilon _{i,2}%
\overline{y}_{i-1}}{1+\overline{y}_{i-1}^{2}}\right) \\
=& \left( \beta _{0}-\widehat{\beta }_{m}\right) \left( \sum_{i=m+1}^{m+k}%
\frac{y_{i-1}^{2}}{1+y_{i-1}^{2}}-ka_{3}\right) \\
& -ka_{3}\left( \left( \sum_{i=2}^{m}\frac{y_{i-1}^{2}}{1+y_{i-1}^{2}}%
\right) ^{-1}-\frac{1}{ma_{3}}\right) \left( \sum_{i=2}^{m}\left( \frac{%
\epsilon _{i,1}y_{i-1}^{2}}{1+y_{i-1}^{2}}+\frac{\epsilon _{i,2}y_{i-1}}{%
1+y_{i-1}^{2}}\right) \right) \\
& +\frac{k}{m}\sum_{i=2}^{m}\left( \epsilon _{i,1}\left( \frac{\overline{y}%
_{i-1}^{2}}{1+\overline{y}_{i-1}^{2}}-\frac{y_{i-1}^{2}}{1+y_{i-1}^{2}}%
\right) +\epsilon _{i,2}\left( \frac{\overline{y}_{i-1}}{1+\overline{y}%
_{i-1}^{2}}-\frac{y_{i-1}}{1+y_{i-1}^{2}}\right) \right) \\
=& R_{m,1}\left( k\right) +R_{m,2}\left( k\right) +R_{m,3}\left( k\right) .
\end{align*}%
Using Lemmas \ref{lemma1} and \ref{lemma2} with $\eta <1/2$%
\begin{align*}
& \max_{1\leq k<\infty }\frac{\left\vert R_{m,1}\left( k\right) \right\vert 
}{m^{1/2}\left( 1+\displaystyle\frac{k}{m}\right) \left( \displaystyle\frac{k%
}{m+k}\right) ^{\psi }} \\
=& O_{P}\left( m^{-1}\right) \max_{1\leq k\leq M}\frac{k^{1/2+\eta }}{\left(
1+\displaystyle\frac{k}{m}\right) \left( \displaystyle\frac{k}{m+k}\right)
^{\psi }} \\
& +O_{P}\left( m^{-1}\right) \max_{M<k<\infty }\frac{k^{1/2+\eta }}{\left( 1+%
\displaystyle\frac{k}{m}\right) \left( \displaystyle\frac{k}{m+k}\right)
^{\psi }} \\
=& O_{P}\left( 1\right) \left( m^{-1/2}\max_{1\leq k\leq M}k^{\eta
}+\max_{M<k<\infty }k^{\eta -1/2}\right) =o_{P}\left( 1\right) .
\end{align*}%
Using the same logic, it can also be shown that%
\begin{align*}
\max_{1\leq k<\infty }\frac{\left\vert R_{m,2}\left( k\right) \right\vert }{%
m^{1/2}\left( 1+\displaystyle\frac{k}{m}\right) \left( \displaystyle\frac{k}{%
m+k}\right) ^{\psi }}& =o_{P}\left( 1\right) , \\
\max_{1\leq k<\infty }\frac{\left\vert R_{m,3}\left( k\right) \right\vert }{%
m^{1/2}\left( 1+\displaystyle\frac{k}{m}\right) \left( \displaystyle\frac{k}{%
m+k}\right) ^{\psi }}& =o_{P}\left( 1\right) .
\end{align*}%
Using Lemma \ref{lemma3}, it follows that%
\begin{align}
& \max_{1\leq k<\infty }\frac{\left\vert \displaystyle\frac{k}{m}%
\sum_{i=2}^{m}\left( \displaystyle\frac{\epsilon _{i,1}\overline{y}_{i-1}^{2}%
}{1+\overline{y}_{i-1}^{2}}+\displaystyle\frac{\epsilon _{i,2}\overline{y}%
_{i-1}}{1+\overline{y}_{i-1}^{2}}\right) -\displaystyle\frac{k}{m}\left(
a_{1}\sigma _{1}^{2}+a_{2}\sigma _{2}^{2}\right) ^{1/2}W_{m,1}\left(
m\right) \right\vert }{m^{1/2}\left( 1+\displaystyle\frac{k}{m}\right)
\left( \displaystyle\frac{k}{m+k}\right) ^{\psi }}  \label{sip-1} \\
=& O_{P}\left( 1\right) \max_{1\leq k<\infty }\frac{k}{m}\frac{m^{1/2-\zeta }%
}{m^{1/2}\left( 1+\displaystyle\frac{k}{m}\right) \left( \displaystyle\frac{k%
}{m+k}\right) ^{\psi }}  \notag \\
=& O_{P}\left( m^{-\zeta }\right) \max_{1\leq k<\infty }\left( \frac{k}{m+k}%
\right) ^{1-\psi }=o_{P}\left( 1\right) .  \notag
\end{align}%
By the same token, we can show that%
\begin{align}
& \max_{1\leq k<\infty }\frac{\left\vert \displaystyle\sum_{i=m+1}^{m+k}%
\left( \epsilon _{i,1}\left( \displaystyle\frac{\overline{y}_{i-1}^{2}}{1+%
\overline{y}_{i-1}^{2}}-\displaystyle\frac{y_{i-1}^{2}}{1+y_{i-1}^{2}}%
\right) \right) \right\vert }{m^{1/2}\left( 1+\displaystyle\frac{k}{m}%
\right) \left( \displaystyle\frac{k}{m+k}\right) ^{\psi }}  \label{sip-2} \\
=& O_{P}\left( 1\right) \max_{1\leq k<\infty }\frac{1}{m^{1/2}\left( 1+%
\displaystyle\frac{k}{m}\right) \left( \displaystyle\frac{k}{m+k}\right)
^{\psi }}  \notag \\
=& O_{P}\left( m^{-1/2}\right) \max_{1\leq k<\infty }\frac{m}{m+k}\left( 
\frac{k}{m+k}\right) ^{-\psi }=O_{P}\left( m^{-1/2+\psi }\right)
=o_{P}\left( 1\right) ,  \notag
\end{align}%
and 
\begin{align}
& \max_{1\leq k<\infty }\frac{\left\vert \displaystyle\sum_{i=m+1}^{m+k}%
\left( \epsilon _{i,2}\left( \displaystyle\frac{\overline{y}_{i-1}}{1+%
\overline{y}_{i-1}^{2}}-\displaystyle\frac{y_{i-1}}{1+y_{i-1}^{2}}\right)
\right) \right\vert }{m^{1/2}\left( 1+\displaystyle\frac{k}{m}\right) \left( %
\displaystyle\frac{k}{m+k}\right) ^{\psi }}  \label{sip-22} \\
=& O_{P}\left( 1\right) \max_{1\leq k<\infty }\frac{1}{m^{1/2}\left( 1+%
\displaystyle\frac{k}{m}\right) \left( \displaystyle\frac{k}{m+k}\right)
^{\psi }}  \notag \\
=& O_{P}\left( m^{-1/2}\right) \max_{1\leq k<\infty }\frac{m}{m+k}\left( 
\frac{k}{m+k}\right) ^{-\psi }=O_{P}\left( m^{-1/2+\psi }\right)
=o_{P}\left( 1\right) ,  \notag
\end{align}%
having used Lemma \ref{lemma2}. Similarly, by Lemma \ref{lemma3}%
\begin{align}
& \max_{1\leq k<\infty }\frac{\left\vert \displaystyle\sum_{i=m+1}^{m+k}%
\left( \displaystyle\frac{\epsilon _{i,1}\overline{y}_{i-1}^{2}+\epsilon
_{i,2}\overline{y}_{i-1}}{1+\overline{y}_{i-1}^{2}}\right) -\left(
a_{1}\sigma _{1}^{2}+a_{2}\sigma _{2}^{2}\right) ^{1/2}W_{m,2}\left(
k\right) \right\vert }{m^{1/2}\left( 1+\displaystyle\frac{k}{m}\right)
\left( \displaystyle\frac{k}{m+k}\right) ^{\psi }}  \label{sip-3} \\
=& O_{P}\left( 1\right) \max_{1\leq k<\infty }\frac{k^{1/2-\zeta }}{%
m^{1/2}\left( 1+\displaystyle\frac{k}{m}\right) \left( \displaystyle\frac{k}{%
m+k}\right) ^{\psi }}=o_{P}\left( 1\right) .  \notag
\end{align}%
Putting all together, (\ref{str-approx-1}) has been shown for the stationary
case. Considering now the case $E\log \left\vert \beta _{0}+\epsilon
_{0,1}\right\vert \geq 0$, write%
\begin{align*}
& \left( \beta _{0}-\widehat{\beta }_{m}\right) \sum_{i=m+1}^{m+k}\frac{%
y_{i-1}^{2}}{1+y_{i-1}^{2}}+\frac{k}{m}\sum_{i=2}^{m}\epsilon _{i,1} \\
=& \left( \beta _{0}-\widehat{\beta }_{m}\right) \left( \sum_{i=m+1}^{m+k}%
\frac{y_{i-1}^{2}}{1+y_{i-1}^{2}}-k\right) \\
& -k\left( \left( \sum_{i=2}^{m}\frac{y_{i-1}^{2}}{1+y_{i-1}^{2}}\right)
^{-1}-\frac{1}{m}\right) \left( \sum_{i=2}^{m}\left( \frac{\epsilon
_{i,1}y_{i-1}^{2}}{1+y_{i-1}^{2}}+\frac{\epsilon _{i,2}y_{i-1}}{1+y_{i-1}^{2}%
}\right) \right) \\
& -\frac{k}{m}\left( \sum_{i=2}^{m}\frac{\epsilon _{i,1}}{1+y_{i-1}^{2}}%
+\sum_{i=2}^{m}\frac{\epsilon _{i,2}y_{i-1}}{1+y_{i-1}^{2}}\right) \\
=& R_{m,4}\left( k\right) +R_{m,5}\left( k\right) +R_{m,6}\left( k\right) .
\end{align*}%
The arguments are similar to the above. Indeed, using Lemmas \ref{lemma1}-%
\ref{lemma3}%
\begin{align}
& \max_{1\leq k<\infty }\frac{\left\vert R_{m,4}\left( k\right) \right\vert 
}{m^{1/2}\left( 1+\displaystyle\frac{k}{m}\right) \left( \displaystyle\frac{k%
}{m+k}\right) ^{\psi }}  \label{sip-4} \\
=& O_{P}\left( m^{-1/2}\right) \max_{1\leq k<\infty }\frac{km^{-\zeta }}{%
m^{1/2}\left( 1+\displaystyle\frac{k}{m}\right) \left( \displaystyle\frac{k}{%
m+k}\right) ^{\psi }}  \notag \\
=& O_{P}\left( m^{-\zeta }\right) =o_{P}\left( 1\right) .  \notag
\end{align}%
Similarly it can be shown that%
\begin{align}
& \max_{1\leq k<\infty }\frac{\left\vert R_{m,5}\left( k\right) \right\vert 
}{m^{1/2}\left( 1+\displaystyle\frac{k}{m}\right) \left( \displaystyle\frac{k%
}{m+k}\right) ^{\psi }}=O_{P}\left( m^{-\zeta }\right) =o_{P}\left( 1\right)
,  \label{sip-5} \\
& \max_{1\leq k<\infty }\frac{\left\vert R_{m,6}\left( k\right) \right\vert 
}{m^{1/2}\left( 1+\displaystyle\frac{k}{m}\right) \left( \displaystyle\frac{k%
}{m+k}\right) ^{\psi }}=O_{P}\left( m^{-\zeta +\eta }\right) =o_{P}\left(
1\right) ,  \label{sip-6}
\end{align}%
having chosen $\eta <\zeta $ in the last equation. Lemma \ref{lemma3} yields%
\begin{equation}
\max_{1\leq k<\infty }\frac{\left\vert \displaystyle\frac{k}{m}%
\sum_{i=2}^{m}\epsilon _{i,1}-\sigma _{1}\displaystyle\frac{k}{m}%
W_{m,1}\left( m\right) \right\vert }{m^{1/2}\left( 1+\displaystyle\frac{k}{m}%
\right) \left( \displaystyle\frac{k}{m+k}\right) ^{\psi }}=O_{P}\left(
m^{\zeta -1/2}\right) =o_{P}\left( 1\right) ,  \label{sip-7}
\end{equation}%
and 
\begin{equation}
\max_{1\leq k<\infty }\frac{\left\vert \displaystyle\sum_{i=m+1}^{m+k}%
\epsilon _{i,1}-\sigma _{1}\displaystyle W_{m,2}\left( k\right) \right\vert 
}{m^{1/2}\left( 1+\displaystyle\frac{k}{m}\right) \left( \displaystyle\frac{k%
}{m+k}\right) ^{\psi }}=O_{P}\left( m^{\zeta -1/2}\right) =o_{P}\left(
1\right) .  \label{sip-8}
\end{equation}%
similarly, by Lemma \ref{lemma2}%
\begin{align}
& \max_{1\leq k<\infty }\frac{\left\vert \displaystyle\sum_{i=m+1}^{m+k}%
\displaystyle\frac{\epsilon _{i,1}}{1+y_{i-1}^{2}}\right\vert +\left\vert %
\displaystyle\sum_{i=m+1}^{m+k}\displaystyle\frac{\epsilon _{i,2}y_{i-1}}{%
1+y_{i-1}^{2}}\right\vert }{m^{1/2}\left( 1+\displaystyle\frac{k}{m}\right)
\left( \displaystyle\frac{k}{m+k}\right) ^{\psi }}  \label{sip-9} \\
=& O_{P}\left( m^{-\zeta }\right) \max_{1\leq k<\infty }\frac{k^{1/2+\eta }}{%
m^{1/2}\left( 1+\displaystyle\frac{k}{m}\right) \left( \displaystyle\frac{k}{%
m+k}\right) ^{\psi }}  \notag \\
=& O_{P}\left( m^{1/2-\zeta }\right) \max_{1\leq k<\infty }\frac{k^{1/2+\eta
-\psi }}{\left( m+k\right) ^{1/2+\eta -\psi +1/2-\eta }}=O_{P}\left( m^{\eta
-\zeta }\right) =o_{P}\left( 1\right) ,  \notag
\end{align}%
having chosen $\eta <\zeta $. Putting all together, (\ref{str-approx-1}) now
follows for the case $E\log \left\vert \beta _{0}+\epsilon _{0,1}\right\vert
\geq 0$.

Let now 
\begin{equation}
\Gamma \left( t\right) =\left\{ 
\begin{array}{ll}
\mathscr{s}^{1/2}\left\vert W_{2}\left( t\right) -tW_{1}\left( 1\right)
\right\vert & \text{if }E\log \left\vert \beta _{0}+\epsilon
_{0,1}\right\vert <0\text{\ holds} \\ 
\sigma _{1}\left\vert W_{2}\left( t\right) -tW_{1}\left( 1\right) \right\vert
& \text{if }E\log \left\vert \beta _{0}+\epsilon _{0,1}\right\vert \geq 0%
\text{\ holds}%
\end{array}%
\right. ,  \label{gamma-t}
\end{equation}%
where $\left\{ W_{1}\left( k\right) ,1\leq k\leq m\right\} $ and $\left\{
W_{2}\left( k\right) ,1\leq k<\infty \right\} $ are two independent standard
Wiener processes. Using the fact that the distribution of $W_{m,1}\left(
\cdot \right) $ and $W_{m,2}\left( \cdot \right) $ does not depend on $m$,
and exploiting the scale transformation and the continuity of the Wiener
process%
\begin{align*}
& \max_{1\leq k<\infty }\frac{\Gamma _{m}\left( k\right) }{m^{1/2}\left( 1+%
\displaystyle\frac{k}{m}\right) \left( \displaystyle\frac{k}{m+k}\right)
^{\psi }} \\
& \overset{\mathcal{D}}{=}\max_{1\leq k<\infty }\frac{\Gamma \left( %
\displaystyle\frac{k}{m}\right) }{\left( 1+\displaystyle\frac{k}{m}\right)
\left( \frac{\displaystyle k/m}{\displaystyle1\displaystyle+\displaystyle k/m%
}\right) ^{\psi }} \\
& \overset{\mathcal{D}}{=}\sup_{0<t<\infty }\frac{\Gamma \left( t\right) }{%
\left( 1+t\right) \left( \displaystyle\frac{t}{1+t}\right) ^{\psi }}.
\end{align*}%
Computing the covariance kernel, it can be easily verified that%
\begin{equation}
\left\{ \frac{W_{2}\left( t\right) -tW_{1}\left( 1\right) }{1+t},t\geq
0\right\} \overset{\mathcal{D}}{=}\left\{ W\left( \frac{t}{1+t}\right)
,t\geq 0\right\} ,  \label{revesz}
\end{equation}%
where $W\left( \cdot \right) $ denotes a standard Wiener. Hence%
\begin{equation*}
\max_{1\leq k<\infty }\frac{\Gamma _{m}\left( k\right) }{m^{1/2}\left( 1+%
\displaystyle\frac{k}{m}\right) \left( \displaystyle\frac{k}{m+k}\right)
^{\psi }}\overset{\mathcal{D}}{=}\sup_{0<t<\infty }\frac{\mathscr{s}%
\left\vert W\left( \displaystyle\frac{t}{1+t}\right) \right\vert }{\left( %
\displaystyle\frac{t}{1+t}\right) ^{\psi }},
\end{equation*}%
whence Theorem \ref{th-1} follows.
\end{proof}

\begin{proof}[Proof of Theorem \protect\ref{th-2}]
Using exactly the same arguments as in the proof of Theorem \ref{th-1}, it
follows that%
\begin{equation}
\max_{1\leq k\leq m^{\ast }}\frac{\left\vert Z_{m}\left( k\right) -\Gamma
_{m}\left( k\right) \right\vert }{m^{1/2}\left( 1+\displaystyle\frac{k}{m}%
\right) \left( \displaystyle\frac{k}{m+k}\right) ^{\psi }}=o_{P}\left(
1\right) .  \label{str-approx-2}
\end{equation}%
Further, it is immediate to see that%
\begin{align*}
& \max_{1\leq k\leq m^{\ast }}\frac{\Gamma _{m}\left( k\right) }{%
m^{1/2}\left( 1+\displaystyle\frac{k}{m}\right) \left( \displaystyle\frac{k}{%
m+k}\right) ^{\psi }} \\
\overset{\mathcal{D}}{=}& \max_{1\leq k\leq m^{\ast }}\frac{\Gamma \left( %
\displaystyle\frac{k}{m}\right) }{\left( 1+\displaystyle\frac{k}{m}\right)
\left( \frac{\displaystyle k/m}{\displaystyle1\displaystyle+\displaystyle k/m%
}\right) ^{\psi }} \\
\overset{\mathcal{D}}{=}& \sup_{0<t<m_{0}}\frac{\Gamma \left( t\right) }{%
\left( 1+t\right) \left( \displaystyle\frac{t}{1+t}\right) ^{\psi }}.
\end{align*}%
Using (\ref{revesz}), it finally follows%
\begin{equation*}
\sup_{0<t<m_{0}}\frac{\Gamma \left( t\right) }{\left( 1+t\right) \left( %
\displaystyle\frac{t}{1+t}\right) ^{\psi }}\overset{\mathcal{D}}{=}%
\sup_{0<t<m_{0}}\frac{\mathscr{s}\left\vert W\left( \displaystyle\frac{t}{1+t%
}\right) \right\vert }{\left( \displaystyle\frac{t}{1+t}\right) ^{\psi }},
\end{equation*}%
whence the desired result.
\end{proof}

\begin{proof}[Proof of Theorem \protect\ref{de}]
We begin by noting that, repeating \textit{verbatim} the proof of (\ref%
{str-approx-1}) with $\psi =1/2$, it can be shown that%
\begin{equation}
\max_{1\leq k\leq m^{\ast }}\frac{\left\vert Z_{m}\left( k\right) -\Gamma
_{m}\left( k\right) \right\vert }{m^{1/2}\left( 1+\displaystyle\frac{k}{m}%
\right) \left( \displaystyle\frac{k}{m+k}\right) ^{1/2}}=O_{P}\left(
1\right) .  \label{str-approx-3}
\end{equation}%
Recall also that the distribution of $\Gamma _{m}\left( k\right) $ does not
depend on $m$, so that%
\begin{equation*}
\left\{ m^{-1/2}\Gamma _{m}\left( mt\right) ,t\geq 0\right\} \overset{%
\mathcal{D}}{=}\left\{ \Gamma \left( t\right) ,t\geq 0\right\} ,
\end{equation*}%
where $\Gamma \left( t\right) $ is defined in (\ref{gamma-t}); further
recall that, by (\ref{revesz})%
\begin{equation*}
\left\{ \frac{1}{\mathscr{s}}\frac{\Gamma \left( t\right) }{1+t},t\geq
0\right\} \overset{\mathcal{D}}{=}\left\{ \left\vert W\left( \frac{t}{1+t}%
\right) \right\vert ,t\geq 0\right\} .
\end{equation*}

We now report some results concerning $\Gamma _{m}\left( k\right) $. We
begin by studying%
\begin{align}
& \frac{1}{\mathscr{s}}\max_{1\leq k\leq \log m}\frac{\Gamma _{m}\left(
k\right) }{m^{1/2}\left( 1+\displaystyle\frac{k}{m}\right) \left( %
\displaystyle\frac{k}{m+k}\right) ^{1/2}}  \label{de-11} \\
\overset{\mathcal{D}}{=}& \frac{1}{\mathscr{s}}\max_{1/m\leq t\leq \left(
\log m\right) /m}\frac{\Gamma \left( t\right) }{\left( 1+t\right) \left( %
\displaystyle\frac{t}{1+t}\right) ^{1/2}}\overset{\mathcal{D}}{=}%
\max_{1/m\leq t\leq \left( \log m\right) /m}\frac{\left\vert W\left( %
\displaystyle\frac{t}{1+t}\right) \right\vert }{\left( \displaystyle\frac{t}{%
1+t}\right) ^{1/2}}  \notag \\
\overset{\mathcal{D}}{=}& \max_{1/\left( m+1\right) \leq u\leq \left( \log
m\right) /\left( m+\log m\right) }\frac{\left\vert W\left( u\right)
\right\vert }{u^{1/2}}\overset{\mathcal{D}}{=}\max_{1\leq v\leq \log m}\frac{%
\left\vert W\left( v\right) \right\vert }{v^{1/2}}  \notag \\
=& O_{P}\left( \sqrt{\log \log \log m}\right) ,  \notag
\end{align}%
where we have repeatedly used the scale transformation of the Wiener
process, and the Law of the Iterated Logarithm in the last passage. Similarly%
\begin{align}
& \frac{1}{\mathscr{s}}\max_{m/\log m\leq k\leq m^{\ast }}\frac{\Gamma
_{m}\left( k\right) }{m^{1/2}\left( 1+\displaystyle\frac{k}{m}\right) \left( %
\displaystyle\frac{k}{m+k}\right) ^{1/2}}  \label{de-22} \\
\overset{\mathcal{D}}{=}& \frac{1}{\mathscr{s}}\max_{1/\log m\leq t\leq
m^{\ast }/\log m}\frac{\Gamma \left( t\right) }{\left( 1+t\right) \left( %
\displaystyle\frac{t}{1+t}\right) ^{1/2}}\overset{\mathcal{D}}{=}%
\max_{1/\log m\leq t\leq m^{\ast }/\log m}\frac{\left\vert W\left( %
\displaystyle\frac{t}{1+t}\right) \right\vert }{\left( \displaystyle\frac{t}{%
1+t}\right) ^{1/2}}  \notag \\
\overset{\mathcal{D}}{=}& \max_{1/\left( \log m+1\right) \leq u\leq m^{\ast
}/\left( m+m^{\ast }\right) }\frac{\left\vert W\left( u\right) \right\vert }{%
u^{1/2}}\overset{\mathcal{D}}{=}\max_{1\leq v\leq m^{\ast }\left( \log
m+1\right) /\left( m+m^{\ast }\right) }\frac{\left\vert W\left( v\right)
\right\vert }{v^{1/2}}  \notag \\
=& O_{P}\left( \sqrt{\log \log \frac{m^{\ast }\left( \log m+1\right) }{%
m+m^{\ast }}}\right) =O_{P}\left( \sqrt{\log \log \log m}\right) .  \notag
\end{align}%
Finally, by the same token, note that%
\begin{align*}
& \frac{1}{\mathscr{s}}\max_{\log m\leq k\leq m/\log m}\frac{\Gamma
_{m}\left( k\right) }{m^{1/2}\left( 1+\displaystyle\frac{k}{m}\right) \left( %
\displaystyle\frac{k}{m+k}\right) ^{1/2}} \\
\overset{\mathcal{D}}{=}& \frac{1}{\mathscr{s}}\max_{\log m\leq k\leq m/\log
m}\frac{\Gamma _{m}\left( \displaystyle\frac{k}{m}\right) }{\left( 1+%
\displaystyle\frac{k}{m}\right) \left( \displaystyle\frac{k}{m+k}\right)
^{1/2}} \\
\overset{\mathcal{D}}{=}& \frac{1}{\mathscr{s}}\max_{\left( \log m\right)
/m\leq t\leq 1/\log m}\frac{\Gamma \left( t\right) }{\left( 1+t\right)
\left( \displaystyle\frac{t}{1+t}\right) ^{1/2}}\overset{\mathcal{D}}{=}%
\max_{\left( \log m\right) /m\leq t\leq 1/\log m}\frac{\left\vert W\left( %
\displaystyle\frac{t}{1+t}\right) \right\vert }{\left( \displaystyle\frac{t}{%
1+t}\right) ^{1/2}} \\
\overset{\mathcal{D}}{=}& \max_{\left( \log m\right) /\left( m+\log m\right)
\leq t\leq 1/\left( \log m+1\right) }\frac{\left\vert W\left( u\right)
\right\vert }{u^{1/2}}\overset{\mathcal{D}}{=}\max_{1\leq v\leq \left(
m+\log m\right) /\left( \left( 1+\log m\right) \log m\right) }\frac{%
\left\vert W\left( v\right) \right\vert }{v^{1/2}},
\end{align*}%
and using the Law of the Iterated Logarithm it follows that%
\begin{equation}
\lim_{m\rightarrow \infty }\frac{1}{\sqrt{2\log \log \displaystyle\frac{%
m+\log m}{\left( 1+\log m\right) \log m}}}\max_{1\leq v\leq \left( m+\log
m\right) /\left( \left( 1+\log m\right) \log m\right) }\frac{\left\vert
W\left( v\right) \right\vert }{v^{1/2}}=1\text{ \ \ a.s.}  \label{de-3}
\end{equation}%
Given that, as $m\rightarrow \infty $%
\begin{equation*}
\frac{\log \log \displaystyle\frac{m+\log m}{\left( 1+\log m\right) \log m}}{%
\log \log m}=1,
\end{equation*}%
it follows that the term in (\ref{de-3}) dominates, and therefore, putting
all together%
\begin{equation*}
\lim_{m\rightarrow \infty }P\left\{ \max_{1\leq k\leq m^{\ast }}\frac{%
Z_{m}\left( k\right) }{m^{1/2}\left( 1+\displaystyle\frac{k}{m}\right)
\left( \displaystyle\frac{k}{m+k}\right) ^{1/2}}=\max_{1\leq v\leq \left(
m+\log m\right) /\left( \left( 1+\log m\right) \log m\right) }\frac{%
\left\vert W\left( v\right) \right\vert }{v^{1/2}}\right\} =1.
\end{equation*}

Using the Darling-Erd\H{o}s theorem (\citealp{darling1956limit}), it follows
that%
\begin{align*}
& \lim_{m\rightarrow \infty }P\left\{ \gamma \left( \log \frac{m+\log m}{%
\left( 1+\log m\right) \log m}\right) \max_{1\leq v\leq \left( m+\log
m\right) /\left( \left( 1+\log m\right) \log m\right) }\frac{\left\vert
W\left( v\right) \right\vert }{v^{1/2}}\right. \\
\leq & \left. x+\delta \left( \log \frac{m+\log m}{\left( 1+\log m\right)
\log m}\right) \right\} =\exp \left( -\exp \left( -x\right) \right) .
\end{align*}%
The desired result follows upon noting that, by elementary arguments%
\begin{align*}
\left\vert \gamma \left( \frac{m+\log m}{\left( 1+\log m\right) \log m}%
\right) -\gamma \left( m\right) \right\vert & \rightarrow 0, \\
\left\vert \delta \left( \frac{m+\log m}{\left( 1+\log m\right) \log m}%
\right) -\delta \left( m\right) \right\vert & \rightarrow 0,
\end{align*}%
as $m\rightarrow \infty $.
\end{proof}

\begin{proof}[Proof of Theorem \protect\ref{monitor-short}]
We start with the proof of (\ref{weighted-bar}), considering the case $E\log
\left\vert \beta _{0}+\epsilon _{0,1}\right\vert <0$. As in the proof of
Theorem \ref{th-1}, it suffices to show that%
\begin{equation*}
\max_{1\leq k\leq m^{\ast }}\frac{Z_{m}^{\ast }\left( k\right) }{\overline{g}%
_{m,\psi }\left( k\right) }\overset{\mathcal{D}}{\rightarrow }\frac{1}{%
\overline{c}_{\alpha ,\psi }}\sup_{0\leq t\leq 1}\frac{\left\vert W\left(
t\right) \right\vert }{t^{\psi }},
\end{equation*}%
where%
\begin{align*}
Z_{m}^{\ast }\left( k\right) =& \left\vert \left( \sum_{i=2}^{m}\frac{%
\overline{y}_{i-1}^{2}}{1+\overline{y}_{i-1}^{2}}\right) ^{-1}\left(
\sum_{i=2}^{m}\frac{\epsilon _{i,1}\overline{y}_{i-1}^{2}}{1+\overline{y}%
_{i-1}^{2}}+\sum_{i=2}^{m}\frac{\epsilon _{i,2}\overline{y}_{i-1}}{1+%
\overline{y}_{i-1}^{2}}\right) \left( \sum_{i=m+1}^{m+k}\frac{\overline{y}%
_{i-1}^{2}}{1+\overline{y}_{i-1}^{2}}\right) \right. \\
& +\left. \sum_{i=m+1}^{m+k}\frac{\epsilon _{i,1}\overline{y}_{i-1}^{2}}{1+%
\overline{y}_{i-1}^{2}}+\sum_{i=m+1}^{m+k}\frac{\epsilon _{i,2}\overline{y}%
_{i-1}}{1+\overline{y}_{i-1}^{2}}\right\vert ,
\end{align*}%
where recall that $\overline{y}_{i}$\ is the stationary solution of (\ref%
{rca}). Using the proofs of Lemmas \ref{lemma1} and \ref{lemma2}, it follows
by routine calculations that%
\begin{equation*}
\left\vert \left( \sum_{i=2}^{m}\frac{\overline{y}_{i-1}^{2}}{1+\overline{y}%
_{i-1}^{2}}\right) ^{-1}-\frac{1}{ma_{3}}\right\vert =O_{P}\left( m\right) ,
\end{equation*}%
and 
\begin{equation*}
\left\vert \sum_{i=2}^{m}\frac{\epsilon _{i,1}\overline{y}_{i-1}^{2}}{1+%
\overline{y}_{i-1}^{2}}+\sum_{i=2}^{m}\frac{\epsilon _{i,2}\overline{y}_{i-1}%
}{1+\overline{y}_{i-1}^{2}}\right\vert =O_{P}\left( m^{1/2}\right) ,
\end{equation*}%
where $a_{3}$ is defined in (\ref{a3}). The approximations in \citet{aue2014}
yield%
\begin{equation*}
\left( m^{\ast }\right) ^{-1/2+\psi }\max_{1\leq k\leq m^{\ast }}\frac{1}{%
k^{\psi }}\left\vert \sum_{i=m+1}^{m+k}\left( \frac{\overline{y}_{i-1}^{2}}{%
1+\overline{y}_{i-1}^{2}}-a_{3}\right) \right\vert \overset{\mathcal{D}}{%
\rightarrow }a_{4}\sup_{0\leq t\leq 1}\frac{\left\vert W\left( t\right)
\right\vert }{t^{\psi }},
\end{equation*}%
where%
\begin{equation*}
a_{4}=\sum_{h=-\infty }^{\infty }E\left[ \left( \frac{\overline{y}_{0}^{2}}{%
1+\overline{y}_{0}^{2}}-a_{3}\right) \left( \frac{\overline{y}_{h}^{2}}{1+%
\overline{y}_{h}^{2}}-a_{3}\right) \right] .
\end{equation*}%
Thus we conclude%
\begin{align*}
& \left( m^{\ast }\right) ^{-1/2+\psi }\left( \sum_{i=2}^{m}\frac{\overline{y%
}_{i-1}^{2}}{1+\overline{y}_{i-1}^{2}}\right) ^{-1}\left\vert \sum_{i=2}^{m}%
\frac{\epsilon _{i,1}\overline{y}_{i-1}^{2}}{1+\overline{y}_{i-1}^{2}}%
+\sum_{i=2}^{m}\frac{\epsilon _{i,2}\overline{y}_{i-1}}{1+\overline{y}%
_{i-1}^{2}}\right\vert \max_{1\leq k\leq m^{\ast }}\frac{1}{k^{\psi }}%
\sum_{i=m+1}^{m+k}\frac{\overline{y}_{i-1}^{2}}{1+\overline{y}_{i-1}^{2}} \\
=& O_{P}\left( m^{-1/2}\right) \left( m^{\ast }\right) ^{-1/2+\psi
}\max_{1\leq k\leq m^{\ast }}\frac{1}{k^{\psi }}\sum_{i=m+1}^{m+k}\frac{%
\overline{y}_{i-1}^{2}}{1+\overline{y}_{i-1}^{2}} \\
=& O_{P}\left( m^{-1/2}\right) \left( m^{\ast }\right) ^{-1/2+\psi
}k^{1-\psi } \\
& +O_{P}\left( m^{-1/2}\right) \left( m^{\ast }\right) ^{-1/2+\psi
}\max_{1\leq k\leq m^{\ast }}\frac{1}{k^{\psi }}\sum_{i=m+1}^{m+k}\left( 
\frac{\overline{y}_{i-1}^{2}}{1+\overline{y}_{i-1}^{2}}-a_{3}\right) \\
=& O_{P}\left( \left( \frac{m^{\ast }}{m}\right) ^{1/2}+m^{-1/2}\right)
=o_{P}\left( 1\right) .
\end{align*}

Using \citet{aue2014}, we can define two independent Wiener processes $%
\left\{ W_{m,1}\left( x\right) ,x\geq 0\right\} $ and $\left\{ W_{m,2}\left(
x\right) ,x\geq 0\right\} $ such that%
\begin{align*}
& \max_{1\leq k\leq m^{\ast }}\frac{1}{\left( m^{\ast }\right) ^{1/2-\psi
}k^{\psi }}\left\vert \sum_{i=m+1}^{m+k}\frac{\epsilon _{i,1}\overline{y}%
_{i-1}^{2}}{1+\overline{y}_{i-1}^{2}}+\sum_{i=m+1}^{m+k}\frac{\epsilon _{i,2}%
\overline{y}_{i-1}}{1+\overline{y}_{i-1}^{2}}-\left( a_{1}^{1/2}\sigma
_{1}W_{m,1}\left( k\right) +a_{2}^{1/2}\sigma _{2}W_{m,2}\left( k\right)
\right) \right\vert \\
=& o_{P}\left( 1\right) .
\end{align*}%
By the scale transformation and the continuity of the Wiener process, it
follows that%
\begin{eqnarray*}
&&\max_{1\leq k\leq m^{\ast }}\frac{1}{\left( m^{\ast }\right) ^{1/2-\psi
}k^{\psi }}\left\vert a_{1}^{1/2}\sigma _{1}W_{m,1}\left( k\right)
+a_{2}^{1/2}\sigma _{2}W_{m,2}\left( k\right) \right\vert \\
&&\overset{\mathcal{D}}{=}\max_{1/m^{\ast }\leq k/m^{\ast }\leq 1}\left( 
\frac{k}{m^{\ast }}\right) ^{-\psi }\left\vert a_{1}^{1/2}\sigma
_{1}W_{m,1}\left( \frac{k}{m^{\ast }}\right) +a_{2}^{1/2}\sigma
_{2}W_{m,2}\left( \frac{k}{m^{\ast }}\right) \right\vert \\
&&\overset{\mathcal{D}}{\rightarrow }\sup_{0\leq u\leq 1}\frac{1}{t^{\psi }}%
\left\vert a_{1}^{1/2}\sigma _{1}W_{1}\left( t\right) +a_{2}^{1/2}\sigma
_{2}W_{2}\left( t\right) \right\vert ,
\end{eqnarray*}%
where $W_{1}$ and $W_{2}$ are independent Wiener processes. Since%
\begin{equation*}
\left\{ a_{1}^{1/2}\sigma _{1}W_{1}\left( t\right) +a_{2}^{1/2}\sigma
_{2}W_{2}\left( t\right) ,t\geq 0\right\} \overset{\mathcal{D}}{=}\left\{ %
\mathscr{s}W\left( t\right) ,t\geq 0\right\} ,
\end{equation*}%
where $\left\{ W\left( t\right) ,t\geq 0\right\} $ is a Wiener process, the
result follows.

Next we consider the case $\psi =1/2$, and we establish a Darling-Erd\H{o}s
limiting law for $Z_{m}^{\ast }\left( k\right) $. Firstly note that%
\begin{align*}
&\max_{1\leq k\leq m^{\ast }}\frac{1}{k^{1/2}}\left( \sum_{i=2}^{m}\frac{%
\overline{y}_{i-1}^{2}}{1+\overline{y}_{i-1}^{2}}\right) ^{-1}\left(
\sum_{i=2}^{m}\frac{\epsilon _{i,1}\overline{y}_{i-1}^{2}}{1+\overline{y}%
_{i-1}^{2}}+\sum_{i=2}^{m}\frac{\epsilon _{i,2}\overline{y}_{i-1}}{1+%
\overline{y}_{i-1}^{2}}\right) \sum_{i=m+1}^{m+k}\frac{\overline{y}_{i-1}^{2}%
}{1+\overline{y}_{i-1}^{2}} \\
=&O_{P}\left( \left( \frac{m^{\ast }}{m}\right) ^{1/2}\right) ,
\end{align*}%
\begin{equation*}
\frac{1}{\left( \log \log m^{\ast }\right) ^{1/2}}\max_{1\leq k\leq m^{\ast
}}\frac{1}{k^{1/2}}\left\vert \sum_{i=m+1}^{m+k}\frac{\epsilon _{i,1}%
\overline{y}_{i-1}^{2}}{1+\overline{y}_{i-1}^{2}}+\sum_{i=m+1}^{m+k}\frac{%
\epsilon _{i,2}\overline{y}_{i-1}}{1+\overline{y}_{i-1}^{2}}\right\vert
=O_{P}\left( 1\right) ,
\end{equation*}%
and

\begin{equation*}
\frac{1}{\left( \log \log m^{\ast }\right) ^{1/2}}\max_{\log m^{\ast }\leq
k\leq m^{\ast }/\log m^{\ast }}\frac{1}{k^{1/2}}\left\vert \sum_{i=m+1}^{m+k}%
\frac{\epsilon _{i,1}\overline{y}_{i-1}^{2}}{1+\overline{y}_{i-1}^{2}}%
+\sum_{i=m+1}^{m+k}\frac{\epsilon _{i,2}\overline{y}_{i-1}}{1+\overline{y}%
_{i-1}^{2}}\right\vert \overset{\mathcal{P}}{\rightarrow }c>0.
\end{equation*}%
The above entails that we need to consider the maximum of $Z_{m}^{\ast
}\left( k\right) $ over $\log m^{\ast }\leq k\leq m^{\ast }/\log m^{\ast }$.
We have%
\begin{align*}
&\max_{\log m^{\ast }\leq k\leq m^{\ast }/\log m^{\ast }}\frac{1}{k^{1/2}}%
\left( \sum_{i=2}^{m}\frac{\overline{y}_{i-1}^{2}}{1+\overline{y}_{i-1}^{2}}%
\right) ^{-1}\left( \sum_{i=2}^{m}\frac{\epsilon _{i,1}\overline{y}_{i-1}^{2}%
}{1+\overline{y}_{i-1}^{2}}+\sum_{i=2}^{m}\frac{\epsilon _{i,2}\overline{y}%
_{i-1}}{1+\overline{y}_{i-1}^{2}}\right) \sum_{i=m+1}^{m+k}\frac{\overline{y}%
_{i-1}^{2}}{1+\overline{y}_{i-1}^{2}} \\
=&O_{P}\left( \left( \frac{m^{\ast }}{m}\right) ^{1/2}\left( \log m^{\ast
}\right) ^{-1/2}\right) .
\end{align*}%
By the Law of the Iterated Logarithm, it follows that%
\begin{equation*}
\max_{1\leq k\leq \log m^{\ast }}\frac{1}{k^{1/2}}\left\vert
\sum_{i=m+1}^{m+k}\frac{\epsilon _{i,1}\overline{y}_{i-1}^{2}}{1+\overline{y}%
_{i-1}^{2}}+\sum_{i=m+1}^{m+k}\frac{\epsilon _{i,2}\overline{y}_{i-1}}{1+%
\overline{y}_{i-1}^{2}}\right\vert =O_{P}\left( \left( \log \log \log
m^{\ast }\right) ^{1/2}\right) .
\end{equation*}%
Hence we need to show that, for all $-\infty <x<\infty $%
\begin{align*}
&\lim_{m\rightarrow \infty }P\left\{ \gamma \left( \log m^{\ast }\right) 
\frac{1}{\mathscr{s}}\max_{1\leq k\leq m^{\ast }}\frac{1}{k^{1/2}}\left\vert
\sum_{i=m+1}^{m+k}\frac{\epsilon _{i,1}\overline{y}_{i-1}^{2}}{1+\overline{y}%
_{i-1}^{2}}+\sum_{i=m+1}^{m+k}\frac{\epsilon _{i,2}\overline{y}_{i-1}}{1+%
\overline{y}_{i-1}^{2}}\right\vert \leq x+\delta \left( \log m^{\ast
}\right) \right\} \\
=&\exp \left( -\exp \left( -x\right) \right) ,
\end{align*}%
which follows from the same logic as the proof of Theorem \ref{de}.

We now turn to the case $E\log \left\vert \beta _{0}+\epsilon
_{0,1}\right\vert \geq 0$. Following the proof of Lemma \ref{lemma1}, it can
be shown that%
\begin{equation}
\left\vert \left( \sum_{i=2}^{m}\frac{y_{i-1}^{2}}{1+y_{i-1}^{2}}\right)
^{-1}-\frac{1}{m}\right\vert =O_{P}\left( m^{-3/2}\right) ,  \label{sh1}
\end{equation}%
\begin{equation}
\left\vert \left( \sum_{i=2}^{m}\frac{\epsilon _{i,1}y_{i-1}^{2}}{%
1+y_{i-1}^{2}}+\sum_{i=2}^{m}\frac{\epsilon _{i,2}y_{i-1}}{1+y_{i-1}^{2}}%
\right) -\sum_{i=2}^{m}\epsilon _{i,1}\right\vert =O_{P}\left( 1\right) ,
\label{sh2}
\end{equation}%
\begin{equation}
\max_{1\leq k<\infty }\left\vert \sum_{i=m+1}^{m+k}\frac{y_{i-1}^{2}}{%
1+y_{i-1}^{2}}-k\right\vert =O_{P}\left( c^{m}\right) ,  \label{sh3}
\end{equation}%
and%
\begin{equation}
\max_{1\leq k<\infty }\left\vert \left( \sum_{i=m+1}^{m+k}\frac{\epsilon
_{i,1}y_{i-1}^{2}}{1+y_{i-1}^{2}}+\sum_{i=m+1}^{m+k}\frac{\epsilon
_{i,2}y_{i-1}}{1+y_{i-1}^{2}}\right) -\sum_{i=m+1}^{m+k}\epsilon
_{i,1}\right\vert =O_{P}\left( c^{m}\right) ,  \label{sh4}
\end{equation}
\ \ with some $0<c<1$. Thus we get%
\begin{align*}
&\left( \sum_{i=2}^{m}\frac{\overline{y}_{i-1}^{2}}{1+\overline{y}_{i-1}^{2}}%
\right) ^{-1}\left\vert \sum_{i=2}^{m}\frac{\epsilon _{i,1}\overline{y}%
_{i-1}^{2}}{1+\overline{y}_{i-1}^{2}}+\sum_{i=2}^{m}\frac{\epsilon _{i,2}%
\overline{y}_{i-1}}{1+\overline{y}_{i-1}^{2}}\right\vert \left( m^{\ast
}\right) ^{-1/2+\psi }\max_{1\leq k\leq m^{\ast }}\frac{1}{k^{\psi }}%
\sum_{i=m+1}^{m+k}\frac{\overline{y}_{i-1}^{2}}{1+\overline{y}_{i-1}^{2}} \\
=&O_{P}\left( m^{-1/2}\right) \left( m^{\ast }\right) ^{-1/2+\psi
}\max_{1\leq k\leq m^{\ast }}\frac{1}{k^{\psi }}\left( k+c^{m}\right) \\
=&O_{P}\left( \left( \frac{m^{\ast }}{m}\right) ^{1/2}\right) =o_{P}\left(
1\right) .
\end{align*}%
Our assumptions entail that%
\begin{equation*}
\left( m^{\ast }\right) ^{-1/2+\psi }\max_{1\leq k\leq m^{\ast }}\frac{1}{%
k^{\psi }}\left\vert \sum_{i=m+1}^{m+k}\epsilon _{i,1}\right\vert \overset{%
\mathcal{D}}{\rightarrow }\sigma _{1}\sup_{0\leq t\leq 1}\frac{\left\vert
W\left( t\right) \right\vert }{t^{\psi }},
\end{equation*}%
where $\left\{ W\left( t\right) ,t\geq 0\right\} $ is a Wiener process; this
completes the proof of (\ref{weighted-bar}). Considering the case $\psi =1/2$%
, we can follows the same arguments as above, using (\ref{sh1})-(\ref{sh4}).
We obtain%
\begin{equation*}
\left\vert \max_{1\leq k<m^{\ast }}\frac{\left\vert Z_{m}^{\ast }\left(
k\right) \right\vert }{k^{1/2}}-\max_{\log m^{\ast }\leq k<m^{\ast }/\log
m^{\ast }}\frac{1}{k^{1/2}}\left\vert \sum_{i=m+1}^{m+k}\epsilon
_{i,1}\right\vert \right\vert =o_{P}\left( \left( \log \log m^{\ast }\right)
^{-1/2}\right) .
\end{equation*}%
Hence we need to prove only that%
\begin{equation*}
\lim_{m\rightarrow \infty }P\left\{ \gamma \left( \log m^{\ast }\right) 
\frac{1}{\sigma _{1}}\max_{\log m^{\ast }\leq k<m^{\ast }/\log m^{\ast }}%
\frac{1}{k^{1/2}}\left\vert \sum_{i=m+1}^{m+k}\epsilon _{i,1}\right\vert
\leq x+\delta \left( \log m^{\ast }\right) \right\} =\exp \left( -\exp
\left( -x\right) \right) ,
\end{equation*}%
which is shown in \citet{csorgo1997}.
\end{proof}

\begin{proof}[Proof of Theorem \protect\ref{thgombay}]
Since the distributions of $W_{m,1}\left( \cdot \right) $\ and $%
W_{m,2}\left( \cdot \right) $\ do not depend on $m$, Corollary \ref{gombay2}
entails that, for all $-\infty <x<\infty $%
\begin{equation*}
P\left\{ \max_{h_{m^{\ast }}\leq k\leq m^{\ast }}\frac{Z_{m}\left( k\right) 
}{g_{m,0.5}\left( k\right) }\geq x\right\} =P\left\{ \max_{h_{m^{\ast }}\leq
k\leq m^{\ast }}\frac{\left\vert W_{1}\left( k\right) -\displaystyle\frac{k}{%
m}W_{2}\left( m\right) \right\vert }{m^{1/2}\left( 1+\displaystyle\frac{k}{m}%
\right) \left( \displaystyle\frac{k}{m+k}\right) ^{1/2}}\geq x\right\}
+o\left( 1\right) ,
\end{equation*}%
as $m^{\ast }\rightarrow \infty $. In turn, this entails that we can
construct critical values $\widehat{c}_{\alpha ,0.5}$ based on%
\begin{equation*}
P\left\{ \max_{h_{m^{\ast }}\leq k\leq m^{\ast }}\frac{\left\vert
W_{1}\left( k\right) -\displaystyle\frac{k}{m}W_{2}\left( m\right)
\right\vert }{m^{1/2}\left( 1+\displaystyle\frac{k}{m}\right) \left( %
\displaystyle\frac{k}{m+k}\right) ^{1/2}}\geq \widehat{c}_{\alpha
,0.5}\right\} =\alpha .
\end{equation*}%
Using the scale transformation for Wiener processes%
\begin{align*}
& \max_{h_{m^{\ast }}\leq k\leq m^{\ast }}\frac{\left\vert W_{1}\left(
k\right) -\displaystyle\frac{k}{m}W_{2}\left( m\right) \right\vert }{%
m^{1/2}\left( 1+\displaystyle\frac{k}{m}\right) \left( \displaystyle\frac{k}{%
m+k}\right) ^{1/2}} \\
\overset{\mathcal{D}}{=}& \max_{h_{m^{\ast }}/m^{\ast }\leq k/m^{\ast }\leq
1}\frac{\left( 1+\displaystyle\frac{k}{m}\right) \left\vert W\left( %
\displaystyle\frac{k/m}{1+k/m}\right) \right\vert }{\left( 1+\displaystyle%
\frac{k}{m}\right) \left( \displaystyle\frac{k/m}{1+k/m}\right) ^{1/2}} \\
\overset{\mathcal{D}}{=}& \max_{h_{m^{\ast }}/m^{\ast }\leq \tau \leq 1}%
\frac{\left\vert W\left( \displaystyle\frac{\tau }{1+\tau }\right)
\right\vert }{\left( \displaystyle\frac{\tau }{1+\tau }\right) ^{1/2}} \\
\overset{\mathcal{D}}{=}& \max_{h_{m^{\ast }}/\left( m^{\ast }+h_{m^{\ast
}}\right) \leq u\leq 1/2}\frac{\left\vert W\left( u\right) \right\vert }{%
u^{1/2}}\overset{\mathcal{D}}{=}\max_{1\leq v\leq \phi _{m}}\frac{\left\vert
W\left( v\right) \right\vert }{v^{1/2}} \\
\overset{\mathcal{D}}{=}& \max_{1\leq v\leq \exp \left( \log \phi
_{m}\right) }\frac{\left\vert W\left( v\right) \right\vert }{v^{1/2}},
\end{align*}%
where $W\left( \cdot \right) $\ is a standard Wiener process and $\phi
_{m}=\left( m^{\ast }+h_{m^{\ast }}\right) /\left( 2h_{m^{\ast }}\right) $.
Using equation (18) in \citet{vostrikova}, we have 
\begin{equation}
P\left\{ \max_{1\leq v\leq \exp \left( \log \phi _{m}\right) }\frac{%
\left\vert W\left( v\right) \right\vert }{v^{1/2}}\geq \widehat{c}_{\alpha
,0.5}\right\} \simeq \frac{\widehat{c}_{\alpha ,0.5}\exp \left( -\frac{1}{2}%
\widehat{c}_{\alpha ,0.5}^{2}\right) }{\left( 2\pi \right) ^{1/2}}\left(
\log \phi _{m}+\frac{4-\log \phi _{m}}{\widehat{c}_{\alpha ,0.5}^{2}}%
+O\left( \widehat{c}_{\alpha ,0.5}^{-4}\right) \right) ,  \notag
\end{equation}%
as $\widehat{c}_{\alpha ,0.5}\rightarrow \infty $. This proves the claim.
\end{proof}

\begin{proof}[Proof of Theorem \protect\ref{th-kirch}]
The proof follows, with minor modifications, the proofs of the results
above, so we only outline its main passages. We begin with part \textit{(i)}%
, and consider the case $E\log \left\vert \beta _{0}+\epsilon
_{0,1}\right\vert <0$. Lemma \ref{lemma1} entails that we need to prove only
that%
\begin{align}
& \max_{1\leq k<\infty }\frac{1}{g_{m,\psi }\left( k\right) }\max_{1\leq
\ell \leq k}\left\vert \sum_{i=m+\ell }^{m+k}\left( \frac{\epsilon _{i,1}%
\overline{y}_{i-1}^{2}}{1+\overline{y}_{i-1}^{2}}+\frac{\epsilon _{i,2}%
\overline{y}_{i-1}}{1+\overline{y}_{i-1}^{2}}\right) \right.  \label{eq12} \\
& \left. -\frac{k-\ell }{m}\sum_{i=2}^{m}\left( \frac{\epsilon _{i,1}%
\overline{y}_{i-1}^{2}}{1+\overline{y}_{i-1}^{2}}+\frac{\epsilon _{i,2}%
\overline{y}_{i-1}}{1+\overline{y}_{i-1}^{2}}\right) \right\vert  \notag \\
& \overset{\mathcal{D}}{\rightarrow }\sup_{0<x<\infty }\frac{\sup_{0\leq
t\leq x}\left\vert \left( W_{2}\left( x\right) -W_{2}\left( t\right) \right)
-\left( x-t\right) W_{1}\left( 1\right) \right\vert }{\left( 1+x\right)
\left( x/\left( 1+x\right) \right) ^{\psi }},  \notag
\end{align}%
where recall that $\overline{y}_{i}$ is the stationary solution of (\ref{rca}%
). The approximations in Lemma \ref{lemma3} imply%
\begin{align*}
& \max_{1\leq k<\infty }\frac{1}{g_{m,\psi }\left( k\right) }\max_{1\leq
\ell \leq k}\left\vert \sum_{i=m+\ell }^{m+k}\left( \frac{\epsilon _{i,1}%
\overline{y}_{i-1}^{2}}{1+\overline{y}_{i-1}^{2}}+\frac{\epsilon _{i,2}%
\overline{y}_{i-1}}{1+\overline{y}_{i-1}^{2}}\right) -\frac{k-\ell }{m}%
\sum_{i=2}^{m}\left( \frac{\epsilon _{i,1}\overline{y}_{i-1}^{2}}{1+%
\overline{y}_{i-1}^{2}}+\frac{\epsilon _{i,2}\overline{y}_{i-1}}{1+\overline{%
y}_{i-1}^{2}}\right) \right. \\
& \left. -\mathscr{s}\left( W_{m,2}\left( k\right) -W_{m,2}\left( \ell
\right) -\frac{k-\ell }{m}W_{m,1}\left( m\right) \right) \right\vert \\
=& o_{P}\left( 1\right) .
\end{align*}%
By the scale transformation and the continuity of the Wiener process we have%
\begin{align*}
& \max_{1\leq k<\infty }\left( \frac{k+m}{k}\right) ^{\psi }\frac{m}{m+k}%
\max_{1\leq \ell \leq k}\left\vert \left( W_{m,2}\left( k\right)
-W_{m,2}\left( \ell \right) -\frac{k-\ell }{m}W_{m,1}\left( m\right) \right)
\right\vert \\
& \overset{\mathcal{D}}{\rightarrow }\sup_{0<x<\infty }\left( \frac{1+x}{x}%
\right) ^{\psi }\frac{1}{1+x}\sup_{0\leq t\leq x}\left\vert \left(
W_{2}\left( x\right) -W_{2}\left( t\right) \right) -\left( x-t\right)
W_{1}\left( 1\right) \right\vert ,
\end{align*}%
whence the desired result.

Next we consider the nonstationary case. Applying Lemma \ref{lemma2}\textit{%
(ii)}, we can replace (\ref{eq12}) with the proof of%
\begin{align*}
&\max_{1\leq k<\infty }\frac{1}{g_{m,\psi }\left( k\right) }\max_{1\leq \ell
\leq k}\left\vert \sum_{i=m+\ell }^{m+k}\frac{\epsilon _{i,1}y_{i-1}^{2}}{%
1+y_{i-1}^{2}}-\frac{k-\ell }{m}\sum_{i=2}^{m}\frac{\epsilon
_{i,1}y_{i-1}^{2}}{1+y_{i-1}^{2}}\right\vert \\
&\overset{\mathcal{D}}{\rightarrow }\sup_{0<x<\infty }\left( \frac{1+x}{x}%
\right) ^{\psi }\frac{1}{1+x}\sup_{0\leq t\leq x}\left\vert \left(
W_{2}\left( x\right) -W_{2}\left( t\right) \right) -\left( x-t\right)
W_{1}\left( 1\right) \right\vert ,
\end{align*}%
which follows from the approximations in Lemma \ref{lemma3}\textit{(ii)}.

We conclude our proof by showing that 
\begin{equation}
P\left\{ \sup_{0<x<\infty }\left( \frac{1+x}{x}\right) ^{\psi }\frac{1}{1+x}%
\sup_{0\leq t\leq x}\left\vert \left( W_{2}\left( x\right) -W_{2}\left(
t\right) \right) -\left( x-t\right) W_{1}\left( 1\right) \right\vert <\infty
\right\} =1.  \label{finite}
\end{equation}%
When $0<x\leq 1$, using Theorem 2.1 in \citet{garsia},\footnote{%
See also Lemma 4.1 in \citealp{csorgo1993}.} it follows that there exists a
random variable $\zeta $ such that $E\left\vert \zeta \right\vert
^{p}<\infty $ for all $p\geq 1$ and%
\begin{align*}
&\left( \frac{1+x}{x}\right) ^{\psi }\frac{1}{1+x}\sup_{0\leq t\leq
x}\left\vert \left( W_{2}\left( x\right) -W_{2}\left( t\right) \right)
\right\vert \\
\leq &\left\vert \zeta \right\vert \left( \frac{1+x}{x}\right) ^{\psi }\frac{%
1}{1+x}\sup_{0\leq t\leq x}\left\vert \left( \left( x-t\right) \log \left( 
\frac{1}{x-t}\right) \right) ^{1/2}\right\vert
\end{align*}%
a.s., whence%
\begin{align}
&\left( \frac{1+x}{x}\right) ^{\psi }\frac{1}{1+x}\sup_{0\leq t\leq
x}\left\vert \left( W_{2}\left( x\right) -W_{2}\left( t\right) \right)
-\left( x-t\right) W_{1}\left( 1\right) \right\vert  \label{garsia} \\
\leq &\left\vert \zeta \right\vert \left( \frac{1+x}{x}\right) ^{\psi }\frac{%
1}{1+x}\sup_{0\leq t\leq x}\left\vert \left( \left( x-t\right) \log \left( 
\frac{1}{x-t}\right) \right) ^{1/2}\right\vert +W_{1}\left( 1\right) \left( 
\frac{1+x}{x}\right) ^{\psi }\frac{1}{1+x}\sup_{0\leq t\leq x}\left\vert
x-t\right\vert .  \notag
\end{align}%
It is now easy to see that%
\begin{equation*}
\left\vert \zeta \right\vert \left( \frac{1+x}{x}\right) ^{\psi }\frac{1}{1+x%
}\sup_{0\leq t\leq x}\left\vert \left( \left( x-t\right) \log \left( \frac{1%
}{x-t}\right) \right) ^{1/2}\right\vert =O_{P}\left( 1\right) .
\end{equation*}%
As far as the second term in (\ref{garsia}) is concerned, it immediately
follows that%
\begin{equation*}
\left\vert W_{1}\left( 1\right) \right\vert \left( \frac{1+x}{x}\right)
^{\psi }\frac{1}{1+x}\sup_{0\leq t\leq x}\left\vert x-t\right\vert
=O_{P}\left( 1\right) \left( \frac{1+x}{x}\right) ^{\psi -1}=O_{P}\left(
1\right) .
\end{equation*}%
Also, by the Law of the Iterated Logarithm, it follows that%
\begin{align*}
&\sup_{1\leq x<\infty }\left( \frac{1+x}{x}\right) ^{\psi }\frac{1}{1+x}%
\sup_{0\leq t\leq x}\left\vert \left( W_{2}\left( x\right) -W_{2}\left(
t\right) \right) -\left( x-t\right) W_{1}\left( 1\right) \right\vert \\
\leq &\sup_{1\leq x<\infty }\frac{1}{x}2\sup_{0\leq t\leq x}\left\vert
W_{2}\left( x\right) \right\vert +\left\vert W_{1}\left( 1\right)
\right\vert <\infty ,
\end{align*}%
a.s., whence (\ref{finite}) follows.

Finally, as far as parts \textit{(ii)} and \textit{(iii)} of the theorem are
concerned, the proofs are based on the observation that%
\begin{equation*}
\max_{1\leq k\leq m^{\ast }}\frac{1}{g_{m,\psi }^{\ast }\left( k\right) }%
\max_{1\leq \ell \leq k}\frac{k-\ell }{m}\left\vert \sum_{i=2}^{m}\frac{%
\epsilon _{i,1}y_{i-1}^{2}}{1+y_{i-1}^{2}}+\frac{\epsilon _{i,2}y_{i-1}}{%
1+y_{i-1}^{2}}\right\vert =o_{P}\left( 1\right) ,
\end{equation*}%
and the same using $\overline{g}_{m,\psi }\left( k\right) $; hence, the
arguments above can be repeated with minor modifications.
\end{proof}

\begin{proof}[Proof of Corollary \protect\ref{sig-est}]
The result is shown e.g. in the proof of Theorem 3.4 in \citet{HT2016}, or
Corollary 3.1 in \citet{horvath2022changepoint}.
\end{proof}

\medskip

In order to prove the next set of results, we will make use of the following
decomposition, valid under (\ref{alternative})%
\begin{align}
Z_{m}\left( k\right) =& \left\vert \sum_{i=m+1}^{m+k^{\ast }}\frac{\epsilon
_{i,1}y_{i-1}^{2}+\epsilon _{i,2}y_{i-1}}{1+y_{i-1}^{2}}-\left( \widehat{%
\beta }_{m}-\beta _{0}\right) \sum_{i=m+1}^{m+k^{\ast }}\frac{y_{i-1}^{2}}{%
1+y_{i-1}^{2}}\right.  \label{z-delay} \\
& +\left( \beta _{A}-\beta _{0}\right) \sum_{i=m+k^{\ast }+1}^{m+k}\frac{%
y_{i-1}^{2}}{1+y_{i-1}^{2}}+\sum_{i=m+k^{\ast }+1}^{m+k}\frac{\epsilon
_{i,1}y_{i-1}^{2}+\epsilon _{i,2}y_{i-1}}{1+y_{i-1}^{2}}  \notag \\
& -\left. \left( \widehat{\beta }_{m}-\beta _{0}\right) \sum_{i=m+k^{\ast
}+1}^{m+k}\frac{y_{i-1}^{2}}{1+y_{i-1}^{2}}\right\vert ,  \notag
\end{align}%
with the convention that $\sum_{a}^{b}=0$ whenever $b<a$.

\begin{proof}[Proof of Theorem \protect\ref{power}]
Based on the results above, the theorem follows from standard arguments -
see e.g. Theorem 4.1 in \citet{horvath2022changepoint}.
\end{proof}

\begin{proof}[Proof of Theorem \protect\ref{th-x-1}]
We begin by decomposing the detector $Z_{m}^{X}\left( k\right) $ as%
\begin{equation}
Z_{m}^{X}\left( k\right) =\left\vert \left( \mathbf{b}_{0}-\widehat{\mathbf{b%
}}_{m}\right) ^{\intercal }\sum_{i=m+1}^{m+k}\left( y_{i-1},\mathbf{x}%
_{i}^{\intercal }\right) \frac{y_{i-1}}{1+y_{i-1}^{2}}+\sum_{i=m+1}^{m+k}%
\frac{\left( \epsilon _{i,1}y_{i-1}+\epsilon _{i,2}\right) y_{i-1}}{%
1+y_{i-1}^{2}}\right\vert .  \label{dec-p}
\end{equation}%
Henceforth, the proof is very similar to the proof of Theorem \ref{th-1},
and we only report the main passages where the two proofs differ.

Consider first the case $E\log \left\vert \beta _{0}+\epsilon
_{0,1}\right\vert <0$; using Lemmas \ref{beta-x}-\ref{beta-x-3} instead of
Lemmas \ref{lemma1}-\ref{lemma3}, it can be shown along the same lines as in
the proof of Theorem \ref{th-1} that%
\begin{equation*}
\max_{1\leq k<\infty }\frac{Z_{m}^{X}\left( k\right) }{m^{1/2}\left( 1+%
\displaystyle\frac{k}{\mathscr{s}_{x}^{2}m}\right) \left( \displaystyle\frac{%
k}{k+\mathscr{s}_{x}^{2}m}\right) ^{\psi }}\overset{\mathcal{D}}{\rightarrow 
}\sup_{0<t<\infty }\frac{\left\vert -\mathscr{s}_{x,1}tW_{1}\left( 1\right) +%
\mathscr{s}_{x,2}W_{2}\left( t\right) \right\vert }{\left( 1+\displaystyle%
\frac{t}{\mathscr{s}_{x}^{2}}\right) \left( \displaystyle\frac{t}{t+%
\mathscr{s}_{x}^{2}}\right) ^{\psi }},
\end{equation*}%
where $\left\{ W_{1}\left( t\right) ,t\geq 0\right\} $ and $\left\{
W_{2}\left( t\right) ,t\geq 0\right\} $\ are two independent standard Wiener
processes. It is not hard to see that%
\begin{equation*}
\sup_{0<t<\infty }\frac{\left\vert -\mathscr{s}_{x,1}tW_{1}\left( 1\right) +%
\mathscr{s}_{x,2}W_{2}\left( t\right) \right\vert }{\left( 1+\displaystyle%
\frac{t}{\mathscr{s}_{x}^{2}}\right) \left( \displaystyle\frac{t}{t+%
\mathscr{s}_{x}^{2}}\right) ^{\psi }}\overset{\mathcal{D}}{=}%
\sup_{0<u<\infty }\frac{\left\vert -\mathscr{s}_{x,1}\mathscr{s}%
_{x}^{2}uW_{1}\left( 1\right) +\mathscr{s}_{x,2}W_{2}\left( \mathscr{s}%
_{x}^{2}u\right) \right\vert }{\left( 1+u\right) \left( \displaystyle\frac{u%
}{1+u}\right) ^{\psi }},
\end{equation*}%
and that, for $u\leq v$%
\begin{align*}
& E\left( \left( -\mathscr{s}_{x,1}\mathscr{s}_{x}^{2}uW_{1}\left( 1\right) +%
\mathscr{s}_{x,2}W_{2}\left( \mathscr{s}_{x}^{2}u\right) \right) \left( -%
\mathscr{s}_{x,1}\mathscr{s}_{x}^{2}vW_{1}\left( 1\right) +\mathscr{s}%
_{x,2}W_{2}\left( \mathscr{s}_{x}^{2}v\right) \right) \right) \\
=& \mathscr{s}_{x,1}^{2}\left( \mathscr{s}_{x}^{2}\right) ^{2}uv+\mathscr{s}%
_{x,2}^{2}\mathscr{s}_{x}^{2}u=\frac{\mathscr{s}_{x,2}^{4}}{\mathscr{s}%
_{x,1}^{2}}\left( uv+u\right) .
\end{align*}%
Hence it follows that%
\begin{equation*}
\left\{ \frac{-\mathscr{s}_{x,1}\mathscr{s}_{x}^{2}uW_{1}\left( 1\right) +%
\mathscr{s}_{x,2}W_{2}\left( \mathscr{s}_{x}^{2}u\right) }{1+u},u\geq
0\right\} \overset{\mathcal{D}}{=}\left\{ \frac{\mathscr{s}_{x,2}^{2}}{%
\mathscr{s}_{x,1}}W\left( u\right) ,u\geq 0\right\} ,
\end{equation*}%
where $\left\{ W\left( u\right) ,u\geq 0\right\} $\ is an independent
standard Wiener process. This completes the proof when $E\log \left\vert
\beta _{0}+\epsilon _{0,1}\right\vert <0$.

When $E\log \left\vert \beta _{0}+\epsilon _{0,1}\right\vert >0$, consider (%
\ref{dec-p}) again. Seeing as%
\begin{equation}
\left\Vert \sum_{i=m+1}^{\infty }\frac{\mathbf{x}_{i}y_{i-1}}{1+y_{i-1}^{2}}%
\right\Vert =O_{P}\left( 1\right) ,  \label{decline}
\end{equation}%
it is easy to see that%
\begin{align*}
& \left\vert \left( \mathbf{b}_{0}-\widehat{\mathbf{b}}_{m}\right)
^{\intercal }\left( \sum_{i=m+1}^{m+k}\left( y_{i-1},\mathbf{x}%
_{i}^{\intercal }\right) \frac{y_{i-1}}{1+y_{i-1}^{2}}-\sum_{i=m+1}^{m+k}%
\left( y_{i-1},\mathbf{0}_{p}^{\intercal }\right) \frac{y_{i-1}}{%
1+y_{i-1}^{2}}\right) \right\vert \\
\leq & \left\Vert \mathbf{b}_{0}-\widehat{\mathbf{b}}_{m}\right\Vert
\left\Vert \sum_{i=m+1}^{m+k}\left( y_{i-1},\mathbf{x}_{i}^{\intercal
}\right) \frac{y_{i-1}}{1+y_{i-1}^{2}}-\sum_{i=m+1}^{m+k}\left( y_{i-1},%
\mathbf{0}_{p}^{\intercal }\right) \frac{y_{i-1}}{1+y_{i-1}^{2}}\right\Vert
\\
=& O_{P}\left( m^{-1/2}\right) ,
\end{align*}%
where $\mathbf{0}_{p}$ is a $p$-dimensional vector of zeros. Similarly,
using (\ref{decline}) and standard algebra, it follows that 
\begin{equation*}
\left\Vert \left( \mathbf{b}_{0}-\widehat{\mathbf{b}}_{m}\right) ^{\intercal
}-\left( \left( \sum_{i=2}^{m}\frac{y_{i-1}^{2}}{1+y_{i-1}^{2}}\right)
^{-1}\left( \sum_{i=2}^{m}\frac{\epsilon _{i,1}y_{i-1}^{2}}{1+y_{i-1}^{2}}%
\right) ,\mathbf{0}_{p}^{\intercal }\right) \right\Vert =O_{P}\left(
m^{-1/2}\right) ,
\end{equation*}%
and by (\ref{lemma2-4})%
\begin{equation*}
\max_{1\leq k<\infty }\frac{1}{k}\left\vert \left( \sum_{i=2}^{m}\frac{%
y_{i-1}^{2}}{1+y_{i-1}^{2}}\right) ^{-1}\left( \sum_{i=2}^{m}\frac{\epsilon
_{i,1}y_{i-1}^{2}}{1+y_{i-1}^{2}}\right) \left( \sum_{i=m+1}^{m+k}\frac{%
y_{i-1}^{2}}{1+y_{i-1}^{2}}-k\right) \right\vert =O_{P}\left( m^{-1/2-\zeta
}\right) ,
\end{equation*}%
for some $\zeta >0$, and similarly to (\ref{lemma1-3})-(\ref{lemma1-4})%
\begin{equation*}
\left\vert \left( \sum_{i=2}^{m}\frac{y_{i-1}^{2}}{1+y_{i-1}^{2}}\right)
^{-1}\left( \sum_{i=2}^{m}\frac{\epsilon _{i,1}y_{i-1}^{2}}{1+y_{i-1}^{2}}%
\right) -\frac{1}{m}\sum_{i=2}^{m}\epsilon _{i,1}\right\vert =O_{P}\left(
m^{1/2-\zeta }\right) .
\end{equation*}%
Thus, putting all together%
\begin{equation*}
\max_{1\leq k<\infty }\frac{1}{k}\left\vert \left( \mathbf{b}_{0}-\widehat{%
\mathbf{b}}_{m}\right) ^{\intercal }\sum_{i=m+1}^{m+k}\left( y_{i-1},\mathbf{%
x}_{i}^{\intercal }\right) \frac{y_{i-1}}{1+y_{i-1}^{2}}-\frac{k}{m}%
\sum_{i=2}^{m}\epsilon _{i,1}\right\vert =O_{P}\left( m^{-1/2-\zeta }\right)
.
\end{equation*}%
Finally, similarly to (\ref{lemma2-5}) it can be shown 
\begin{equation*}
\max_{1\leq k<\infty }k^{-1/2-\eta }\left\vert \sum_{i=m+1}^{m+k}\frac{%
\epsilon _{i,1}y_{i-1}^{2}}{1+y_{i-1}^{2}}-\sum_{i=m+1}^{m+k}\epsilon
_{i,1}\right\vert =O_{P}\left( m^{-\zeta }\right) ,
\end{equation*}%
for some $\zeta >0$ and for all $\eta >0$. Putting all together and
recalling that, when $E\log \left\vert \beta _{0}+\epsilon _{0,1}\right\vert
>0$, $\mathscr{s}_{x}^{2}=1$, it follows that%
\begin{equation*}
\max_{1\leq k<\infty }\frac{\left\vert Z_{m}^{X}\left( k\right) -\left\vert %
\displaystyle\sum_{i=m+1}^{m+k}\epsilon _{i,1}-\displaystyle\frac{k}{m}%
\displaystyle\sum_{i=2}^{m}\epsilon _{i,1}\right\vert \right\vert }{%
m^{1/2}\left( 1+\displaystyle\frac{k}{m}\right) \left( \displaystyle\frac{k}{%
k+m}\right) ^{\psi }}=o_{P}\left( 1\right) .
\end{equation*}%
Hereafter, the final result follows by making appeal to the proof of Theorem %
\ref{th-1}.
\end{proof}

\begin{proof}[Proof of Theorem \protect\ref{th-x-2}]
The proof follows immediately from the approximations derived above, and
from the proof of Theorem \ref{de}.
\end{proof}

\begin{proof}[Proof of Theorem \protect\ref{kirch-x}]
The proof follows by the same arguments as the proof of Theorem \ref%
{th-kirch}, and we therefore omit it to save space.
\end{proof}

\begin{proof}[Proof of Theorem \protect\ref{alt-covariates}]
The proof follows from the same arguments as in the case with no covariates.
\end{proof}

\end{document}